\documentclass{article}

\usepackage{bm}
\usepackage{amsmath}
\usepackage{cases}
\usepackage{float}
\usepackage{multicol}
\usepackage{subcaption}
\usepackage[ruled,vlined,algosection]{algorithm2e}
\usepackage{caption}
\usepackage{mathtools}
\usepackage{amsthm}
\usepackage{booktabs}
\usepackage{natbib}
\usepackage{url}
\usepackage{hyperref}
\usepackage{enumitem}
\usepackage{xcolor}
\usepackage{graphicx}
\usepackage{fullpage}
\usepackage{amssymb}
\usepackage{fancyhdr}
\usepackage{authblk}

\numberwithin{table}{section}
\numberwithin{figure}{section}
\numberwithin{equation}{section}

\usepackage[running]{lineno}

\newcommand*\patchAmsMathEnvironmentForLineno[1]{%
\expandafter\let\csname old#1\expandafter\endcsname\csname #1\endcsname
\expandafter\let\csname oldend#1\expandafter\endcsname\csname end#1\endcsname
\renewenvironment{#1}%
{\linenomath\csname old#1\endcsname}%
{\csname oldend#1\endcsname\endlinenomath}}%
\newcommand*\patchBothAmsMathEnvironmentsForLineno[1]{%
\patchAmsMathEnvironmentForLineno{#1}%
\patchAmsMathEnvironmentForLineno{#1*}}%
\AtBeginDocument{%
\patchBothAmsMathEnvironmentsForLineno{equation}%
\patchBothAmsMathEnvironmentsForLineno{align}%
\patchBothAmsMathEnvironmentsForLineno{flalign}%
\patchBothAmsMathEnvironmentsForLineno{alignat}%
\patchBothAmsMathEnvironmentsForLineno{gather}%
\patchBothAmsMathEnvironmentsForLineno{multline}%
}

\newtheorem{remark}{Remark}[section]
\newtheorem{assumption}{Assumption}[section]

\newtheorem{theorem}{Theorem}[section]
\newtheorem{proposition}{Proposition}[section]
\newtheorem{corollary}{Corollary}[section]
\newtheorem{lemma}{Lemma}[section]
\newtheorem{definition}{Definition}[section]
\makeatletter
\newcommand{\vast}{\bBigg@{4}}
\newcommand{\Vast}{\bBigg@{5}}
\makeatother

\SetKwFor{For}{for (}{) $\lbrace$}{$\rbrace$}

\begin{document}


\title{{Neural Network Approach to\\Portfolio Optimization with Leverage Constraints:\\a Case Study on High Inflation Investment}}

\author{Chendi Ni}
\affil{Cheriton School of Computer Science, University of Waterloo, Waterloo, N2L 3G1, Canada,
\texttt{chendi.ni@uwaterloo.ca}
}
\author{Yuying Li}
\affil{Cheriton School of Computer Science, University of Waterloo, Waterloo, N2L 3G1, Canada,
\texttt{yuying@uwaterloo.ca}
}
\author{Peter Forsyth}
\affil{Cheriton School of Computer Science, University of Waterloo, Waterloo, N2L 3G1, Canada,
\texttt{paforsyt@uwaterloo.ca}
}

\maketitle
\begin{abstract}
Motivated by the current global high inflation scenario, we aim to discover a dynamic multi-period allocation strategy to optimally outperform a passive benchmark while adhering to a bounded leverage limit. To this end, we formulate an optimal control problem to outperform a benchmark portfolio throughout the investment horizon. Assuming the asset prices follow the jump-diffusion model during high inflation periods, we first establish a closed-form solution for the optimal strategy that outperforms a passive strategy under the cumulative quadratic tracking difference (CD) objective, assuming continuous trading and no bankruptcy. 
{To obtain strategies under the bounded leverage constraint among other realistic constraints, we then propose a novel leverage-feasible neural network (LFNN) to represent control, which converts the original constrained optimization problem into an unconstrained optimization problem that is computationally feasible with standard optimization methods.} 
We establish mathematically that the LFNN approximation can yield a solution that is arbitrarily close to the solution of the original optimal control problem with bounded leverage. We further apply the LFNN approach to a four-asset investment scenario with bootstrap resampled asset returns from the filtered high inflation regime data. 
The LFNN strategy is shown to consistently outperform the passive benchmark strategy by about 200 bps 
(median annualized return), with a greater than 90\% probability of 
outperforming the benchmark at the end of the investment horizon.

\vspace{5pt}


\end{abstract}

\section{Introduction}\label{sec:intro}
Since the global outbreak of COVID-19 in March 2020, there has been a significant increase in worldwide inflation. Specifically, from May 2021 to February 2023, the 12-month change in the CPI index in the U.S. has not dropped below 5\% \citep{statistics2023consumer}. Prior to the pandemic, the U.S. economy experienced nearly four decades of low inflation. The abrupt shift from a long-term low-inflation environment to a high-inflation environment has created substantial uncertainty and volatility in the financial markets. In 2022, the technology-heavy NASDAQ stock index recorded a yearly return of -33.10\% \citep{nasdaq2023}.

Equally concerning is the uncertainty around the duration of this round of high inflation. Some believe that the geopolitical tensions and the COVID-19 pandemic will overturn the trend of globalization and lead to global supply chain restructuring \citep{javorcik2020global} which may result in a higher cost of production in the foreseeable future. Moreover, \citet{ball2022understanding} suggests that the future inflation rate may remain high if the unemployment rate remains low. 

In this article, we aim to answer the following question: 
{with the goal of outperforming a passive benchmark, how should an active investor optimize the portfolio during high inflation?
It is important to note that} we do not attempt to make predictions about future inflation conditions. {Instead, we approach the problem by formulating a multi-period optimal control problem that considers bounded leverage constraints and specific investment criteria.} This optimal control problem requires the specification of an appropriate objective function, realistic constraints, as well as stochastic models for returns of traded assets during high inflation regimes. {Given the complexities and challenges associated, it is crucial to develop an efficient method capable of computing optimal solutions, accommodating flexible data sources, handling high-dimensional cases, and dealing with complex constraints. In this paper, we propose a framework to address these challenges.}

In Section \ref{sec:active_allocation}, we assume that the real (inflation-adjusted) asset returns during a high-inflation regime follow stochastic processes and treat allocation decisions as the control of a dynamic system. 
Specifically, we formulate an optimal control problem to outperform a fixed-mix benchmark portfolio consistently throughout the investment horizon by minimizing a cumulative quadratic tracking difference (CD) objective. There is a large amount of extant literature on closed-form solutions for beating a stochastic benchmark under synthetic market assumptions \citep{browne_1999_a,browne_2000,tepla_2001,basak_2006,davis_2008, Lim_2010,Oderda_2015,Alekseev_2016,alaradi_2018}. In these articles, the common objective function involves a log-utility function, e.g. the log wealth ratio. Under the log wealth ratio formulation, it is often hard to accommodate a fixed stream of cash injections, which is a common characteristic of open-ended funds. \citet{forsyth2022too} consider a scenario where a fixed amount of cash injections is allowed and provides a closed-form solution under a cumulative quadratic tracking difference (CD) objective given the assumption that the stock price follows a double exponential jump-diffusion model and the bond price is deterministic. 
Since the assumption that the bond index price is stochastic and has jumps is more reasonable under a high-inflation scenario, we develop a closed-form solution under the case that both the stock index and the bond index follow jump-diffusion models. 

The closed-form solution is derived, unfortunately, under unrealistic assumptions such as continuous rebalancing, infinite leverage, and continued trading when insolvent. A discrete-time multi-period asset allocation problem is generally solved using a dynamic programming (DP) based approach, which converts a multi-step optimization problem into multiple single-step optimization problems. However, \citet{van2023beating} point out that dynamic programming-based approaches require the evaluation of a high-dimensional performance criterion to obtain the optimal control which is comparatively low-dimensional. This means that solving the discrete-time problem numerically using dynamic programming-based techniques (for example numerical solutions to the corresponding PIDE \citep{wang2010numerical}, or reinforcement learning (RL) techniques \citep{dixon2020machine,park2020intelligent,lucarelli2020deep,gao2020application}) are inefficient and are computationally prone to known issues such as error amplification over recursions \citep{wang2020statistical}.

Acknowledging these limitations, in Section \ref{sec:nn_model}, we propose to use a single neural network model to approximate the optimal control and solve the original optimal control problem directly via a single standard finite-dimensional optimization. 
This direct approximation of the control exploits the lower dimensionality of the optimal control and bypasses the problem of solving high-dimensional conditional expectations associated with DP methods. We 
note that the idea of using a neural network to directly approximate the control process is also used 
in \citet{han2016deep,BuehlerGononEtAl2018,tsang2020deep, reppen2022deep}, in which they propose a stacked neural network approach that includes individual sub-networks for each rebalancing step. In contrast, we propose a single shallow neural network that includes time as an input feature, and thus avoids the need to have multiple sub-networks for each rebalancing step and greatly reduces the computational and modeling complexity. Furthermore, using time as a feature in the neural network approximation function is consistent with the observation that (under assumptions) the optimal control is naturally a continuous function of time, which we discuss in detail in Section \ref{sec:closed-form}. 

{The idea of using a single neural network to approximate controls has also been explored in previous studies such as \citet{li2019data} and \citet{ni2022optimal}. These studies focus on portfolio optimization problems with long-only constraints. The neural network architecture proposed in these studies transforms the constrained portfolio optimization problems into unconstrained optimization problems, making them computationally easier to solve. However, these existing neural network architectures do not address the bounded leverage constraint, which limits the total long exposure in the portfolio. The limited literature on portfolio optimization with bounded leverage is likely due to the added complexity arising from the combination of long and short positions in the portfolio. 

A significant contribution of this article is the introduction of a novel leverage-feasible neural network (LFNN) model, which converts the leverage-constrained optimization problem into an unconstrained optimization problem. This model enables the incorporation of the bounded leverage constraint into the portfolio optimization framework. Additionally, in Section \ref{sec:justify_LFNN}, we provide a mathematical proof that, under reasonable assumptions, the solution of the unconstrained optimization problem obtained using the LFNN model can approximate the optimal control of the original problem arbitrarily well. This mathematical justification validates the effectiveness and validity of the LFNN approach.

In Section \ref{sec:numerical_experiments}, we present a case study on active portfolio optimization in a high-inflation regime. To identify historical high-inflation periods, we employ a simple filtering method. Subsequently, we use bootstrap resampling to generate training and testing data sets, which consist of price paths for four assets: the equal-weighted and cap-weighted stock indexes, as well as the 30-day and 10-year U.S. treasury indexes.

Using the leverage-feasible neural network (LFNN) model and the cumulative quadratic tracking shortfall (CS) objective, we derive a leverage-constrained strategy for portfolio optimization. Our results demonstrate that the LFNN model produces a strategy that consistently outperforms the fixed-mix benchmark. Specifically, the strategy achieves a median (annualized) internal rate of return (IRR) that is more than 2\% higher than the benchmark. Moreover, there is a probability of over 90\% that the strategy will yield a higher terminal wealth compared to the benchmark.

These findings highlight the efficacy of the LFNN model in optimizing portfolios under high-inflation conditions. By incorporating the bounded leverage constraint and utilizing the CS objective, our approach enables investors to achieve superior performance and mitigate risks in a high-inflation environment.
}

Our contributions are summarized below:
\begin{enumerate}[label=(\roman*)]
   
    \item 
    To gain intuition about the behavior of the optimal controls, we derive the closed-form solution under a jump-diffusion asset price model and other typical assumptions (such as continuous rebalancing) for a two-asset case. The closed-form solution provides important insights into the properties of the optimal control as well as meaningful interpretations of the neural network models that approximate the controls.
    \item We propose to represent the control directly by a neural network representation so that the stochastic optimal control problem can be solved numerically under realistic constraints such as discrete rebalancing and limited leverage. Particularly, we propose the novel leverage-feasible neural network (LFNN) model to convert the original complex leverage-constrained optimization problem into an unconstrained optimization problem that can be solved easily by standard optimization methods. 
    \item We prove that, with a suitable choice of the hyperparameter of the LFNN model, the solution of the parameterized unconstrained optimization problem can approximate the optimal control arbitrarily well. This provides a mathematical justification for the validity of the LFNN approach. This is further supported by the numerical results that the performance of the LFNN model matches the clipped form of the closed-form solution on simulated data.

    \item {In the case study on active portfolio optimization in high-inflation,} we apply the neural network method to bootstrap resampled asset returns with four underlying assets, including the equal-weighted/cap-weighted stock indexes, and the 30-day/10-year treasury bond indexes. The dynamic strategy from the learned LFNN model outperforms the fixed-mix benchmark strategy consistently throughout the investment horizon, with a 2\% higher median (annualized) internal rate of return (IRR), and more than 90\% probability of achieving a higher terminal wealth. Furthermore, the learned allocation strategy suggests that the equal-weighted stock index and short-term bonds are preferable investment assets during high-inflation regimes.
\end{enumerate}

\section{Outperform dynamic benchmark under bounded leverage}\label{sec:active_allocation}

\subsection{Sovereign wealth funds and benchmark targets}

Instead of taking a passive approach, some of the largest sovereign wealth funds often adopt an 
active management philosophy and use passive portfolios as the benchmark to evaluate the efficiency 
of active management. For example, the Canadian Pension Plan (CPP) uses a base reference portfolio of 
85\% global equity and 15\% Canadian government bonds \citep{cpp2022}. Another example is the 
Government Pension Fund Global of Norway (also known as the oil fund) managed by Norges Bank 
Investment Management (NBIM), which uses a benchmark index consisting of 70\% equity index and 30\% bond 
index.\footnote{The Ministry of Finance of Norway sets the allocation fraction between the equity 
index and the bond index. It gradually raised the weight for equities from 60\% to 70\% from 2015-2018.} 
The benchmark equity index is constructed based on the market capitalization for equities 
in the countries included in the benchmark. The benchmark index for bonds specifies a 
defined allocation between government bonds and corporate bonds, with a 
weight of 70 percent to government bonds and 30 percent to corporate bonds \citep{norges2022}. 

However, the excess return that these well-known sovereign wealth funds have 
achieved over their respective passive benchmark portfolios cannot be described as impressive. 
In the 2022 fiscal year report, CPP claims to have beaten the base reference portfolio by an 
annualized 80 bps after fees over the past 5 years \citep{cpp2022}. 
On the other hand, NBIM reports a mere average of 27 bps of annual excess 
return over the benchmark over the last decade (see Table \ref{tb:nbim_return}). 
It is worth noting that these behemoth funds achieve seemingly meager results 
by hiring thousands of highly paid investment professionals and spending billions 
of dollars on day-to-day operations. For example, the CPP 2021 annual 
report \citep{cpp2021} lists personnel costs as CAD 938 million, for 1,936 employees, 
which translates to average costs of about CAD 500,000 per employee-year.

\begin{table}[htb]
\begin{tabular}{c|c|c|c|c|c|c|c|c|c|c|c}
\hline
Year                 & 2012 & 2013 & 2014  & 2015 & 2016 & 2017 & 2018  & 2019 & 2020 & 2021 & Average \\ \hline
Excess return (\%) & 0.21 & 0.99 & -0.77 & 0.45 & 0.15 & 0.70 & -0.30 & 0.23 & 0.27 & 0.74 & 0.27    \\ \hline
\end{tabular}
\caption{Norges Bank Investment Management, relative return to benchmark portfolio}\label{tb:nbim_return}
\end{table} 

The stark contrast between the enormous spending of sovereign wealth 
funds and the meager outperformance of the funds relative to the passive benchmark 
portfolios is probably provocative to taxpayers and pensioners who invest their 
hard-earned money in the funds. Equally concerning is the potential of a long, persistent 
inflation regime and the funds' ability to consistently beat the benchmark portfolio 
in such times. After all, both the CPP Investments and NBIM were established in the 
late 1990s, a decade after the last long inflation period ended in the mid-1980s. 

These concerns prompt us to ask the following question: in a presumed persistent high-inflation environment, 
can a fund manager find a simple asset allocation strategy that consistently beats the 
benchmark passive portfolios by a reasonable margin (preferably without spending billions of dollars in personnel costs)? 

\subsection{Mathematical formulation}\label{sec:formulation}
In this section, we mathematically formulate the problem of outperforming a benchmark. Let $[t_0(=0),T]$ denote the 
investment horizon, and let $W(t)$ denote the wealth (value) of the portfolio actively 
managed by the manager at time $t\in[t_0,T]$. We refer to the actively managed portfolio as the ``active portfolio''. 
Furthermore, let $\hat{W}(t)$ denote the wealth of the benchmark portfolio at time $t\in[t_0,T]$.
To ensure a fair assessment of the relative performance of the two portfolios, we assume both portfolios start with an equal initial wealth amount $w_0>0$, i.e., $W(t_0)=\hat{W}(t_0)=w_0>0$.

Technically, the admissible sets of underlying assets  for the active and 
passive portfolio need not be identical.
However, for simplicity, we assume that both the active portfolio and the benchmark 
portfolio can allocate wealth to the same set of $N_a$ assets.

{
Let vector $\boldsymbol{S}(t) = (S_i(t):i=1,\cdots,N_a)^\top\in\mathbb{R}^{N_a}$ denote the asset prices of the $N_a$ underlying assets at time $t\in[t_0,T]$. In addition, let vectors $\boldsymbol{p}(t)= (p_i(t):i=1,\cdots,N_a)^\top\in\mathbb{R}^{N_a}$ and $\boldsymbol{\hat{p}}(t)= (\hat{p}_i(t):i=1,\cdots,N_a)^\top\in\mathbb{R}^{N_a}$ 
denote the fraction of wealth allocated to the $N_a$ underlying assets at time $t\in[t_0,T]$, 
respectively, for the active portfolio and the benchmark portfolio. 
From a control theory perspective, the allocation vector $\boldsymbol{p}$ can be regarded as the control of the system, as it determines how the wealth of the 
active portfolio evolves over time. 
We will seek to find the optimal
feedback control. In other words, the closed-loop controls (allocation decisions) are assumed to also depend on the value 
of the state variables (e.g. portfolio wealth). 
Therefore, we consider the control $\boldsymbol{p}$ to be a function of time as well as the relevant state variables. 
In addition, the benchmark portfolio allocation $\boldsymbol{\hat{p}}$ can be regarded as a known function of
its state variables and time as well.

Mathematically, $\boldsymbol{p}(\boldsymbol{X}(t))= (p_i(\boldsymbol{X}(t)):i=1,\cdots,N_a)^\top\in\mathbb{R}^{N_a}$ and $\boldsymbol{\hat{p}}(\hat{\boldsymbol{X}}(t))= (\hat{p}_i(\hat{\boldsymbol{X}}(t)):i=1,\cdots,N_a)^\top\in\mathbb{R}^{N_a}$, where ${\boldsymbol{X}(t)}\in\mathcal{X}\subseteq\mathbb{R}^{N_x}$ and $\hat{\boldsymbol{X}}(t)\in\hat{\mathcal{X}}\subseteq\mathbb{R}^{N_{\hat{x}}}$ are the state variables taken into account by the active portfolio and the benchmark portfolio respectively. Here we include $t$ in 
$\boldsymbol{X}(t)$ and $\hat{\boldsymbol{X}}(t)$ for notational simplicity. In this article, we consider the particular problem of outperforming a passive portfolio, in which $\boldsymbol{X}(t)=\big(t,W(t),\hat{W}(t)\big)^\top$. 
}

We assume that the active portfolio and the benchmark portfolio follow the same rebalancing schedule denoted by $\mathcal{T}\subseteq[t_0,T]$. In the case of discrete rebalancing, $\mathcal{T}\subset[t_0,T]$ is a discrete set. In the case of continuous rebalancing, $\mathcal{T}=[t_0,T]$, i.e., rebalancing happens continuously throughout the entire investment horizon. 

Additionally, we assume both portfolios follow the same deterministic sequence of cash injections, defined by the set $\mathcal{C}=\{c(t),\;t\in\mathcal{T}_c\}$, where $\mathcal{T}_c\subseteq[t_0,T]$ is the schedule of the cash injections. 
When $\mathcal{T}_c$ is a discrete injection schedule, $c(t)$ is the amount of cash injection at $t$. 
In the case of continuous cash injections, i.e., $\mathcal{T}=[t_0,T]$, $c(t)$ is the rate of cash injection at $t$, 
i.e., the total cash injection amount during $[t,t+dt]$ is $c(t)dt$, where 
$dt$ is an infinitesimal time interval. For simplicity, we assume that $\mathcal{T}_c=\mathcal{T}$, 
so that the cash injections schedule is the same as the rebalancing schedule. At $t\in\mathcal{T}$, $W(t)$ and $\hat{W}(t)$ always denote the wealth after the cash injection (assuming there is a cash injection event happening at $t$).

The active and benchmark strategies, respectively, are defined as the sequence of the allocation fractions following the rebalancing schedule. Mathematically, the active and benchmark strategies are defined by sets

\begin{equation}
    \mathcal{P}=\{\boldsymbol{p}(\boldsymbol{X}(t)),\;t\in\mathcal{T}\},\quad\text{and}\quad\hat{\mathcal{P}}=\{\boldsymbol{\hat{p}}(\hat{\boldsymbol{X}}(t)),\;t\in\mathcal{T}\}.
\end{equation}

Denote $\mathcal{A}$ as the set of admissible strategies, which reflects the investment constraints on the controls.  We assume that admissibility can vary with state and let \{$\mathcal{X}_i$: $i=1,\cdots,k$\} be a partition of $\mathcal{X}$ (the state variable space), i.e.
\begin{equation}
    \begin{cases}
    \bigcup_{i=1}^k \mathcal{X}_i=\mathcal{X},\\
    \mathcal{X}_i\bigcap \mathcal{X}_j=\varnothing,\forall 1\leq i<j\leq k,
    \end{cases}
\end{equation}
and \{$\mathcal{Z}_i\subseteq\mathbb{R}^{N_a}$: $i=1,\cdots,k$\} be the corresponding value sets of feasible controls such that any feasible control $\boldsymbol{p}$ satisfies
\begin{equation}
    \boldsymbol{p}(\boldsymbol{x})\in\mathcal{Z}_i,\forall \boldsymbol{x}\in\mathcal{X}_i,\;\forall i\in\{1,\cdots,k\}.
\end{equation}
We say that strategy $\mathcal{P}$ is an admissible strategy, i.e., $\mathcal{P}\in\mathcal{A}$, 
if and only if 
\begin{equation}
 \mathcal{P}=\Big\{\boldsymbol{p}(\boldsymbol{X}(t)),\; t\in\mathcal{T}\;\Big|\;\boldsymbol{p}(\boldsymbol{X}(t))\in\mathcal{Z}_i,\;\text{if }\boldsymbol{X}(t)\in\mathcal{X}_i\Big\}
\end{equation}

Consider a discrete rebalancing schedule $\mathcal{T}=\{t_j,\;j=0,\cdots,N\}$ with $N$ rebalancing events, where $t_0<t_1<\cdots<t_{N}=T$.\footnote{Technically, at $t=t_0$, the manager makes the initial asset allocation, rather than a ``rebalancing'' of the portfolio. However, despite the different purposes, a rebalancing of the portfolio is simply a new allocation of the portfolio wealth. Therefore, for notational simplicity, we include $t_0$ in the rebalancing schedule.} Then, the wealth evolution of the active portfolio and the benchmark portfolio can be described by the equations
\begin{equation}
\begin{cases}
    W(t_{j+1})=\Big(\sum\limits_{i=1}\limits^{N_a}p_i(\boldsymbol{X}(t_j))\cdot\frac{S_i(t_{j+1})-S_i(t_j)}{S_i(t_j)}\Big)W(t_j)+c(t_{j+1}),\;j=0,\cdots,N-1,
    \\
    \hat{W}(t_{j+1})=\Big(\sum\limits_{i=1}\limits^{N_a}\hat{p}_i(\hat{\boldsymbol{X}}(t_j))\cdot\frac{S_i(t_{j+1})-S_i(t_j)}{S_i(t_j)}\Big)\hat{W}(t_j)+c(t_{j+1}),\;j=0,\cdots,N-1.
\end{cases}\label{dynamics_disc}
\end{equation}

In the continuous rebalancing case, $\mathcal{T}=[t_0,T]$. Let $dS_i(t)$ denote the instantaneous change in price for asset $i$, $i\in[1,\cdots,N_a]$.\footnote{For illustration purposes, here we assume $S_i(t),i\in[1,\cdots,N_a]$ follow standard diffusion processes, i.e., no jumps. We will discuss the case with jumps in detail in Section \ref{sec:closed-form}.} Then, at $t\in\mathcal{T}=[t_0,T]$, the wealth dynamics of the active portfolio and the benchmark portfolio, following their respective strategies $\mathcal{P}$ and $\hat{\mathcal{P}}$, can be described by the equations
\begin{equation}
\begin{cases}
    dW(t)=\Big(\sum\limits_{i=1}\limits^{N_a}p_i(\boldsymbol{X}(t))\cdot\frac{dS_i(t)}{S_i(t)}\Big)W(t_j)+c(t)dt,
    \\
    d\hat{W}(t)=\Big(\sum\limits_{i=1}\limits^{N_a}\hat{p}_i(\hat{\boldsymbol{X}}(t))\cdot\frac{dS_i(t)}{S_i(t)}\Big)\hat{W}(t_j)+c(t)dt.
\end{cases}\label{dynamics_cont}
\end{equation}

Let sets $\mathcal{W}_{\mathcal{P}}=\{W(t), t\in\mathcal{T}\}$ and $\hat{\mathcal{W}}_{\hat{\mathcal{P}}}=\{\hat{W}(t), t\in\mathcal{T}\}$ denote the wealth trajectories of the active portfolio and the benchmark portfolio following their respective investment strategies $\mathcal{P}$ and $\hat{\mathcal{P}}$. Let $F(\mathcal{W}_{\mathcal{P}},\hat{\mathcal{W}}_{\hat{\mathcal{P}}})\in\mathbb{R}$ denote an investment metric that measures the performances of the active and benchmark strategies, based on their respective wealth trajectories.

In this article, we assume that the asset prices $\boldsymbol{S}(t)\in\mathbb{R}^{N_a}$ are stochastic. Then, the wealth trajectories $\mathcal{W}_{\mathcal{P}}$ and $\hat{\mathcal{W}}_{\hat{\mathcal{P}}}$ are also stochastic, as well as the performance metric $F(\mathcal{W}_{\mathcal{P}},\hat{\mathcal{W}}_{\hat{\mathcal{P}}})$, which measures the relative performance of the active strategy with respect to the benchmark strategy. Therefore, when investment managers target to optimize an investment metric, the evaluation is often on the expectation of the random metric.

Let $\mathbb{E}_{\mathcal{P}}^{(t_0,w_0)}[F(\mathcal{W}_{\mathcal{P}},\hat{\mathcal{W}}_{\hat{\mathcal{P}}})]$ denote the expectation of the value of the performance metric $F$, with respect to a given initial wealth $w_0=W(0)=\hat{W}(0)$ at time $t_0=0$, following an admissible investment strategies $\mathcal{P}\in\mathcal{A}$, and the benchmark investment strategy $\hat{\mathcal{P}}$. Since the benchmark strategy is often pre-determined and known, we keep the benchmark strategy $\hat{\mathcal{P}}$ implicit in this notation for simplicity. Subsequently we use $\mathbb{E}_{\mathcal{P}}^{(t_0,w_0)}[F(\mathcal{W}_{\mathcal{P}},\hat{\mathcal{W}}_{\hat{\mathcal{P}}})]$, the expectation of a desired performance metric, as the {\it (investment) objective function} and solve

\begin{equation}\text{(Optimization problem):}
    \quad\inf_{\mathcal{P}\in\mathcal{A}}\mathbb{E}_{\mathcal{P}}^{(t_0,w_0)}\big[F(\mathcal{W}_{\mathcal{P}},\hat{\mathcal{W}}_{\hat{\mathcal{P}}})\big].\label{generic_optimal_problem}
\end{equation}

\subsection{Choice of investment objective}\label{sec:obj_choice}

The first step to designing a proper outperforming investment objective is to clarify the definition of {\it beating the benchmark}. In the context of measuring the performance of the portfolio against the benchmark, a common metric is the tracking error, which measures the volatility of the {difference in returns}, i.e.,
\begin{equation}
 \text{Tracking error}=stdev(R-\hat{R}),\label{metric:track_err}   
\end{equation}
where $R$ denotes the return of the active portfolio, and $\hat{R}$ denotes the return of the benchmark portfolio. Note that the returns of the active portfolio and the benchmark portfolio are determined from their respective wealth trajectories ($\mathcal{W}_{\mathcal{P}}$ and $\hat{\mathcal{W}}_{\hat{\mathcal{P}}}$) that are evaluated under the same investment horizon and same market conditions. The tracking error measures the volatility of the { difference in returns} over the investment horizon. A criticism of the tracking error is that it only measures the variability in the { difference in returns}, but does not reflect the magnitude of the {return difference} itself. For example, an active strategy with a constant { negative return difference} over the investment horizon would yield a better tracking error than an active strategy with a positive but volatile { return difference}. For this reason, many prefer the tracking difference \citep{johnson2013right, hougan2015tracking,charteris2020tracking,boyde2021etf}, which is defined as the annualized difference between the active portfolio's cumulative return and the benchmark portfolio's cumulative return over a specific period. Note that both tracking error and tracking difference metrics measure the { return difference} of the active portfolio over the benchmark portfolio. In other words, these metrics measure how closely the return of the active portfolio tracks the return of the benchmark portfolio. In practice, if an investment manager aims to achieve a certain annualized relative return target, e.g., $\beta$, then the tracking difference metric may not {be appropriate}. To address this, \citet{van2022dynamic} suggests the investment objective 
\begin{equation}
    \inf_{\mathcal{P}\in\mathcal{A}} \mathbb{E}_{\mathcal{P}}^{(t_0,w_0)}\Big[\Big(W(T)-e^{\beta T}\hat{W}(T) \Big)^2\Big],\label{prob:quad_terminal}
\end{equation}
where $W(T)$ and $\hat{W}(T)$ are the respective terminal wealth of the active portfolio and the benchmark portfolio at terminal time $T$, and $\beta$ is the annualized relative return target.

The optimal control problem (\ref{prob:quad_terminal}) aims to produce an active strategy that minimizes the quadratic difference between $W(T)$ and the terminal portfolio value target of $e^{\beta T}\hat{W}(T)$. In other words, the optimal control policy tries to outperform the benchmark portfolio by a total factor of $e^{\beta T}$ over the time horizon $[0,T]$, which is equivalent to an annualized relative return of $\beta$. The quadratic term of the difference incentivizes the terminal wealth of the active portfolio $W(T)$ to closely track the {\it elevated target} $e^{\beta T}$. 

It is worth noting that the relative return target $\beta$ can be intuitively interpreted as the manager's willingness to take more risks. As $\beta\downarrow0$, the optimal solution to problem (\ref{prob:quad_terminal}) is simply to mimic the benchmark strategy. However, as $\beta$ grows larger, the manager needs to take on more risk (for more return) in order to beat the benchmark portfolio by the relative return target rate.

A criticism of the investment objective (\ref{prob:quad_terminal}) is that it is symmetrical in terms of the outperformance and underperformance of $W(T)$ relative to the elevated target $e^{\beta T}\hat{W}(T)$. This is a common issue for volatility-based measures, such as the Sharpe ratio \citep{ziemba2005symmetric}. In practice, investors may favor outperformance more than underperformance, while still aiming to track the elevated target closely. Acknowledging this, instead of (\ref{prob:quad_terminal}), \citet{ni2022optimal} propose the following asymmetrical objective function,
\begin{equation}
   \inf_{\mathcal{P}\in\mathcal{A}}\mathbb{E}_{\mathcal{P}}^{(t_0,w_0)}\Bigg[\Big(\min(W(T)-e^{\beta T}\hat{W}(T),0) \Big)^2+\max\big(W(T)-e^{\beta T}\hat{W}(T),0\big)\Bigg].\label{prob:quad_terminal_asym}
\end{equation}

The investment objective (\ref{prob:quad_terminal_asym}) penalizes the outperformance (of $W(T)$ relative to the elevated target $e^{\beta T}\hat{W}(T)$) linearly but the underperformance quadratically, thus encouraging the optimal policy to favor outperformance more than underperformance when necessary. Note that the use of objective function (\ref{prob:quad_terminal_asym}) does not permit closed-form solutions and machine learning techniques are used \citep{ni2022optimal} to compute the desired optimal strategy numerically.

Another criticism of the investment objective (\ref{prob:quad_terminal}) and (\ref{prob:quad_terminal_asym}) is that both are only concerned with the relative performance at terminal time $T$. In reality, investment managers are often required to report intermediate portfolio performance internally or externally at regular time intervals. Instead of achieving the annualized relative return target when reviewing the portfolio performance at the end of the investment horizon,  managers may want to consistently achieve the relative return target throughout the entire investment horizon. In this case, managers may need to set an investment objective function to control the deviation of the wealth of the portfolio from the target along a market scenario within the investment horizon. Consequently, \citet{van2022dynamic} propose the following cumulative quadratic tracking difference (CD) objectives

\begin{numcases}{(CD(\beta)):\quad}
   \inf_{\mathcal{P}\in\mathcal{A}}\mathbb{E}_{\mathcal{P}}^{(t_0,w_0)}\Bigg[\int_{t_0}^T\Big(W(t)-e^{\beta t}\hat{W}(t)\Big)^2dt\Bigg],\; \text{if }\mathcal{T}=[t_0,T],\label{prob:CD_cont}
   \\
   \inf_{\mathcal{P}\in\mathcal{A}}\mathbb{E}_{\mathcal{P}}^{(t_0,w_0)}\Bigg[\sum_{t\in\mathcal{T}\cup\{T\}}\Big(W(t)-e^{\beta t}\hat{W}(t)\Big)^2\Bigg],\; \text{if }\mathcal{T}\subseteq[t_0,T], \mathcal{T}\text{ discrete}.\label{prob:CD_disc}
\end{numcases}
Here, note  that objective (\ref{prob:CD_cont}) 
is for the continuous rebalancing case, and (\ref{prob:CD_disc}) for discrete rebalancing. Both (\ref{prob:CD_cont}) and (\ref{prob:CD_disc}) measure the cumulative deviation of the wealth of the active portfolio relative to the target, along a market scenario within the entire investment horizon. Therefore, they  { measure} the intermediate performance deviations effectively. 
However, similar to (\ref{prob:quad_terminal}), (\ref{prob:CD_cont}) and (\ref{prob:CD_disc}) penalize outperformance and underperformance symmetrically. Therefore, we also consider following cumulative quadratic shortfall (CS) objectives that only penalize the shortfall (underperformance with respect to the target)
\begin{numcases}{(CS(\beta)):\quad}
   \inf_{\mathcal{P}\in\mathcal{A}}\mathbb{E}_{\mathcal{P}}^{(t_0,w_0)}\Bigg[\int_{t_0}^T\Big(\min\big(W(t)-e^{\beta t}\hat{W}(t),0\big)\Big)^2dt+\epsilon W(T)\Bigg],\; \text{if }\mathcal{T}=[t_0,T],\label{prob:CS_cont}
   \\
   \inf_{\mathcal{P}\in\mathcal{A}}\mathbb{E}_{\mathcal{P}}^{(t_0,w_0)}\Bigg[\sum_{t\in\mathcal{T}\cup\{T\}}\Big(\min\big(W(t)-e^{\beta t}\hat{W}(t),0\big)\Big)^2+\epsilon W(T)\Bigg],\; \text{if }\mathcal{T}\subseteq[t_0,T], \mathcal{T}
\text{ discrete}. \nonumber \\
    \label{prob:CS_disc}
\end{numcases}
Here (\ref{prob:CS_cont}) and (\ref{prob:CS_disc}) are the investment objectives for the continuous rebalancing and discrete rebalancing cases respectively. $\epsilon$ is a small regularization parameter to ensure that problem (\ref{prob:CS_cont}) and (\ref{prob:CS_disc}) are well-posed. A more detailed comparison of the CD and CS objective functions can be found in Appendix \ref{sec:shorfall_obj}.

\subsection{Closed-form solution for CD problem}\label{sec:closed-form}
In this section, we present the closed-form solution to the CD problem (\ref{prob:CD_cont}) under several assumptions. The closed-form solution not only provides us with insights for understanding the CD-optimal controls for problem (\ref{prob:CD_cont}), but also serves as the {baseline} for understanding the numerical results derived from the neural network method (discussed in later sections). Specifically, in this section, we consider the case that all asset prices follow jump-diffusion processes and portfolios with cash injections, which are aspects not frequently considered in benchmark outperformance literature \citet{browne_1999_a,browne_2000,tepla_2001,basak_2006,yao2006tracking,zhao2007dynamic,davis_2008,lim2010benchmarking,Oderda_2015,zhang2017portfolio,al2018outperformance,nicolosi2018portfolio,bo2021optimal}. 

We first summarize the assumptions for obtaining the closed-form solution to CD problem (\ref{prob:CD_cont}).
\begin{assumption}
(Two assets, no friction, unlimited leverage, trading in insolvency, constant rate of cash injection) The active portfolio and the benchmark portfolio have access to 
two underlying assets, a stock index, and a constant-maturity bond index. Both portfolios are rebalanced continuously, i.e., $\mathcal{T}=[t_0,T]$. There is no transaction cost and no leverage limit. Furthermore, we assume that trading continues in the event of insolvency, i.e., when $W(t)<0$ for some $t\in[t_0,T]$. Finally, we assume both portfolios receive constant cash injections with an injection rate of $c$, which means during any time interval $[t,t+\Delta t]\subseteq[t_0,T],\forall\Delta t>0$, both portfolios receive cash injection amount of $c\Delta t$.\label{as:basic_assump}
\end{assumption}
\begin{remark}
\textup{(Remark on Assumption \ref{as:basic_assump}) For illustration purposes, we assume only two underlying assets. 
However, the technique for deriving the closed-form solution can be extended to multiple assets. 
We remark that unlimited leverage is unrealistic, and is only assumed for deriving the closed-form solution. 
In Section \ref{sec:approx_form}, 
we will discuss the technique for handling the leverage constraint in more detail. We also acknowledge that it is 
not realistic to assume that the manager can continue to trade and borrow when insolvent. 
However, this assumption is typically required for obtaining closed-form solutions, 
see \citet{zhou2000continuous,li2000optimal} for the case of a multi-period mean-variance asset 
allocation problem. 
In Appendix \ref{sec:approx_form}, we have more discussion on the impact of insolvency and its handling in experiments.}
\end{remark}

\begin{assumption}
(Fixed-mix benchmark strategy) We assume that the benchmark strategy is a fixed-mix strategy (also known as the constant weight strategy). We assume the benchmark always allocates a constant fraction of $\hat{\varrho}\;(\in\mathbb{R})$ of the portfolio wealth to the stock index, and a constant fraction of $1-\hat{\varrho}$ to the bond index. Let $\hat{\boldsymbol{\varrho}}=(\hat{\varrho},1-\hat{\varrho})^\top\in\mathbb{R}^2$ denote the vector of allocation fractions to the stock index and the bond index, the benchmark strategy is the fixed-mix strategy defined by $\hat{\mathcal{P}}=\{\hat{\boldsymbol{p}}(\hat{\boldsymbol{X}}(t))=\big(\hat{p}_1(\hat{\boldsymbol{X}}(t)),\hat{p}_2(\hat{\boldsymbol{X}}(t))\big)^\top\equiv\hat{\boldsymbol{\varrho}},\;\forall t\in\mathcal{T}\}$.\label{as:fixed-mix-bm}
\end{assumption}

Finally, we assume the stock index price and bond index price follow the jump-diffusion processes described below.
\begin{assumption}
(Jump-diffusion processes) Let $S_1(t)$ and $S_2(t)$ denote the deflated (adjusted by inflation) price of the stock index and the bond index at time $t\in[t_0,T]$. We assume $S_i(t),\;i\in\{1,2\}$ follow the jump-diffusion processes
\begin{equation}
    \frac{dS_i(t)}{S_i(t^-)} = (\mu_i-\lambda_i\kappa_i+{r_i\cdot\textbf{1}_{S_i(t^-)<0}})dt+\sigma_idZ_i(t)+d\Big(\sum_{k=1}^{\pi_i(t)}(\xi_i^{(k)}-1)\Big),\;i=1,2.
\label{model:jump_diffusion}
\end{equation}
Here $\mu_i$ are the (uncompensated) drift rate, $\sigma_i$ is the diffusive volatility, $Z_1(t),Z_2(t)$ are correlated Brownian motions, where $\mathbb{E}[dZ_1(t)\cdot dZ_2(t)]=\rho dt$. 
{$r_i$  are the borrowing
premiums when $S_i(t^-)$ is negative.\footnote{Intuitively, there is a premium for shorting an asset. In the closed-form solution derivation, we assume $r_i=0$.}} $\pi_i(t)$ is a Poisson process with positive intensity parameter $\lambda_i$. $\{\xi_i^{(k)},\;k=1,\cdots,\pi_i(t)\}$ are i.i.d. positive random variables that describe jump multipliers associated with the assets. If a jump occurs for asset $i$ at time $t\in(t_0,T]$, its underlying price jumps from $S_i(t^-)$ to $S_i(t)=\xi_i\cdot S_i(t^-)$.\footnote{For any functional $\psi(t)$, we use the notation $\psi(t^-)$ as shorthand for the left-sided limit $\psi(t^-)=\lim_{\Delta t\downarrow0}\psi(t-\Delta t)$.} $\kappa_i=\mathbb{E}[\xi_i-1]$. $\xi_i$ and $\pi_i(t)$ are independent of each other. Moreover, $\pi_1(t)$ and $\pi_2(t)$ are assumed to be independent.\footnote{See \citet{forsyth2020optimal} for the discussion on the empirical evidence for stock-bond jump independence. Also note that the assumption of independent jumps can be relaxed without technical difficulty if needed \citep{kou2002jump}, but will significantly increase the complexity of notations.} \label{as:sto_process}
\end{assumption}

\begin{remark}
\textup{(Motivation for jump-diffusion model) The assumption of stock index price following a jump-diffusion model is common in the financial mathematics literature \citep{merton1976option,kou2002jump}. 
In addition, we follow the practitioner approach and directly model the returns of the constant maturity bond index as a stochastic process, see for example \citet{lin2015risking,macminn2014securitization}. 
As in \citet{macminn2014securitization}, we also assume that the constant maturity bond index follows a jump-diffusion process. 
During high-inflation regimes, central banks often make rate hikes to curb inflation, which causes sudden jumps in bond prices \citep{lahaye2011jumps}. 
We believe this is an appropriate assumption for bonds in high-inflation regimes.} 
\end{remark}

Under the jump-diffusion model (\ref{model:jump_diffusion}), the wealth processes for the active portfolio and benchmark portfolio are
\begin{equation}
\begin{cases}
    dW(t)=\Big(\sum\limits_{i=1}\limits^{N_a}p_i(\boldsymbol{X}(t^-))\cdot\frac{dS_i(t)}{S_i(t^-)}\Big)W(t_j^-)+cdt,
    \\
    d\hat{W}(t)=\Big(\sum\limits_{i=1}\limits^{N_a}\hat{p}_i(\hat{\boldsymbol{X}}(t^-))\cdot\frac{dS_i(t)}{S_i(t^-)}\Big)\hat{W}(t_j^-)+cdt,
\end{cases}\label{dynamics_jump_diffusion}
\end{equation}
where $t\in(t_0,T]$, $W(t_0)=\hat{W}(t_0)=w_0$ and $X(t^-)=(t,W(t^-),\hat{W}(t^-))^\top\in\mathbb{R}^3$ is the state variable vector.

We now derive the closed-form solution of the CD problem (\ref{prob:CD_cont}) under Assumption \ref{as:basic_assump}, \ref{as:fixed-mix-bm} and \ref{as:sto_process}. We first present the verification theorem for the HJB integro-differential equation (PIDE) satisfied by the value function and the optimal control of the CD problem (\ref{prob:CD_cont}).

\begin{theorem}
(Verification theorem for CD problem (\ref{prob:CD_cont}) For a fixed $\beta>0$, assume that for all $(t,w,\hat{w},\hat{\varrho})\in[t_0,T]\times\mathbb{R}^{3}$, there exists a function $V(t,w,\hat{w},\hat{\varrho}):[t_0,T]\times\mathbb{R}^{3}\mapsto\mathbb{R}$ and $p^*(t,w,\hat{w},\hat{\varrho}):[t_0,T]\times\mathbb{R}^{3}\mapsto\mathbb{R}^2$ that satisfy the following two properties. (i) $V$ and $\boldsymbol{p}^*$ are sufficiently smooth and solve the HJB PIDE (\ref{eq:PIDE}), and (ii) the function $\boldsymbol{p}^*(t,w,\hat{w},\hat{\varrho})$ attains the pointwise infimum in (\ref{eq:PIDE}) below
\begin{equation}
\begin{cases}
    \frac{\partial V}{\partial t}+(w-e^{\beta t}\hat{w})^2+\inf\limits_{\boldsymbol{p}\in\mathbb{R}^2}H(\boldsymbol{p};t,w,\hat{w},\hat{\boldsymbol{\varrho}}) = 0,
    \\
    V(T,w,\hat{w},\hat{\varrho})=0,
\end{cases}\label{eq:PIDE}
\end{equation}
where
\begin{align}
 H(\boldsymbol{p};t,w,\hat{w},\hat{\boldsymbol{\varrho}})=& \big(w\cdot\boldsymbol{\alpha}^\top\boldsymbol{p}+c\big)\cdot\frac{\partial V}{\partial w}+\big(\hat{w}\cdot\boldsymbol{\alpha}^\top\hat{\boldsymbol{\varrho}}+c\big)\cdot\frac{\partial V}{\partial \hat{w}}-\Big(\sum\limits_{i}\lambda_i\Big)\cdot V(t,w,\hat{w},\hat{\varrho})\notag\\
 &+\frac{w^2}{2}\cdot\big(\boldsymbol{p}^\top\boldsymbol{\Sigma}\boldsymbol{p}\big)\cdot\frac{\partial^2V}{\partial w^2}+\frac{\hat{w}^2}{2}\cdot\big(\hat{\boldsymbol{\varrho}}^\top\boldsymbol{\Sigma}\hat{\boldsymbol{\varrho}}\big)\cdot\frac{\partial^2V}{\partial \hat{w}^2}+w\hat{w}\cdot\big({\boldsymbol{p}}^\top\boldsymbol{\Sigma}\hat{\boldsymbol{\varrho}}\big)\cdot\frac{\partial^2V}{\partial w\partial\hat{w}}\notag\\
 &+\sum\limits_{i}\lambda_i\int_0^\infty V(w+p_iw(\xi-1),\hat{w}+\hat{p}_i\hat{w}(\xi-1),t,\hat{\varrho})f_{\xi_i}(\xi)d\xi.\label{eq:H}
\end{align}
Here $\boldsymbol{\alpha}=(\mu_1-\lambda_1\kappa_1,\mu_2-\lambda_2\kappa_2)^\top$ is the vector of (compensated) drift rates, 
$\boldsymbol{\Sigma}=\begin{bmatrix}
\sigma_1^2 & \rho\sigma_1\sigma_2 \\
\rho\sigma_1\sigma_2 & \sigma_2^2
\end{bmatrix}$ is the covariance matrix, and $f_{\xi_i}$ is the density function for $\xi_i$.

Then, under Assumption \ref{as:basic_assump}, \ref{as:fixed-mix-bm} and \ref{as:sto_process}, $V$ is the value function and $\boldsymbol{p}^*$ is the optimal control for the CD problem (\ref{prob:CD_cont}).
\end{theorem}
\begin{proof}
See Appendix \ref{sec:proof_verification}
\end{proof}
Define several auxiliary variables
\begin{equation}
\begin{cases}
\kappa_i^{(2)}=\mathbb{E}\big[(\xi_i-1)^2\big],\quad (\sigma_i^{(2)})^2=(\sigma_i)^2+\lambda_i\kappa_i^{(2)},\;i\in\{1,2\},\\
\vartheta=\sigma_1\sigma_2\rho-(\sigma_2^{(2)})^2,\quad\gamma=(\sigma_1^{(2)})^2+(\sigma_2^{(2)})^2-2{\sigma_1\sigma_2}\rho,\\
\phi = \frac{(\mu_1-\mu_2)(\mu_1-\mu_2+\vartheta)}{\gamma},\quad\eta = \frac{(\mu_1-\mu_2+\theta)^2}{\gamma}-(\sigma_2^{(2)})^2,
\end{cases}\label{def:soln_params}
\end{equation}
then we have the following proposition regarding the optimal control of problem (\ref{prob:CD_cont}).

\begin{proposition} (CD-optimal control)
Suppose Assumption \ref{as:basic_assump}, \ref{as:fixed-mix-bm} and \ref{as:sto_process} are applicable, then the optimal control fraction of the wealth of the active portfolio to be invested in the stock index for the $CD(\beta)$ problem (\ref{prob:CD_cont}) is given by ${p}^*(t,w,\hat{w},\hat{\varrho})\in\mathbb{R}$, where
\begin{equation}
p^*(t,w,\hat{w},\hat{\varrho}) = \frac{1}{W^*(t)}\Bigg[\frac{(\mu_1-\mu_2)}{\gamma }h(t;\beta,c)+\frac{(\mu_1-\mu_2+\vartheta)}{\gamma}\Big(g(t;\beta)\hat{W}(t)-W^*(t)\Big)+g(t;\beta)\hat{W}(t)\cdot\hat{\varrho}\Bigg].\label{def:optimal_control}
\end{equation}
Here $W^*(t)$ denotes the wealth process of the active portfolio from (\ref{dynamics_cont}) following control $\boldsymbol{p}^*(t,W^*(t),\hat{W}(t),\hat{\varrho})=\Big(p^*(t,W^*(t),\hat{W}(t),\hat{\varrho}),1-p^*(t,W^*(t),\hat{W}(t),\hat{\varrho})\Big)^\top$, where $p^*$ is the optimal stock allocation described in (\ref{def:optimal_control}), and $\hat{W}(t)$ is the wealth process of the benchmark portfolio following the fixed-mixed strategy described in Assumption \ref{as:fixed-mix-bm}. Here, $h$ and $g$ are deterministic functions of time,
\begin{equation}
g(t;\beta)=-\frac{D(t;\beta)}{2A(t)},\quad\quad h(t;\beta,c) = -\frac{B(t;\beta,c)}{2A(t)},\label{def:g_h}
\end{equation}
where $A,D$ and $B$ are deterministic functions defined as
\begin{equation}
    A(t)=\frac{e^{(2\mu_2-\eta)(T-t)}-1}{(2\mu_2-\eta)},\quad\quad 
    D(t;\beta)=2e^{\beta T}\Big(\frac{e^{-\beta(T-t)}-e^{(2\mu_2-\eta)(T-t)}}{2\mu_2-\eta+\beta}\Big),\label{def:A_D}
\end{equation}
and 
\begin{align}
    B(t;\beta,c)&=\frac{2c}{2\mu_2-\eta}\Big(\frac{e^{(2\mu_2-\eta)(T-t)}-e^{(\mu_2-\phi)(T-t)}}{\mu_2+\phi-\eta}-\frac{e^{(\mu_2-\phi)(T-t)}-1}{\mu_2-\phi}\Big)\nonumber\\
    &+\frac{2ce^{\beta T}}{2\mu_2-\eta+\beta}\Big(\frac{e^{(\mu_2-\phi)(T-t)}-e^{-\beta(T-t)}}{\mu_2-\phi+\beta}-\frac{e^{(2\mu_2-\eta)(T-t)}-e^{(\mu_2-\phi)(T-t)}}{\mu_2+\phi-\eta}\Big).\label{def:B}
\end{align}
\end{proposition}
\begin{proof}
    See Appendix \ref{sec:obtain_opt_control}.
\end{proof}

\subsubsection{Insights from CD-optimal control}\label{sec:insights}
The CD-optimal control (\ref{def:optimal_control}) provides insights into the behaviour of
the optimal allocation policy. 
For ease of exposition, we first establish the following properties of $g(t;\beta)$ and $h(t;\beta,c)$.

\begin{corollary}\label{coro:prop_g} (Properties of $g(t;\beta)$) 
The function $g(t;\beta)$ defined in (\ref{def:g_h}) has the following properties for $t\in[t_0,T]$ and $\beta>0$:
\begin{enumerate}[label=(\roman*)]
\item For fixed $t\in[t_0,T]$, $g(t;\beta)$ is strictly increasing on $\beta\in(0,\infty)$. 
\item For fixed $\beta>0$, $g(t;\beta)$ is strictly increasing on $t\in[t_0,T]$.
\item $g(t;\beta)$ admits the following bounds:
\begin{equation}
    e^{\beta t}\leq g(t;\beta)\leq e^{\beta T}.
\end{equation}
\end{enumerate}
\end{corollary}
\begin{proof}
    See Appendix \ref{sec:proof-g-h}.
\end{proof}
\begin{corollary}\label{coro:prop_h} (Properties of $h(t;\beta,c)$) 
The function $h(t;\beta,c)$ defined in (\ref{def:g_h}) has the following properties for $t\in[t_0,T]$, $\beta>0$ and $c\geq0$:
\begin{enumerate}[label=(\roman*)]
\item For fixed $t\in[t_0,T]$ and $c>0$, $h(t;\beta,c)$ is strictly increasing  on $\beta\in(0,\infty)$. 
\item $h(t;\beta,c)\geq0$, $\forall (t,\beta,c)\in[t_0,T]\times(0,\infty)\times[0,\infty)$. 
\item For fixed $t\in[t_0,T]$ and $\beta>0$, $h(t;\beta,c)$ is strictly increasing on $c\in[0,\infty)$. $h(t;\beta,0)\equiv0$. Moreover, $h(t;\beta,c)\propto c$, i.e. $h(t;\beta,c)$ is proportional to $c$.
\end{enumerate}
\end{corollary}
\begin{proof}
    See Appendix \ref{sec:proof-g-h}.
\end{proof}
In order to analyze the closed-form solution, we make the following assumptions.
\begin{assumption}\label{as:drift_rates}
    (Drift rates of the two assets) We assume that the drift rates of the stock and the bond index $\mu_1$ and $\mu_2$ satisfy the following properties,
    \begin{equation}
        \mu_1-\mu_2>0,\quad \mu_1-\mu_2+\vartheta>0,
    \end{equation}
    where $\vartheta$ is defined in \eqref{def:soln_params}.
\end{assumption}
\begin{remark}\textup{(Remark on drift rate assumptions) The first inequality $\mu_1-\mu_2>0$ indicates that the stock index has a higher drift rate than the bond index, which is a standard assumption.\footnote{In fact, in this two-asset case, this assumption does not cause loss of generality.} The second inequality $\mu_1-\mu_2+\vartheta>0$ is also practically reasonable. $\vartheta$ is a variance term that is usually on a smaller scale compared to the drift rates. In reality, it is unlikely that $\mu_1-\mu_2>0$ but $\mu_1-\mu_2+\vartheta\leq0$.\footnote{For reference, based on the calibrated jump-diffusion model (\ref{model:jump_diffusion}) on historical high-inflation regimes, $\mu_1=0.051,\mu_2=-0.014,\vartheta=-0.00024$, and thus both inequalities are satisfied.}}
\end{remark}

Now we proceed to summarize the insights from the CD-optimal control (\ref{def:optimal_control}). The first obvious observation is that the CD-optimal control is a contrarian strategy. This can be seen from the fact that fixing time and the wealth of the benchmark portfolio $\hat{W}(t)$, the allocation to the more risky stock index decreases when the wealth of the active portfolio $W^*(t)$ increases.

If we take a deeper look at (\ref{def:optimal_control}), we can see that the CD-optimal control consists of two components: a cash injection component $p_{cash}^*$ and a tracking component $p_{track}^*$. Mathematically, 
\begin{equation}
p^*(t,w,\hat{w},\hat{\varrho}) = p_{cash}^*(t,w,\hat{w})+p_{track}^*(t,w,\hat{w},\hat{\varrho}),
\end{equation}
where
\begin{equation}
\begin{cases}
p_{cash}^*(t,w,\hat{w}) = \frac{1}{W^*(t)}\Bigg[\frac{(\mu_1-\mu_2)}{\gamma }h(t;\beta,c)\Bigg],\\
p_{track}^*(t,w,\hat{w},\hat{\varrho}) = \frac{1}{W^*(t)}\Bigg[\frac{(\mu_1-\mu_2+\vartheta)}{\gamma}\Big(g(t;\beta)\hat{W}(t)-W^*(t)\Big)+g(t;\beta)\hat{W}(t)\cdot\hat{\varrho}
\Bigg].\label{def:p_decomp}    
\end{cases}
\end{equation}

Based on Assumption \ref{as:drift_rates} and Corollary \ref{coro:prop_h}, the cash injection component $p_{cash}$ is always non-negative. Furthermore, from Corollary \ref{coro:prop_h}, we know that the stock allocation from the cash injection component is proportional to the cash injection rate $c$. In addition, as $t\uparrow T$, $h(t;\beta,c)$ increases, and thus the stock allocation from the cash injection component also increases with time.

On the other hand, the tracking component $p_{track}$ does not depend on the cash injection rate $c$, but only concerns the tracking performance of the active portfolio. One key finding is that 
\begin{equation}
    \begin{cases}
        p_{track}^*(t,w,\hat{w},\hat{\varrho})\geq \hat{\varrho},\qquad\text{if }W^*(t)\leq g(t;\beta)\hat{W}(t),\\
        p_{track}^*(t,w,\hat{w},\hat{\varrho})< \hat{\varrho},\qquad\text{if }W^*(t)>g(t;\beta)\hat{W}(t).
    \end{cases}
\end{equation}
This means that the CD-optimal control uses $g(t;\beta)\hat{W}(t)$ as the true target for the active portfolio to decide if the active portfolio should take more or less risk than the benchmark portfolio. This is a 
key observation, since the CD objective function (\ref{prob:CD_cont}) measures the difference between $W(t)$ and $e^{\beta t}\hat{W}(t)$. 
One would naively think that the optimal strategy would be based 
on the deviation from $e^{\beta t}\hat{W}(t)$. In contrast, from Corollary \ref{coro:prop_g}, we know that the true target $g(t;\beta)\hat{W}(t)$ used for decision making is greater than $e^{\beta t}\hat{W}(t)$. The insight from this observation is that if the manager wants to track an elevated target $e^{\beta t}\hat{W}(t)$, she should aim higher than the target itself.

\subsection{Leverage constraints}\label{sec:leverage}
In practice, large pension funds such as the Canadian Pension Plan often have exposures to alternative assets, such as private equity \citep{cpp2022}. Unfortunately, due to practical limitations, we only have access to 
long-term historical returns of publicly traded stock indexes and treasury bond indexes. Although controversial, some literature suggests that returns on private equity can be replicated using a leveraged small-cap stock index \citep{phalippou2014performance, l2016bottom}. Following this line of argument, we allow managers to take leverage to invest in public stock index funds to roughly mimic the pension fund portfolios with some exposure to private equities. 

Essentially, taking leverage to invest in stocks requires borrowing additional capital, 
which incurs borrowing costs. For simplicity, we assume the borrowing activity is represented by shorting some bond assets within the portfolio, and thus the manager is required to pay the cost of shorting these { shortable assets}. We assume that the cost consists of two parts: the returns of the shorted assets, and an additional borrowing premium (rate depends on specific investment scenarios) so that the total borrowing cost reflects both the interest rate environment (the return of shorted bond assets) and is reasonably estimated (with the added borrowing premium).

Following the notation from Section \ref{sec:formulation}, we assume that the total $N_a$ underlying assets are divided into two groups. The first group of $N_l$ assets are long-only assets, which we index by the set $\{1,\cdots,N_l\}$. The second group of $N_a-N_l$ assets are shortable assets that can be shorted to create 
leverage and are indexed by the set $\{N_l+1,\cdots,N_a\}$. Recall the notation of $p_i(\boldsymbol{X}(t))$ for the allocation fraction for asset $i$ at time $t$. 

For long-only assets, the wealth fraction needs to be non-negative, hence we have
\begin{equation}
    \text{(Long-only constraint):}\quad p_i(\boldsymbol{X}(t)) \geq 0,\; i\in\{1,\cdots,N_l\},\;t\in\mathcal{T}.\label{cons:long-only}
\end{equation}

Furthermore, the total allocation fraction for all assets should be one. Therefore, the following summation constraint needs to be satisfied
\begin{equation}
    \text{(Summation constraint):}\quad\sum_{i=1}^{N_a}p_i(\boldsymbol{X}(t))=1,\;t\in\mathcal{T}.\label{cons:sum-one}
\end{equation}

{
In practice, due to borrowing costs (from taking leverage) 
and risk management mandates, the use of leverage is often constrained. 
For this reason, we cap the maximum leverage by introducing a constant $p_{max}$, which represents the total allocation fraction for long-only assets. Therefore,
\begin{equation}
    \text{(Maximum leverage constraint):}
         \quad\sum_{i=1}^{N_l}p_i(\boldsymbol{X}(t)) \leq p_{max},\;t\in\mathcal{T}.\label{cons:max-lev}
\end{equation}
Note that no leverage is permitted if $p_{max}=1$.

Finally, we make the following assumption on the scenario of shorting multiple shortable assets.
\begin{assumption}
    (Simultaneous shorting) If one shortable asset has a
negative weight, other shortable assets 
must have nonpositive weights. Mathematically, this assumption can be expressed as \textup{
\begin{equation}
    \text{(Simultaneous shorting constraint):}
    \begin{cases}
       \;p_i(\boldsymbol{X}(t))\leq  0,\;\forall i\in\{N_l+1,\cdots,N_a\},\;\text{if }
                   \sum_{i=1}^{N_l}p_i(\boldsymbol{X}(t))>1,\;t\in\mathcal{T} \\
       \;p_i(\boldsymbol{X}(t)) \geq 0,\;\forall i\in\{N_l+1,\cdots,N_a\},\;\text{if }
                   \sum_{i=1}^{N_l}p_i(\boldsymbol{X}(t)) \leq 1,\;t\in\mathcal{T}
      \end{cases}
           .\label{constraint:simult_short}
\end{equation}    
    }\label{as:simult_short}
\end{assumption}

\begin{remark}
\textup{(Remark on Assumption \ref{as:simult_short}) This assumption avoids the ambiguity between the long-only assets and shortable assets in scenarios that involve leverage. When leveraging 
occurs, all shortable assets are treated as one group to provide the needed liquidity to achieve the desired leverage level.}
\end{remark}

The above constraints consider scenarios with non-negative portfolio wealth. Before we proceed to the handling of the negative portfolio wealth scenarios, we first define the following partition of the state space $\mathcal{X}$,
\begin{definition}\label{def:partition}{(Partition of state space)} 
We define $\big\{\mathcal{X}_1,\mathcal{X}_2\big\}$ to be a partition of the state space $\mathcal{X}$, such that
\begin{equation}
\begin{cases}
   \mathcal{X}_1 = \Big\{x=(t,W,\hat{W})^\top\in\mathcal{X}\Big|W \geq0\Big\},\\
   \mathcal{X}_2 = \Big\{x=(t,W,\hat{W})^\top\in\mathcal{X}\Big|W<0\Big\}.
\end{cases}
\end{equation}
\end{definition}
Intuitively, we separate the state space $\mathcal{X}$ into two regions by the wealth of the active portfolio, one with non-negative wealth and the other with negative wealth. Then, we present the following assumption concerning the negative wealth (insolvency) scenarios.

\begin{assumption}
    (No trading in insolvency) If the wealth of the active portfolio is negative, then all long-only asset positions should be liquidated, and all the debt (i.e. the negative wealth) is allocated to the least-risky shortable asset (in terms of volatility). Particularly, without loss of generality, we assume all debt is allocated to asset $N_l+1$. Let $\boldsymbol{e}_i\in\mathbb{R}^{N_a}=(0,\cdots,0,1,0,\cdots,0)^\top$ denote the standard basis vector of which the $i$-th entry is 1 and all other entries are 0. Then, we can formulate this assumption as follows. \textup{
\begin{equation}
    \text{(No trading in insolvency):}
    \quad p(\boldsymbol{X}(t))=\boldsymbol{e}_{N_l+1},\quad \text{if }\boldsymbol{X}(t)\in\mathcal{X}_2.
\end{equation}    
    }\label{as:no_trading_insolvency}
\end{assumption}

\begin{remark}
\textup{(Remark on Assumption \ref{as:no_trading_insolvency}) Essentially, when the portfolio wealth is negative, we assume the debt is allocated to a short-term bond asset and accumulates over time.}
\end{remark}

{
Summarizing the constraints, we can define two sets $\mathcal{Z}_1,\mathcal{Z}_2$:
\begin{numcases}
   \mathcal{Z}_1=\vast\{\boldsymbol{z}\in\mathbb{R}^{N_a}\Bigg|
   \begin{cases}
   z_i \geq 0,\forall i\in\{1,\cdots,N_l\},
   \\    
   \sum_{i=1}^{N_a}z_i=1,
   \\
   \sum_{i=1}^{N_l}z_i \leq  p_{max},
   \\
   z_i \leq  0,\;\forall i\in\{N_l+1,\cdots,N_a\},\;\text{if }\sum_{i=1}^{N_l}z_i>1, \\
   z_i \geq  0,\;\forall i\in\{N_l+1,\cdots,N_a\},\;\text{if }\sum_{i=1}^{N_l}z_i \leq 1
   \end{cases}
   \vast\}
   ,\label{def:Z_1}
   \\
   \mathcal{Z}_2=\big\{\boldsymbol{e}_{N_l+1}\big\},\label{def:Z_2}
\end{numcases}

Then, the corresponding space of feasible control vector values $\mathcal{Z}$ and the admissible strategy set $\mathcal{A}$ are

\begin{numcases}{(\text{Admissible set}):\quad}
   \mathcal{Z}=\mathcal{Z}_1\cup\mathcal{Z}_2,\label{control_space_lev}
   \\
   \mathcal{A}=\Bigg\{\mathcal{P}=\Big\{\boldsymbol{p}(\boldsymbol{X}(t)),\;t\in\mathcal{T}\Bigg|
   \begin{cases}
   \boldsymbol{p}(\boldsymbol{X}(t))\in\mathcal{Z}_1,\;\text{if } \boldsymbol{X}(t)\in\mathcal{X}_1,\\ 
   \boldsymbol{p}(\boldsymbol{X}(t))\in \mathcal{Z}_2,\;\text{if } \boldsymbol{X}(t)\in\mathcal{X}_2,\\
   \end{cases}
   \Big\}\Bigg\}.\label{admissible_strat_lev}
\end{numcases}

} 
{It is not obvious how the conditional constraints in (\ref{control_space_lev}) and (\ref{admissible_strat_lev}) can be formulated into a standard constrained optimization problem.}


\subsection{Neural network method}

In Section \ref{sec:closed-form}, we derive the closed-form solution under the jump-diffusion model, which requires several unrealistic assumptions such as continuous rebalancing, unlimited leverage, and trading in insolvency. Furthermore, the closed-form solution is specific to the investment objective defined in the CD problem (\ref{prob:CD_cont}).
 To discover optimal strategies for high inflation regimes,  capability in solving general investment problem  (\ref{generic_optimal_problem})  for different objectives and under realistic constraints, such as discrete rebalancing and limited leverage (i.e., leverage constraints discussed in Section \ref{sec:leverage}), is critically beneficial.  Therefore, we need computationally efficient methods to solve these problems numerically, particularly in high-dimensional cases.

Solving a discrete-time multi-period optimal asset allocation problem often utilizes dynamic programming (DP). For example, \citet{dixon2020machine,park2020intelligent,lucarelli2020deep,gao2020application} use Q-learning algorithms to solve the discrete-time multi-period optimal allocation problem. In general, if there are $N_a$ assets to invest in, then the use of Q-learning involves approximation of an action-value function (``Q'' function) which is a ${(2N_a+1)}$-dimensional function \citep{van2023beating} which represents the conditional expectation of the cumulative rewards at an intermediate state.\footnote{Intuitively, the dimensionality comes from tracking the allocation in the $N_a$ assets for both the active portfolio and benchmark portfolio when evaluating the changes in wealth of both portfolios over one period in the action-value function.} Meanwhile, the optimal control is a mapping from the state space to the allocation fractions to the assets. If the state space is relatively low-dimensional, \footnote{For example, the state space of problem (\ref{prob:CD_cont}) with assumptions of a fixed-mix strategy is a vector in $\mathbb{R}^3$.} then the DP-based approaches are potentially unnecessarily high-dimensional. 

Instead of using dynamic programming methods, 
\citet{han2016deep,BuehlerGononEtAl2018,tsang2020deep, reppen2022deep} propose to approximate the optimal control function by neural network functions directly. In particular, they propose a stacked neural network approach that essentially uses a sub-network to approximate the control at every rebalancing step. 
Therefore, the number of neural networks required grows linearly with the number of rebalancing periods.
Note that, in the taxonomy of \citet{Powell_2023}, this method is termed as
Policy Function Approximation (PFA).

In this article, we follow the lines of \citet{li2019data,ni2022optimal} and propose a single neural network to approximate the optimal control function. The direct representation 
of the control function avoids the high-dimensional approximation required in DP-based methods. In addition, we consider time $t$ as an input feature (along with the wealth of the active portfolio and benchmark portfolio), therefore avoiding the need for multiple sub-networks in the stacked neural network approach. 

The numerical solution to the general problem 
(\ref{generic_optimal_problem}) requires solving 
for the feedback control $\boldsymbol{p}$. We approximate the control function $\boldsymbol{\theta}$ by a neural network function $f(\boldsymbol{X}(t);\boldsymbol{\theta}):\mathcal{X}\mapsto\mathbb{R}^{N_a}$, where $\boldsymbol{\theta}\in\mathbb{R}^{N_{\boldsymbol{\theta}}}$ represents the parameters of the neural network (i.e., weights and biases). In other words, 
\begin{equation}
    \boldsymbol{p}(\boldsymbol{X}(t)) \simeq f(\boldsymbol{X}(t);\boldsymbol{\theta})\equiv f(\cdot;\boldsymbol{\theta}).\label{eq:nn_approx_idea}
\end{equation}
Then, the optimization problem (\ref{generic_optimal_problem}) can be converted to solving the following optimization problem.
\begin{equation}
    \text{(Parameterized optimization problem):}\quad\inf_{\boldsymbol{\theta}\in\mathcal{Z}_{\boldsymbol{\theta}}}\mathbb{E}_{f(\cdot;\boldsymbol{\theta})}^{(t_0,w_0)}\big[F(\mathcal{W}_{\boldsymbol{\theta}},\hat{\mathcal{W}}_{\hat{\mathcal{P}}})\big].\label{prob:NN_optimization}
\end{equation}
Here $\mathcal{W}_{\boldsymbol{\theta}}$ is the wealth trajectory of the active portfolio following the neural network approximation function parameterized by $\theta$. $\mathcal{Z}_{\boldsymbol{\theta}}\subseteq\mathbb{R}^{N_{\boldsymbol{\theta}}}$ is the feasibility domain of the parameter $\boldsymbol{\theta}$, which is translated from the constraints of the original problem, e.g., (\ref{control_space_lev}) and (\ref{admissible_strat_lev}). Mathematically,

{
\begin{equation}
 \mathcal{Z}_{\boldsymbol{\theta}}=\Bigg\{\boldsymbol{\theta}:
 \begin{cases}
 f(\boldsymbol{X};\boldsymbol{\theta})\in\mathcal{Z}_1,\;\text{if }\boldsymbol{X}\in\mathcal{X}_1,\\
 f(\boldsymbol{X};\boldsymbol{\theta})\in\mathcal{Z}_2,\;\text{if }\boldsymbol{X}\in\mathcal{X}_2.
 \end{cases}
 \Bigg\}.
 \label{def:Z_theta}
\end{equation}
Here $\mathcal{Z}_1,\mathcal{Z}_2$ are defined in (\ref{def:Z_1}), (\ref{def:Z_2}) and $\mathcal{X}_1,\mathcal{X}_2$ are  partitions of the state space $\mathcal{X}$ defined in Definition \ref{def:partition}.

Note here that $\mathcal{Z}_{\boldsymbol{\theta}}$ depends on  the structure of the neural network function $f(\cdot;\boldsymbol{\theta})$.  Intuitively, $\mathcal{Z}_{\boldsymbol{\theta}}$ is the preimage of $\mathcal{Z}$, i.e.,  any $\theta \in \mathcal{Z}_{\boldsymbol{\theta}}$,  $f(\cdot;\boldsymbol{\theta}) \in 
\mathcal{Z}$.} 
Specific neural network model design may result in $\mathcal{Z}_{\boldsymbol{\theta}}=\mathbb{R}^{N_{\boldsymbol{\theta}}}$, which means (\ref{prob:NN_optimization}) becomes an unconstrained optimization problem. For long-only investment problems, the only constraints are the long-only constraint (\ref{cons:long-only}) and the summation constraint (\ref{cons:sum-one}). 
Previous work has proposed a neural network architecture with a softmax 
activation function at the last layer so that the output (vector of allocation fractions) automatically satisfies the two constraints, and thus $\mathcal{Z}_{\boldsymbol{\theta}}=\mathbb{R}^{N_{\boldsymbol{\theta}}}$ and problem (\ref{prob:NN_optimization}) becomes an unconstrained optimization problem (see, e.g., \citet{li2019data, ni2022optimal}). However, as discussed in Section \ref{sec:leverage}, we consider the more complicated case where leverage and shorting are allowed. The problem thus involves more constraints than the long-only case and therefore we would like to design a new model architecture to convert the constrained optimization problem to an unconstrained problem. We will discuss the design of the {\it leverage-feasible neural network} (LFNN) model in the next section, and how the LFNN model achieves this goal.

It is worth noting that for the particular CD problem (\ref{prob:CD_disc}) and CS problem (\ref{prob:CS_disc}), our technique may be formulated to appear similar to policy gradient methods in RL literature \citep{silver2014deterministic} on a high level. Examples of policy gradient methods in financial problems include \citet{coache2021reinforcement}, in which the authors develop an actor-critic algorithm for portfolio optimization problems with convex risk measures. However, there are two main differences between our proposed methodology and policy gradient algorithms. Firstly, we assume that the randomness of the environment (i.e., asset returns) over the entire investment horizon is readily available upfront (e.g., through calibration of parametric models or resampling of historical data), which is a common assumption adopted by practitioners when backtesting investment strategies. On the other hand, RL literature often considers an unknown environment, and the algorithms focus on the exploration of the agent to learn from the unknown environment and thus may be unnecessarily complicated for our use case. Secondly, our proposed methodology is not limited to the cumulative reward framework in RL and thus is more universal and suitable for problems in which the investment objective cannot be easily expressed in the form of a cumulative reward.

\subsection{Leverage-feasible neural network (LFNN)}\label{sec:nn_model}
In this section, we propose the leverage-feasible neural network (LFNN) model, which { yields} $\mathcal{Z}_{\boldsymbol{\theta}}=\mathbb{R}^{N_{\boldsymbol{\theta}}}$ for leverage constraints defined in equation (\ref{control_space_lev}), and converts a constrained optimization problem (\ref{prob:NN_optimization}) to an unconstrained problem.

Let vector $\boldsymbol{x}=(t,W(t),\hat{W}(t))^\top\in\mathcal{X}$ be the feature (input) vector. We first define a standard fully-connected feedforward neural network (FNN) function $\Tilde{f}:\mathcal{X}\mapsto\mathbb{R}^{N_a+1}$ as follows:
\begin{equation}\label{def:FNN}(\text{FNN}):\quad
    \begin{cases}
    h^{(1)}_j=\text{Sigmoid}\Big(\sum_{i=1}^{N_x}x_i\theta^{(1)}_{ij}+b^{(1)}_{j}\Big),\;j = 1,\cdots, N_h^{(1)},\\
    h^{(k)}_j=\text{Sigmoid}\Big(\sum_{i=1}^{N_h^{(k-1)}}h^{(k-1)}_i\theta^{(k)}_{ij}+b^{(k)}_{j}\Big),\;j = 1,\cdots, N_h^{(k)},\;\forall k\in\{2,\cdots,K\},\\
    o_j=\sum_{i=1}^{N_h^{(K)}}h_i\theta^{(K+1)}_{ij},\;j = 1,\cdots, 
    N_a+1,\\
    {
    \Tilde{f}(\boldsymbol{x};\boldsymbol{\theta}):=(o_1,\cdots,o_{N_a+1})^\top.}
\end{cases}
\end{equation}
Here Sigmoid($\cdot$) denotes the sigmoid activation function, $K$ denotes the number of hidden layers, $h^{k}_j$ denotes the value of the $j$-th node in the $k$-th hidden layer, and $N_h^{(k)}$ is the number of nodes in the $k$-th hidden layer. Additionally, $\boldsymbol{\theta}^{(k)}=(\theta^{(k)}_{ij})\in\mathbb{R}^{N_h^{(k)}\times N_h^{(k-1)}}$ and $\boldsymbol{b}^{(k)}=(b^{(k)}_{j})\in\mathbb{R}^{N_h^{(k)}}$ are the (vectorized) weight matrix and bias vector for the $k$-th layer,\footnote{For $k=K+1$, define $N_h^{(K+1)}=N_a+1$ so $\boldsymbol{\theta}^{(K+1)}$ is well-defined.} and the parameter vector of the entire neural network is $\boldsymbol{\theta}=(\boldsymbol{\theta}^{(1)},\boldsymbol{b}^{(1)},\cdots,\boldsymbol{\theta}^{(K)},\boldsymbol{b}^{(K)},\boldsymbol{\theta}^{(K+1)})^\top\in\mathbb{R}^{N_{\boldsymbol{\theta}}}$, where $N_{\boldsymbol{\theta}}=\sum_{k=1}^{K+1} N_h^{(k)}\cdot N_h^{(k-1)}+\sum_{k=1}^{K}N_h^{(k)}$. 

Building on $\Tilde{f}$, we propose the following {\it leverage-feasible neural network} (LFNN) model $f:\mathcal{X}\mapsto{\mathcal{Z}}$:
\begin{equation}\label{LFNN}(\text{LFNN}):\quad
    f(\boldsymbol{x};\boldsymbol{\theta}):=\psi\Big(\Tilde{f}(\boldsymbol{x};\boldsymbol{\theta}),\boldsymbol{x}\Big)\in{\mathcal{Z}}.
\end{equation}

Here, $\psi(\cdot)$ is the {\it leverage-feasible activation function}. For $\boldsymbol{o}=(o_1,\cdots,o_{N_a+1})^\top\in\mathbb{R}^{N_a+1}$, and $\boldsymbol{p}=\psi(\boldsymbol{o},\boldsymbol{x})$, $\psi(\cdot): {(\boldsymbol{o},\boldsymbol{x})\in\mathbb{R}^{N_a+1}\times{\mathcal{X}}}\mapsto{\mathcal{Z}}$ is defined by

\begin{equation} \boldsymbol{p}=
\psi(\boldsymbol{o},\boldsymbol{x})=\begin{cases}
\begin{cases}
    l = p_{max}\cdot\text{Sigmoid}(o_{N_a+1}),\\
    p_i = l\cdot\frac{e^{o_i}}{\sum_{k=1}^{N_l}e^{o_k}},\;i\in\{1,\cdots,N_l\},\qquad\qquad\qquad\text{if } \boldsymbol{x}\in\mathcal{X}_1,\\
    p_i = (1-l)\cdot\frac{e^{o_i}}{\sum_{k=N_l+1}^{N_a}e^{o_k}},\;i\in\{N_l+1,\cdots,N_a\},
\end{cases}\\
\boldsymbol{e}_{N_l+1},\;\;\;\qquad\qquad\qquad\qquad\qquad\qquad\qquad\qquad\qquad\text{if } \boldsymbol{x}\in\mathcal{X}_2.\\
\end{cases}
\label{psi_lev_feas_func}
\end{equation}

Recall that $N_l$ is the number of long-only assets and $p_{max}$ is the maximum leverage allowed. We show that the leverage-feasible activation function $\psi$ has the following property.
{
\begin{lemma}\label{lemma:property_psi}  (Decomposition of $\psi$) 
The leverage-feasible function $\psi$ defined in (\ref{psi_lev_feas_func}) has the function decomposition that 
\begin{equation}
    \psi(\boldsymbol{o},\boldsymbol{x})=\varphi(\zeta(\boldsymbol{o}),\boldsymbol{x}),
\end{equation}
where
\begin{equation}\label{eq:psi_decomp}
    \begin{cases}
        \zeta:\mathbb{R}^{N_a+1}\mapsto \Tilde{\mathcal{Z}},\zeta(o)=\Bigg(\text{Softmax}\Big((o_1,\cdots,o_{N_l})\Big),\text{Softmax}\Big((o_{N_l+1},\cdots,o_{N_a})\Big),p_{max}\cdot\text{Sigmoid}(o_{N_a+1})\Bigg)^\top,\\
        \varphi:\Tilde{\mathcal{Z}}\times{\mathcal{X}}\mapsto{\mathcal{Z}},\varphi(z)=\Big(z_{N_a+1}\cdot (z_1,\cdots,z_{N_l}),(1-z_{N_a+1})\cdot(z_{N_l+1},\cdots,z_{N_a})\Big)^\top\cdot\textbf{1}_{\boldsymbol{x}\in\mathcal{X}_1}+\boldsymbol{e}_{N_l+1}\cdot\textbf{1}_{\boldsymbol{x}\in\mathcal{X}_2},
    \end{cases}
\end{equation}
and 
\begin{equation}
 \Tilde{\mathcal{Z}}=\Bigg\{z\in\mathbb{R}^{N_a+1},\sum_{i=1}^{N_l}z_i=1,\sum_{i=N_l+1}^{N_a}z_i=1,z_{N_a+1}\leq p_{max}, z_i\geq0, \forall i\Bigg\}.   
\end{equation}
\end{lemma}
\begin{proof} This is easily verifiable by definition of $\psi$ in (\ref{psi_lev_feas_func}).
\end{proof}

\begin{remark}
\textup{(Remark on Lemma \ref{lemma:property_psi}) The leverage-feasible activation function $\psi$ { corresponds to}  a two-step decision process described by $\zeta$ and $\varphi$. Intuitively, $\zeta$ first determines the internal allocations within long-only assets and shortable assets, as well as the total leverage. Then, $\varphi$ converts the internal allocations and total leverage into final allocation fractions, which depend on the wealth of the active portfolio.
}
\end{remark}
}

With the LFNN model outlined above, the parameterized optimization problem (\ref{prob:NN_optimization}) becomes an unconstrained optimization problem. Specifically, we present the following theorem regarding the feasibility domain $\mathcal{Z}_{\boldsymbol{\theta}}$ associated with the LFNN model (\ref{LFNN}).
\begin{theorem}\label{theorem:feasibility_domain}
(Unconstrained feasibility domain) The feasibility domain $\mathcal{Z}_{\boldsymbol{\theta}}$ defined in (\ref{def:Z_theta}) associated with the LFNN model (\ref{LFNN}) is $\mathbb{R}^{N_{\boldsymbol{\theta}}}$.
\end{theorem}
\begin{proof} See Appendix \ref{app:proof_feasibility_dom}.
\end{proof}

{Following Theorem \ref{theorem:feasibility_domain}, the constrained optimization problem (\ref{generic_optimal_problem}) can be transformed into the following unconstrained optimization problem\begin{equation} \text{(Unconstrained parameterized problem):}
    \quad\inf_{\boldsymbol{\theta}\in\mathbb{R}^{N_{\boldsymbol{\theta}}}}\mathbb{E}_{f(\cdot;\boldsymbol{\theta})}^{(t_0,w_0)}\big[F(\mathcal{W}_{\boldsymbol{\theta}},\hat{\mathcal{W}}_{\hat{\mathcal{P}}})\big].\label{prob:unconstrained_LFNN}
\end{equation}

\subsection{Mathematical justification for LFNN approach}\label{sec:justify_LFNN}
By approximating the feasible control with a parameterized LFNN model, we have shown that the original constrained optimization problem is transformed into an unconstrained optimization problem, which is computationally more implementable.

However, an important question remains: is the solution to the parameterized unconstrained optimization problem (\ref{prob:unconstrained_LFNN}) capable of yielding the optimal control of the original problem (\ref{generic_optimal_problem})? In other words, suppose $\boldsymbol{\theta}^*$ is the solution to (\ref{prob:unconstrained_LFNN}), can  $f(\cdot;\boldsymbol{\theta}^*)$ approximates solution to (\ref{generic_optimal_problem}) with desired accuracy?

In this section, we prove that under benign assumptions and appropriate choices 
of the hyperparameter of the LFNN model (\ref{LFNN}), solving 
the unconstrained problem (\ref{prob:unconstrained_LFNN}) provides
an arbitrarily close approximation 
the original 
problem (\ref{generic_optimal_problem}). We start by establishing the following lemma.
\begin{lemma}\label{lemma:decomp_p_opt}  (Structure of feasible control) 
Any feasible control function $p:\mathcal{X}\mapsto \mathcal{Z}$, where $\mathcal{Z}$ is defined in \eqref{admissible_strat_lev}, has the function decomposition
\begin{equation}
    p(x)=\varphi(\omega(x),x),
\end{equation}
where
$\varphi:\Tilde{\mathcal{Z}}\times\mathcal{X}\mapsto{\mathcal{Z}}$ is defined in (\ref{eq:psi_decomp}) and $\omega:\mathcal{X}\mapsto\Tilde{\mathcal{Z}}$.
\end{lemma}
\begin{proof}
    See Appendix \ref{app:validity_LFNN}.
\end{proof}

Next, we propose the following benign assumptions on the state space and the optimal control.
\begin{assumption} (Assumption on state space and optimal control) \\
\begin{enumerate}[label=(\roman*)]
\item The space $\mathcal{X}$ of state variables  is a compact set.
\item Following Lemma \ref{lemma:decomp_p_opt}, the optimal control $p^*:\mathcal{X}\mapsto{\mathcal{Z}}$ has the decomposition $p^*(x)=\varphi(\omega^*(x),x)$ for some $\omega^*:\mathcal{X}\mapsto\Tilde{\mathcal{Z}}$. We assume $\omega^*\in C(\mathcal{X},\Tilde{\mathcal{Z}})$, {where $C(\mathcal{X},\Tilde{\mathcal{Z}})$ denotes the set of continuous mappings from $\mathcal{X}$ to $\Tilde{\mathcal{Z}}$.}
\end{enumerate}\label{as:approx_control}
\end{assumption}

\begin{remark}
\textup{(Remark on Assumption \ref{as:approx_control}) In our particular problem of outperforming a benchmark portfolio, the state variable vector is $X(t)=(t,W(t),\hat{W}(t))^\top\in\mathcal{X}$ where $t\in[0,T]$. In this case, assumption (i) is equivalent to the assumption that the wealth of the active portfolio and benchmark portfolio is bounded, i.e. {$\mathcal{X}=[0,T]\times[w_{min},w_{max}]\times[w_{min},\hat{w}_{max}]$, where $w_{min},w_{max}$ and $\hat{w}_{min},\hat{w}_{max}$ are the respective wealth bounds for the portfolios.} Intuitively, assumption (ii) states that the decision process for the optimal control to obtain the internal allocation fractions within the long-only assets, shortable assets, and the total leverage is a continuous function. This is a natural extension of the long-only case, in which it is commonly assumed that the allocation within long-only assets is a continuous function of state variables. }
\end{remark}

Finally, we present the following theorem.
\begin{theorem}\label{theorem:approximation_opt_control} (Approximation of optimal control) Following Assumption \ref{as:approx_control}, $\forall\epsilon>0$, there exists $N_h\in\mathbb{N}$, and $\boldsymbol{\theta}\in\mathbb{R}^{N_{\boldsymbol{\theta}}}$ such that the corresponding LFNN model $f(\cdot;\boldsymbol{\theta})$ described in (\ref{LFNN}) satisfies the following:
\begin{equation}\label{def:dist_func}
    \sup_{x\in\mathcal{X}}\|f(x;\boldsymbol{\theta})-p^*(x)\| <\epsilon.
\end{equation}
\end{theorem}
\begin{proof}
    See Appendix \ref{app:validity_LFNN}.
\end{proof}

Theorem \ref{theorem:approximation_opt_control} shows that given any arbitrarily small tolerance $\epsilon>0$, there exists a suitable choice of the hyperparameter of the LFNN model (e.g. the number of hidden layers and nodes), and a parameter vector $\boldsymbol{\theta}$, such that the corresponding parameterized LFNN function is within this tolerance of the optimal control function.\footnote{The distance is defined in (\ref{def:dist_func}), i.e. the supremum of the pointwise distance over the extended state space $\mathcal{X}$.} In other words, with a large enough LFNN model (in terms of the number of hidden nodes), solving the unconstrained parameterized problem (\ref{prob:unconstrained_LFNN}) approximately solves the original optimization problem (\ref{generic_optimal_problem}) with any required precision.
}
{
\begin{remark}
\textup{(Empirical evidence of approximation) In practice, we find that a small neural network structure with one single hidden layer and only 10 hidden nodes 
achieves excellent approximation 
performance. In particular, in a numerical experiment with simulated data, we compare the LFNN model with the approximate form of the closed-form 
solution derived in Section \ref{sec:closed-form}, and find that the LFNN model mimics the closed-form solution very well. 
This provides further empirical evidence that supports Theorem \ref{theorem:approximation_opt_control}. Additional details can be found in Appendix \ref{sec:validate_NN}.}
\end{remark}
}

\subsection{Training LFNN}
Since the numerical experiments involve the solution and evaluation of the optimal parameters $\boldsymbol{\theta}^*$ of the LFNN model (\ref{LFNN}) in problem (\ref{prob:unconstrained_LFNN}), we briefly review {how the parameters are computed in experiments.}

In numerical experiments, the expectation in (\ref{prob:unconstrained_LFNN}) is approximated by using a finite set of samples of the set $\boldsymbol{Y}=\{Y^{(j)}:j=1,\cdots,N_d\}$, where $N_d$ is the number of samples, and $Y^{(j)}$ represents a time series sample of {\it joint} asset return observations $R_i(t),\;i\in\{1,\cdots,N_a\}$, observed at $t\in\mathcal{T}$.\footnote{Note that the corresponding set of asset prices can be easily inferred from the set of asset returns, or vice versa.} Mathematically, problem (\ref{prob:unconstrained_LFNN}) is approximated by 
\begin{equation}
    \quad\inf_{\boldsymbol{\theta}\in\mathbb{R}^{N_{\boldsymbol{\theta}}}}\Bigg\{\frac{1}{N_d}\sum_{j=1}^{N_d}F\left(\mathcal{W}^{(j)}_{\boldsymbol{\theta}},\hat{\mathcal{W}}^{(j)}_{\hat{\mathcal{P}}}\right)\Bigg\}.\label{prob:training}
\end{equation}
Here $\mathcal{W}^{(j)}_{\boldsymbol{\theta}}$ is the wealth trajectory of the active portfolio following the LFNN parameterized by $\boldsymbol{\theta}$, and $\hat{\mathcal{W}}^{(j)}$ is the wealth trajectory of the benchmark portfolio following the benchmark strategy $\hat{\mathcal{P}}$, both evaluated on $Y^{(j)}$, the $j$-th time series sample. 

 We use a shallow neural network model,  specifically,  an LFNN model with one single hidden layer with 10 hidden nodes, i.e., $K=1$ and $N_h^{(1)}=10$. We use the 3-tuple vector $(t,W_{\boldsymbol{\theta}}(t),\hat{W}(t))^\top$ as the input (feature) to the LFNN network. At $t\in[t_0,T]$, $W_{\boldsymbol{\theta}}(t)$ is the wealth of the active portfolio of the strategy {that follows} the LFNN model parameterized by $\boldsymbol{\theta}$, and $\hat{W}(t)$ is the wealth of the benchmark portfolio.

Then, the optimal parameter $\boldsymbol{\theta}^*$ can be numerically obtained by solving problem (\ref{prob:training}) using standard optimization algorithms such as ADAM \citep{kingma2014adam}. This process is commonly referred to as ``training'' of the neural network model, and $\boldsymbol{Y}$ is often referred to as the training data set \citep{goodfellow2016deep}.

Once $\boldsymbol{\theta}^*$ is numerically obtained, the resulting optimal strategy $f(\cdot;\boldsymbol{\theta}^*)$ is evaluated on a separate ``testing'' data set $\boldsymbol{Y}^{test}$, which contains a different set of samples generated from either the same distribution of the training process or a different process (depending on experiment purposes) so that the ``out-of-sample'' performance of $f(\cdot;\boldsymbol{\theta}^*)$ is assessed.


\section{Numerical experiments}\label{sec:numerical_experiments}
{In this section, we present a case study that explores optimal asset allocation during high-inflation periods using the LFNN model through numerical experiments. To conduct our analysis, we need data specifically from high inflation periods. Such data can be acquired using parametric modeling or non-parametric sample generation methods. It is important to note that our LFNN approach is agnostic to the choice of data modeling methods. While there is no universally accepted method for identifying or modeling high inflation regimes, for the purpose of this demonstration, we employ a simple filtering technique to identify inflation regime data and generate the required samples for training the LFNN.}

\subsection{Filtering historical inflation regimes}\label{sec:filtering_regimes}
We use the U.S. CPI index and monthly data from the Center for
Research in Security Prices (CRSP) over the
1926:1-2022:1 period.\footnote{The date convention is that, for example, 1926:1 refers
to January 1, 1926.}\footnote{More specifically, results presented here
were calculated based on data from Historical Indexes, \copyright
2022 Center for Research in Security Prices (CRSP), The University of
Chicago Booth School of Business. Wharton Research Data Services (WRDS) was
used in preparing this article. This service and the data available
thereon constitute valuable intellectual property and trade secrets
of WRDS and/or its third-party suppliers.}
We select high-inflation periods as determined by the CPI index using the following filtering procedure.  
Using a moving window of
$k$ months,  we determine the cumulative CPI index log return (annualized) in this window.  If
the cumulative annualized CPI index log return is greater than a cutoff, then all the months in the window
are flagged as part of a high-inflation regime.  Note that some months may appear in more
than one moving window. Any months which do not meet this criterion are considered to be
in low-inflation regimes.
{ See Algorithm \ref{Algo:filter}  in Appendix \ref{sec:window_algo} for the pseudo-code. }

Since the average annual inflation over the period 1926:1-2022:1 was 2.9\%,  
and Federal Reserve policymakers have been targeting the inflation rate of 2\% over the long run to achieve maximum employment and price stability \citep{fed2011},
we use a cutoff of 5\% as the threshold for high inflation.
In addition, we use the moving window size of 5 years (see Appendix \ref{app:filter_window_size} for more discussion). This uncovers two inflation regimes: 1940:8-1951:7 and 1968:9-1985:10, which correspond to well-known market shocks (i.e. the second world war, and price controls; the oil price shocks and stagflation of the seventies).

Table \ref{regime_data} shows the average annual inflation over the two regimes identified from our filter.
\begin{table}[hbt!]
\begin{center}
\begin{tabular}{cc} \toprule
Time Period & Average Annualized Inflation \\
\hline
1940:8-1951:7 & .0564  \\
1968:9-1985:10 & .0661 \\
\bottomrule
\end{tabular}
\caption{Inflation regimes determined using a five-year moving window with a cutoff inflation rate of $0.05$.}
\label{regime_data}
\end{center}
\end{table}

For possible investment assets, we consider the 30-day U.S. T-bill index (CRSP designation ``t30ind''),  a constant maturity 10-year U.S. treasury index,\footnote{The
10-year treasury index was generated from monthly returns from CRSP
back to 1941 (CRSP designation ``b10ind''). The data for 1926-1941 are interpolated from annual
returns in \citet{homer1996history}. The 10-year treasury index
is constructed by (a) buying a 10-year treasury at the start of each month, (b) collecting interest during the month,
and then (c) selling the treasury at the end of the month.  We repeat the process at the start of
the next month.  The gains in the index then reflect both interest and capital gains and losses.} 
 and the cap-weighted stock index (CapWt)
and the equal-weighted stock index (EqWt), also from CRSP.\footnote{The capitalization-weighted total returns have the CRSP designation  ``vwretd'', and the equal-weighted
total returns have the CRSP designation ``ewretd''.}  
All of these various indexes are adjusted for inflation by using the U.S.\ CPI
index.


We find that the equal-weighted stock index has a higher average return and higher volatility than the cap-weighted stock index. In addition, we find that the 30-day T-bill index has a similar average return as the 10-year T-bond index, but much lower volatility, see Appendix \ref{app:inflation_asset_return} for more details. 
This indicates that the T-bill index is the better choice of a defensive asset during high inflation.
Subsequently, we consider the equal-weighted stock index, the cap-weighted stock index, and the 30-day T-bill index.

\subsection{Bootstrap resampling}\label{sec:high_inflation_bootstrap}
{
Once we have obtained the filtered historical high-inflation data series from Section \ref{sec:filtering_regimes}, it becomes necessary to generate training and testing data sets from the original time series data. While one common approach is to assume and fit a parametric model to the underlying data, it is important to acknowledge the limitations associated with this choice.

Parametric models have several drawbacks, including the difficulty of accurately estimating their parameters \citep{black1993estimating}. Even for a simple geometric Brownian motion (GBM) model, accurately estimating the drift rate can be challenging and prone to errors, requiring a long historical period of data coverage \citep{brigo2008stochastic}. More complex models, such as the jump-diffusion model (\ref{model:jump_diffusion}), introduce additional components to the stochastic model, which necessitates the estimation of extra parameters. Furthermore, parametric models inherently make assumptions about the true stochastic model for asset prices, which can be subject to debate.

Acknowledging the above limitations of parametric market data models, we turn to the alternative non-parametric method of bootstrap resampling as a data-generating process for numerical experiments. Unlike parametric models, non-parametric models such as bootstrap resampling do not make assumptions about the parametric form of the asset price dynamics. Intuitively speaking, the bootstrap resampling method randomly chooses data points from the historical time series data and reassembles them into new paths of time series data. 
The bootstrap was initially proposed as a statistical method for estimating the sampling distribution of statistics \citep{efron1992bootstrap}. We use it as a data-generating procedure, as the philosophy behind bootstrap resampling is {consistent with the 
idea that {\em``history does not repeat, but it rhymes.''}.} 
The bootstrap resampling provides an empirical distribution, which seems to be the least prejudiced estimate possible of the underlying distribution of the data-generating process. We also note that bootstrap resampling is widely adopted by practitioners \citep{alizadeh2007investment,cogneau2013block,dichtl2016timing,scott2017wealth,shahzad2019bitcoin,cavaglia2022multi,simonian2022sharpe} as well as academics \citep{anarkulova2022stocks}.

Specifically, we choose to use the stationary block bootstrap resampling method
\citep{politis1994stationary}.} See Appendix \ref{app:bootstrap} for detailed pseudo-code for bootstrap resampling. Compared to the traditional bootstrap method, the block bootstrap technique preserves the local dependency of data within blocks. Furthermore, the stationary block bootstrap uses random blocksizes which preserves the stationarity of the original time series data.
An important parameter is the expected blocksize, which, informally, is a measure of serial correlation in the return data. A challenge in using block bootstrap resampling is the need to choose a single blocksize for multiple underlying time series data so that the bootstrapped data entries for different assets are synchronized in time. Subsequently, we use the expected blocksize of 6 months for all time series data. However, we have compared different numerical experiments using a range of blocksizes, including i.i.d. assumptions (i.e. expected blocksize equal to one month), and find that the results are relatively insensitive to blocksize, as discussed in more detail in Appendix \ref{app:blk_effect}.

Typically, the bootstrap technique resamples from data sourced from one contiguous segment of historical periods. However, the moving-window filtering algorithm has identified two non-contiguous historical inflation regimes. To apply the bootstrap method, there are two intuitive possibilities: 1) concatenate the two historical inflation regimes first, then bootstrap from the concatenated combined series, or 2) bootstrap within each regime (i.e., using circular block bootstrap resampling within each regime), then combine the bootstrapped resampled data points. We have experimented with both methods and find that the difference is minimal (see Appendix \ref{app:bootstrap_noncont}). In this article, we adopt the first method, i.e., we concatenate the historical regimes first, then bootstrap from the combined series. 
This method is also adopted by \citet{anarkulova2022stocks}, where stock returns from different countries 
are concatenated and the bootstrap is applied to the combined data.

\subsection{A case study on high inflation investment: a 4-asset scenario}\label{sec:numerical_experiment}
\subsubsection{Experiment setup}
{
In this section, we conduct a case study on optimal asset allocation during a consistent high-inflation regime. The details of the investment specification are given in Table \ref{nn_main_case}. Briefly, the active portfolio and benchmark portfolio begin with the same initial wealth of 100 at $t_0=0$. Both portfolios are rebalanced monthly. The investment horizon is 10 years, and there is an annual cash injection of 10 for both portfolios, evenly divided over 12 months.

We consider an empirical case in which we allow the manager to allocate between four investment assets: 
the equal-weighted stock index, the cap-weighted stock index, the 30-day U.S. T-bill index, 
and the 10-year U.S. T-bond index. We assume that the stock indexes and the 10-year T-bond index are long-only assets. 
The manager can short the T-bill index to take leverage and invest in the long-only assets 
(with maximum total leverage of 1.3). In this experiment, we assume the borrowing premium rate 
is zero. Essentially, we assume that the manager can borrow short-term funding to take leverage 
at the same cost as the treasury bill. This may be a reasonable assumption for 
sovereign wealth funds, as they are state-owned and enjoy a high credit rating. 
We remark that the borrowing premium does not really affect the results significantly.\footnote{See Appendix \ref{sec:borrow_premium} for a more detailed discussion.} The annual outperformance target $\beta$ is set to be 2\% (i.e. 200 bps per year).

It is worth noting that we choose the benchmark portfolio to be a fixed-mix portfolio that maintains a 70\% weight in the equal-weighted stock index and 30\% in the 30-day U.S. T-bill index. We select this fixed-mix portfolio as the benchmark based on our observation that the equal-weighted stock index shows superior performance compared to the cap-weighted stock index during high-inflation environments. Surprisingly, when analyzing bootstrap resampled data from the historical inflation regimes, we find that the fixed-mix portfolio consisting of 70\% in the equal-weighted stock index and 30\% in the 30-day U.S. T-bill index partially
stochastically dominates the fixed-mix portfolio consisting of 70\% in the cap-weighted stock index and 30\% in the 30-day U.S. T-bill index. For more detailed information, interested readers can refer to Appendix \ref{sec:sto_dom}.
}

\begin{table}[htb]
\begin{center}
\begin{tabular}{lc} \toprule
Investment horizon $T$ (years) & 10  \\
Equity market indexes & CRSP cap-weighted/equal-weighted index (real) \\
Bond index & CRSP 30-day/10-year U.S. treasury index (real) \\
Index samples for bootstrap& Concatenated 1940:8-1951:7, 1968:9-1985:10\\
Initial portfolio wealth/annual cash injection  & 100/10 \\
Rebalancing frequency & Monthly\\
Maximum leverage  & 1.3\\
Outperformance target rate $\beta$ & 2\%\\
\bottomrule
\end{tabular}
\caption{Investment scenario. 
\label{nn_main_case}}
\end{center}
\end{table}

As discussed in the previous section, we use the stationary bootstrap resampling algorithm (see Appendix \ref{app:bootstrap}) to generate 
a training data set $\boldsymbol{Y}$ and a testing data set $\boldsymbol{Y}^{test}$ (both with 10,000 resampled paths) from the concatenated index samples from two historical inflation regimes: 1940:8-1951:7 and 1968:9-1985:10, using an expected blocksize of 6 months. The testing data set $\boldsymbol{Y}^{test}$ is generated using a different random seed as the training data set $\boldsymbol{Y}$, and thus the probability of seeing the same sample in $\boldsymbol{Y}$ and $\boldsymbol{Y}^{test}$ is near zero (see \citet{ni2022optimal} for proof). 

We remark that in this experiment, we train the LFNN model (\ref{LFNN}) on $\boldsymbol{Y}$ under the discrete-time CS objective (\ref{prob:CS_disc_opt}), instead of the CD objective (\ref{prob:CD_disc}). As discussed in Section \ref{sec:obj_choice}, the CS objective function only penalizes underperformance relative to the elevated target. Numerical comparisons of the two objective functions suggest that the CS objective function indeed yields more favorable investment results than the CD objective (see Appendix \ref{sec:shorfall_obj}). 

{In this section, unless stated otherwise, all the results presented are testing results.}

\subsection{Experiment results}
\begin{figure}[htb]
\centerline{%
\begin{subfigure}[t]{.33\linewidth}
\centering
\includegraphics[width=\linewidth]{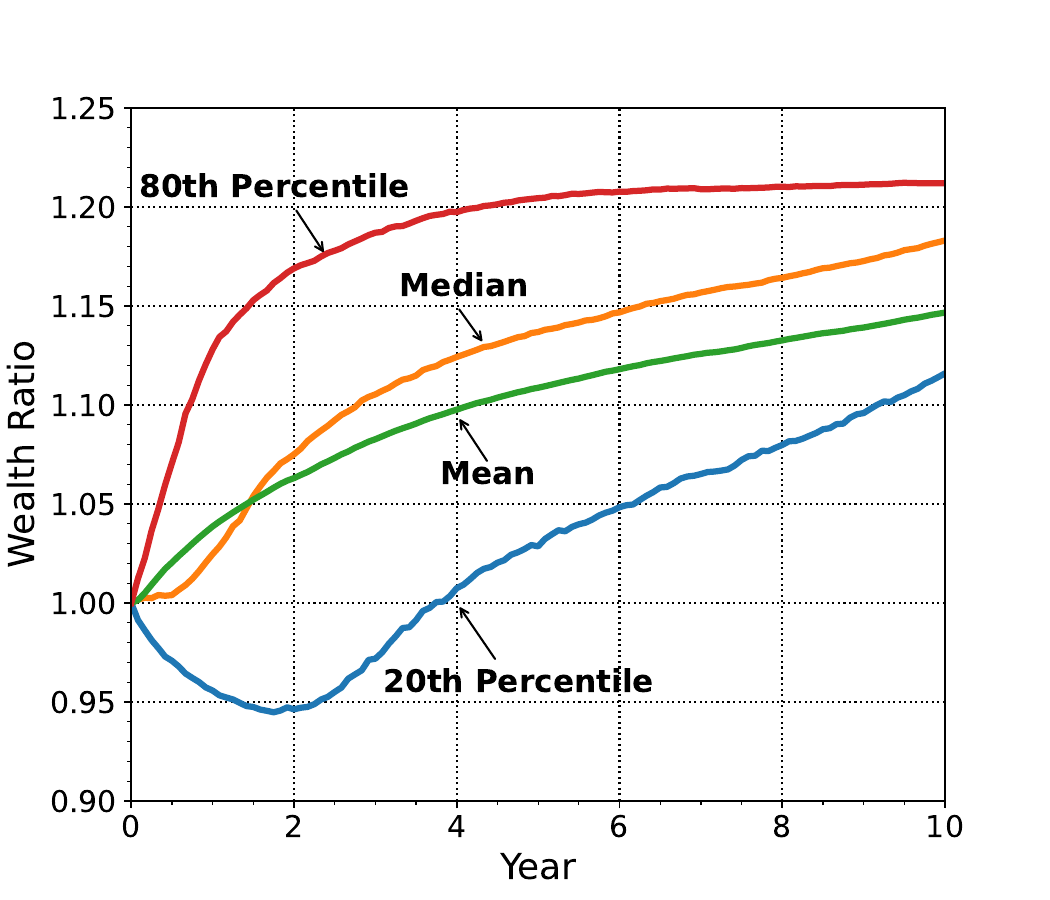}
    \caption{Percentiles of wealth ratio $\frac{W(t)}{\hat{W}(t)}$}
\label{fig:wealth_ratio}
\end{subfigure}
\begin{subfigure}[t]{.33\linewidth}
\centering
\includegraphics[width=\linewidth]{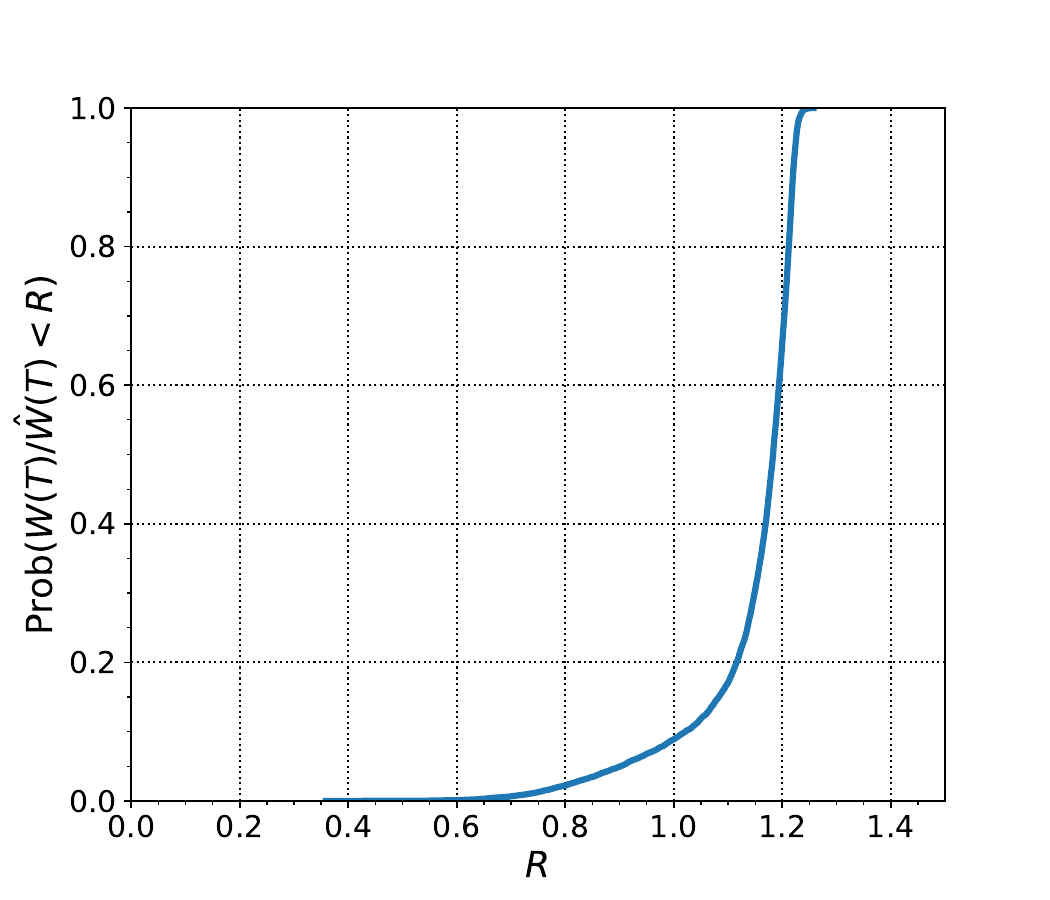}
\caption{Terminal wealth ratio}
\label{fig:cdf_wt_ratio}
\end{subfigure}
\begin{subfigure}[t]{.33\linewidth}
\centering
\includegraphics[width=\linewidth]{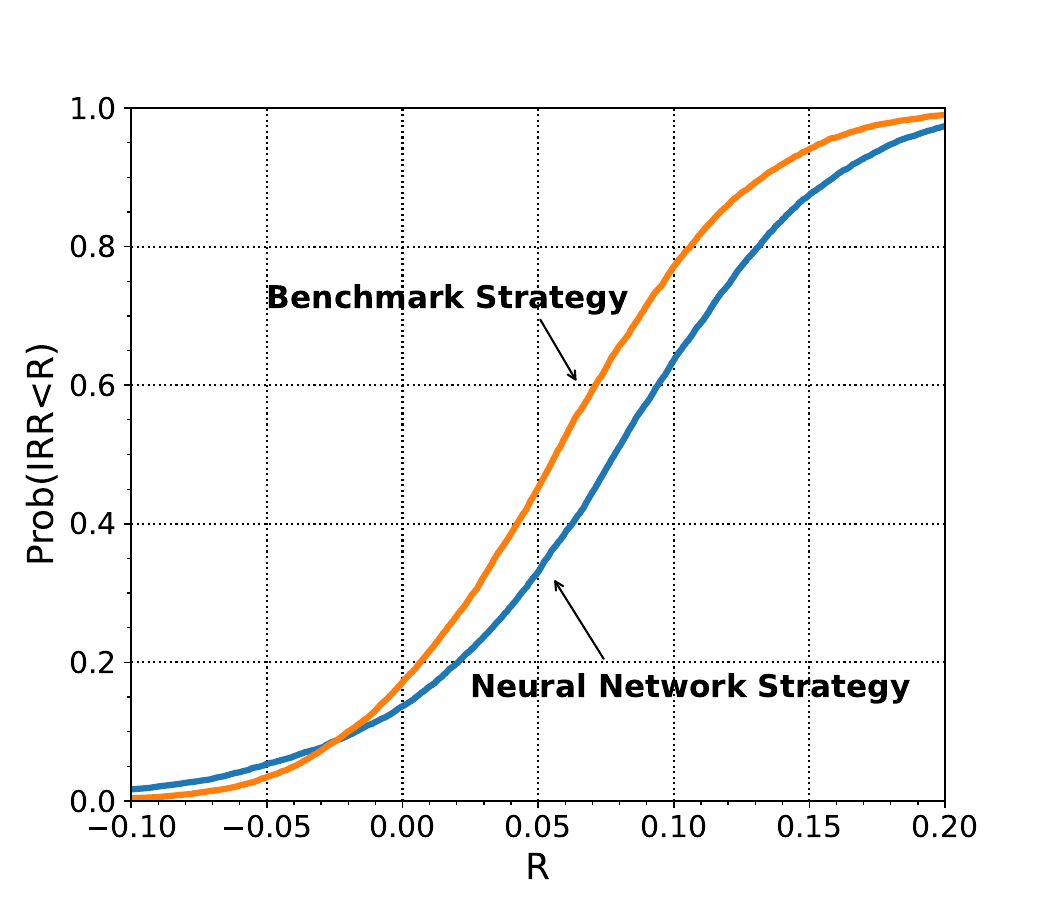}
\caption{CDF of IRR}
\label{fig:cdf_irr}
\end{subfigure}
}
\caption{
Percentiles of wealth ratio $\frac{W(t)}{\hat{W}(t)}$ and CDF of terminal wealth ratio $\frac{W(T)}{\hat{W}(T)}$ and internal rate of return (IRR). Results are based on the evaluation of the learned neural network model on $\boldsymbol{Y}^{test}$.
}
\label{fig_main_results}
\end{figure}

\begin{table}[hbt!]
\begin{center}
\begin{tabular}{lccccc} \hline
Strategy &  Median[$W_T$] & E[$W_T$] &   std[$W_T$] & 5th Percentile  & Median IRR (annual)    \\ \hline
 Neural network & 364.2        & 403.4    &  211.8    & 136.3  & 0.078\\
 Benchmark     & 308.5        & 342.9    & 165.0    & 149.0 & 0.056 \\ \hline
\end{tabular}
\caption{Statistics of strategies. Results are based on the evaluation results on the testing data set.}
\label{tb:main_results}
\end{center}
\end{table}

The analysis of Figure \ref{fig:wealth_ratio} reveals that the neural network strategy (the strategy following the training LFNN model) consistently outperforms the benchmark strategy in terms of the wealth ratio $W(t)/\hat{W}(t)$. Over time, both the mean and median wealth ratios demonstrate a smooth and consistent increase. Regarding tail performance (20th percentile), the neural network strategy initially falls behind the benchmark but gradually recovers and ultimately achieves 10\% greater wealth at the terminal time. This observation indicates that the neural network strategy effectively manages tail risk.

An additional metric that holds significant interest for managers is the distribution of the terminal wealth ratio $\frac{W(T)}{\hat{W}(T)}$. This metric examines the relative performance of the strategies at the end of the investment period. Figure \ref{fig:cdf_wt_ratio} illustrates that there is a greater than 90\% chance that the neural network strategy outperforms the benchmark strategy in terms of terminal wealth. This outcome is particularly noteworthy as the objective function (\ref{prob:CS_disc_opt}) does not directly target the terminal wealth ratio.

Given the constant cash injections in the portfolios, it is appropriate to employ the internal rate of return (IRR) as a measure of the portfolio's annualized performance. Figure \ref{fig:cdf_irr} demonstrates that the neural network strategy has a more than 90\% chance of producing a higher IRR. Furthermore, the median IRR of the neural network strategy exceeds that of the benchmark strategy by slightly over 2\%, aligning with the chosen target outperformance rate of $\beta=0.02$. This indicates that the neural network model consistently achieves the desired target performance across most outcomes.

The results from Table \ref{tb:main_results} indicate that the 5th percentile of the terminal wealth for the neural network strategy is lower than that of the benchmark strategy. This suggests that in some scenarios, particularly during persistent bear markets when stocks perform poorly, the neural network strategy may experience lower terminal wealth compared to the benchmark strategy. The neural network strategy takes on more risk by allocating a higher fraction of wealth to the equal-weighted stock index, which is considered a riskier asset, in comparison to the benchmark portfolio.

It's important to note, however, that these scenarios occur with low probability. As depicted in Figure \ref{fig:cdf_wt_ratio}, the neural network strategy exhibits a significantly high probability of outperforming the benchmark in terms of terminal wealth, exceeding 90\%. This implies that while there might be instances where the neural network strategy suffers relative to the benchmark, the overall performance is consistently strong, resulting in a high likelihood of achieving superior terminal value.

To gain insight into the strong performance of the neural network strategy, we further examine its allocation profile. 
\begin{figure}[htb]
\centering
\includegraphics[width=3.0in]{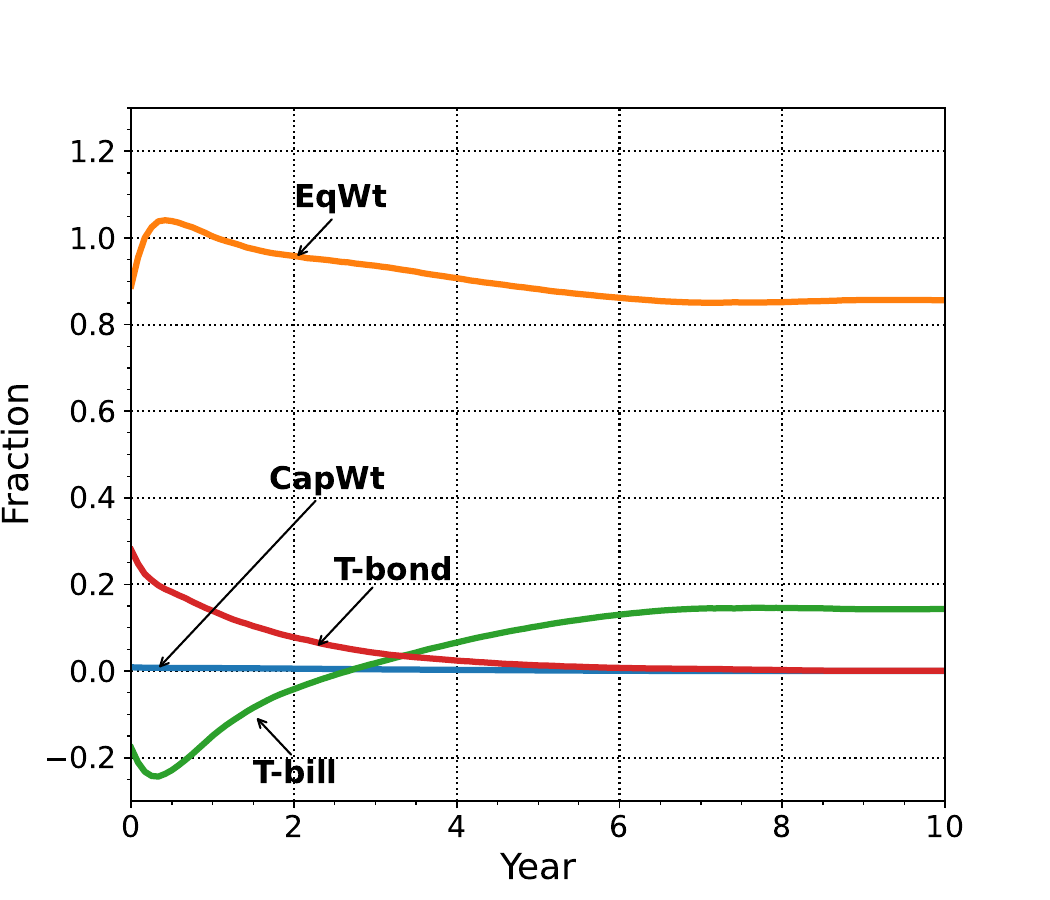}
\caption{Mean allocation fraction over time, evaluated on $\boldsymbol{Y}^{test}$}
\label{fig:mean_alloc}
\end{figure}

We begin by examining the mean allocation fraction for the four assets over time, as depicted in Figure \ref{fig:mean_alloc}. The first noteworthy observation from Figure \ref{fig:mean_alloc} is that, on average, the neural network strategy does not allocate wealth to the cap-weighted stock index. Initially, this might appear surprising; however, it aligns with historical data indicating significantly higher real returns for the equal-weighted stock index during periods of high inflation (refer to Appendix \ref{app:inflation_asset_return}). Given that the objective is to outperform a benchmark heavily invested in the equal-weighted index (70\%), it is logical to avoid allocating wealth to a comparatively weaker index in the active portfolio.

The second observation derived from Figure \ref{fig:mean_alloc} pertains to the evolution of mean bond allocation fractions. Initially, the neural network strategy shorts the 30-day T-bill index and assumes some leverage while heavily investing in the equal-weighted stock index during the first two years. This indicates a deliberate risk-taking approach early on to establish an advantage over the benchmark strategy. Subsequently, the allocation to the 10-year T-bond decreases, coinciding with the reduction in the allocation to the equal-weighted index. This suggests that the initial allocation to the T-bond was primarily for leveraging purposes, with the 10-year bond being the only defensive asset available. As leverage is no longer used in later years, the neural network strategy favors the T-bill over the 10-year bond.

Overall, despite the gradual decrease in stock allocation over time, the neural network strategy maintains an average allocation of more than 80\% to the equal-weighted stock index. This is expected, as outperforming an aggressive benchmark with a 70\% allocation to the equal-weighted stock index necessitates assuming higher levels of risk. Despite the higher allocation to riskier assets, the neural network strategy consistently delivers strong results compared to the benchmark strategy, as illustrated in Figure \ref{fig_main_results}.

Lastly, it is worth noting that the neural network strategy, trained under high-inflation regimes, exhibits remarkable performance on low-inflation testing datasets. This unexpected outcome highlights the robustness of the strategy. For further discussion on this topic, interested readers can refer to Appendix \ref{sec:low_inflation}.

\section{Conclusion}
{In this paper, our primary objective is to propose a framework that generates optimal dynamic allocation strategies under leverage constraints in order to outperform a benchmark during high inflation regimes. Imposing leverage-constraint in multi-period
asset allocation is consistent with the practice in large sovereign wealth funds, which often have exposures to alternative assets. Our proposed framework efficiently solves high-dimensional optimal control problems, accommodating diverse objective functions, constraints, and data sources.

We begin by assuming that both asset prices follow jump-diffusion models. Under this assumption, we derive a closed-form solution for a two-asset case using the cumulative tracking difference (CD) objective function. However, to obtain this closed-form solution, we need to make additional unrealistic assumptions such as continuous rebalancing, unlimited leverage, and continued trading in insolvency. Despite these assumptions, the closed-form solution provides valuable insights into the optimal control behavior. Notably, to track the elevated target, the optimal control needs to aim higher than the target when making allocation decisions.

To overcome the limitations of unrealistic assumptions and derive a more practical solution, we introduce a novel leverage-feasible neural network (LFNN) model. The LFNN model approximates the optimal control directly, eliminating the need for high-dimensional approximations of conditional expectations required in dynamic programming approaches. Additionally, the LFNN model converts the leverage-constrained optimization problem into an unconstrained optimization problem. Importantly, we justify the validity of the LFNN approach by mathematically proving that the solution to the parameterized unconstrained optimization problem can approximate the solution to the original constrained optimization problem with arbitrary precision.

To illustrate the effectiveness of our proposed approach, we conduct a case study on optimal asset allocation during high-inflation regimes. We apply the LFNN model to bootstrap resampled data from filtered historical high-inflation data. In our numerical experiment, we consider an investment case with four assets in high inflation regimes. The results consistently demonstrate that the neural network strategy outperforms the benchmark strategy throughout the investment period. Specifically, the neural network strategy achieves a 2\% higher median Internal Rate of Return (IRR) compared to the benchmark strategy and yields a higher terminal wealth with more than a 90\% probability. The allocation strategy derived from the LFNN model suggests that managers should favor the equal-weighted stock index over the cap-weighted stock index and short-term bonds over long-term bonds during high-inflation periods.}

\section{Acknowledgements}
Forsyth's work was supported by the Natural Sciences and Engineering Research Council of
Canada (NSERC) grant RGPIN-2017-03760.  Li's work was supported by he Natural Sciences and Engineering Research Council of
Canada (NSERC) grant RGPIN-2020-04331.

\section{Conflicts of interest}
The authors have no conflicts of interest to report.

\appendix

\section{Technical details of closed-form solution}
\subsection{Proof of Theorem (\ref{eq:PIDE})}\label{sec:proof_verification}
At any state $(t,w,\hat{w})\in[t_0,T]\times\mathbb{R}^2$, define the value function ${V}(w,\hat{w},t)$ to the CD problem (\ref{prob:CD_cont}) as
\begin{equation}
    {V}(t, w,\hat{w},\hat{\boldsymbol{\varrho}})) = \inf_{\boldsymbol{p}}\Big\{\mathbb{E}_{\boldsymbol{p}}\Bigg[\int_{t}^T\big(W(s)-e^{\beta s}\hat{W}(s)\big)^2ds\Big|W(t)=w,\hat{W}(t)=\hat{w}\Big]\Bigg\}.
\end{equation}
By the dynamic programming principle, we have
\begin{equation}
    {V}(t, w,\hat{w},\hat{\boldsymbol{\varrho}})=\inf_{\boldsymbol{p}}\Big\{\mathbb{E}_{\boldsymbol{p}}\Big[\Big({V}(t+\Delta t,W(t+\Delta t),\hat{W}(t+\Delta t),\hat{\boldsymbol{\varrho}})+\int_{t}^{t+\Delta t}\big(W(s)-e^{\beta s}\hat{W}(s)\big)^2ds\Big)\Big|W(t)=w,\hat{W}(t)=\hat{w}\Big]\Big\}\label{eq:dpp}
\end{equation}
Rearrange equation (\ref{eq:dpp}) to obtain
\begin{align}
    \inf_{\boldsymbol{p}}\Big\{\mathbb{E}_{\boldsymbol{p}}\Big[\Big(d{V}(t,w,\hat{w},\hat{\boldsymbol{\varrho}})+\int_{t}^{t+\Delta t}\big(W(s)-e^{\beta s}\hat{W}(s)\big)^2ds\Big)\Big|W(t)=w,\hat{W}(t)=\hat{w}\Big]\Big\}=0\label{eq:dpp_final}
\end{align}
Then, apply Itô's lemma with jumps \citep{cont2011nonparametric}, 
substitute $dW$ and $d\hat{W}$ terms with 
(\ref{dynamics_jump_diffusion}), and take limits as 
$\Delta t \downarrow0$, we obtain (\ref{eq:PIDE}).

The above results merely serve as an intuitive guide to obtain (\ref{eq:PIDE}). The formal proof of (\ref{eq:PIDE}) proceeds by using a suitably smooth test function, see for example \citep{oksendal2007applied}.

\subsection{Proof of results for CD-optimal control}\label{sec:obtain_opt_control}
In Section \ref{sec:closed-form}, we emphasized the dependence of $B$ and $D$ (defined in (\ref{def:B}) and (\ref{def:A_D})) on parameters $\beta$ and $c$ for understanding the optimal control function. As $\beta$ and $c$ are fixed parameters, in this proof, we omit the dependence of $B$ and $D$ on them for notational simplicity. 

The quadratic source term $\big(w-e^{\beta t}\hat{w}\big)^2$ in Theorem (\ref{eq:PIDE}) suggests the following {\it ansatz} for the value function $V$ in Theorem \ref{eq:PIDE} of the form
\begin{equation}
    V(t, w,\hat{w})=A(t)w^2+B(t)w+C(t)+\hat{A}(t)\hat{w}^2+\hat{B}(t)\hat{w}+D(t)w\hat{w},\label{def:ansatz}
\end{equation}
where $A,B,C,\hat{A},\hat{B},D$ are unknown deterministic functions of time $t$. If (\ref{def:ansatz}) is correct, then the pointwise infimum in (\ref{eq:PIDE}) is attained by $p^*$ satisfying the relationship
\begin{align}
    \Bigg(w\cdot\frac{\partial^2 V}{\partial w^2}\Bigg)\cdot p^* 
           = -\frac{1}{\gamma}\Bigg(\big(\mu_1-\mu_2\big)\cdot\frac{\partial V}{\partial w}+\big(\hat{\varrho}\gamma
   +\theta\big)\cdot\hat{w}\cdot\frac{\partial^2 V}{\partial w\hat{w}}
               +\theta\cdot w\cdot\frac{\partial^2 V}{\partial w^2}\Bigg),
    \label{eq:p_first_order_optimality}
\end{align}
assuming $A(t) > 0$.
Here $\gamma$ and $\theta$ are defined in (\ref{def:soln_params}). (\ref{def:ansatz}) implies that the relevant partial derivatives of $V$ are of the form
\begin{equation}
  \frac{\partial^2 V}{\partial w^2}=2A(t),\quad \frac{\partial V}{\partial w}=2A(t)w+B(t)+D(t)\hat{w},\quad\frac{\partial^2 V}{\partial w\hat{w}}=D(t).\label{eq:V_derivatives}
\end{equation}
Substituting (\ref{eq:V_derivatives}) into (\ref{eq:p_first_order_optimality}), the optimal control $p^*$ obtained is in the form of (\ref{def:optimal_control}), where $h$ and $g$ are given by (\ref{def:g_h}). Then, it only remains to determine the functions $A,B,D$. Substituting (\ref{eq:p_first_order_optimality}) into PIDE (\ref{eq:PIDE}), we can obtain the following ordinary differential equations (ODE) for $A,B,D$,
\begin{equation}
    \begin{cases}
      \frac{dA(t)}{dt}=-\Big(2\mu_2-\eta\Big)A(t)-1,\qquad A(T)=0,\\
      \frac{dD(t)}{dt}=-\Big(2\mu_2-\eta\Big)D(t)+2e^{\beta t},\qquad D(T)=0,\\
      \frac{dB(t)}{dt}=-(\mu_2-\phi)B(t)-2cA(t)-cD(t),\qquad B(T)=0,
    \end{cases}\label{eq:ode}
\end{equation}
Solving the ODE system gives us the $A,B,D$ defined in (\ref{def:A_D})
and (\ref{def:B}).  We also note that
$A(t) > 0$,  thus completing the proof.

\subsection{Proof of Corollary (\ref{coro:prop_g}) and (\ref{coro:prop_h}) }\label{sec:proof-g-h}
\citet{van2022dynamic} derive the CD-optimal control under the assumption that the stock price follows the double-exponential jump-diffusion model and the bond is risk-free with the bond price $B(t)$ following
\begin{equation}
    \frac{dB(t)}{B(t)}=r.\label{def:riskfree_bond}
\end{equation}
Under such, assumptions,
\citet{van2022dynamic} shows that the CD-optimal control can be expressed in a similar form as in (\ref{def:optimal_control}) with $g$ and $h$ functions. The $g$ and $h$ functions satisfy the same properties as in Corollary (\ref{coro:prop_g}) and (\ref{coro:prop_h}). 
Despite the
fact that we assume the bond price follows 
a jump-diffusion model, the proof 
of Corollary (\ref{coro:prop_g}) and (\ref{coro:prop_h}) 
follows similar steps as the proof in \citet{van2022dynamic}.

{
\section{Technical details of LFNN model}
\subsection{Proof of Theorem \ref{theorem:feasibility_domain}}\label{app:proof_feasibility_dom}
\textbf{Theorem \ref{theorem:feasibility_domain}.}
(Unconstrained feasibility domain) The feasibility domain $\mathcal{Z}_{\boldsymbol{\theta}}$ defined in (\ref{def:Z_theta}) associated with the LFNN model (\ref{LFNN}) is $\mathbb{R}^{N_{\boldsymbol{\theta}}}$.
\begin{proof}
First, it is obvious that $\mathcal{Z}_{\boldsymbol{\theta}}\subseteq\mathbb{R}^{N_{\boldsymbol{\theta}}}$ by definition of (\ref{def:Z_theta}). Next, we show that $\mathbb{R}^{N_{\boldsymbol{\theta}}}\subseteq\mathcal{Z}_{\boldsymbol{\theta}}$. To prove this, we need to show that for any $\boldsymbol{\theta}\in\mathbb{R}^{N_{\boldsymbol{\theta}}}$,
{
\begin{equation}
    f(x;\boldsymbol{\theta})={p}\in
\begin{cases}
\mathcal{Z}_1,\;\text{if }x\in\mathcal{X}_1,\\\mathcal{Z}_2,\;\text{if }x\in\mathcal{X}_2,
\end{cases}\forall x\in\mathcal{X}.
    \label{eq:Z_theta_contain_R}
\end{equation}
}Here $f$ is the LFNN function defined in (\ref{LFNN}), ${p}=(p_1,\cdots,p_{N_a})^\top\in\mathbb{R}^{N_a}$ is the output of the LFNN model that represents the wealth allocation to the assets, $\mathcal{Z}$ is the feasibility domain defined in (\ref{control_space_lev}), and $x=\big(t,W(t),\hat{W}(t)\big)^\top\in\mathcal{X}$ is a feature vector. {To prove (\ref{eq:Z_theta_contain_R}), we verify the two scenarios ($x\in\mathcal{X}_1$ and $x\in\mathcal{X}_2$) separately.

When $x\in\mathcal{X}_2$, it is easily verifiable that $p=\boldsymbol{e}_{N_l+1}$ via the definition of the leverage-feasible activation function (\ref{psi_lev_feas_func}).

Next, we verify that when $x\in\mathcal{X}_1$, ${p}\in\mathcal{Z}_1$. To prove this, we need to show that constraints of (\ref{cons:long-only})-(\ref{constraint:simult_short}) are satisfied when $x\in\mathcal{X}_1$. }

By definition of (\ref{psi_lev_feas_func}), it is obvious 
that the long-only constraint (\ref{cons:long-only}) holds for long-only assets. 

It is also easy to verify that the summation constraint (\ref{cons:sum-one}) is satisfied. This can be observed after the fact that
\begin{equation}
    \sum_{i=1}^{N_l}p_i=l,\quad\text{and}\quad\sum_{i=N_l+1}^{N_a}p_i=1-l.
\end{equation}

The maximum leverage constraint (\ref{cons:max-lev}) is also satisfied, as 
\begin{equation}
    \sum_{i=1}^{N_l}p_i=l = p_{max}\cdot\text{Sigmoid}(-o_{N_a+1}) \leq p_{max}.
\end{equation}

Finally, the simultaneous shorting constraint (\ref{as:simult_short}) is satisfied. 
To see this, we examine the scenario when leverage occurs, 
i.e., $\sum_{i=1}^{N_l}p_i=l>1$. Then, by definition from (\ref{psi_lev_feas_func}), we know
\begin{equation}
    p_i = (1-l)\cdot\frac{e^{o_i}}{\sum_{k=N+1}^{N_a}e^{o_k}} \leq 0,\;\forall i\in\{N_l+1,\cdots,N_a\}
    \label{short_case_appendix}
\end{equation}
From (\ref{short_case_appendix}) it is clear that if $l \leq 1$, then $p_i \geq 0, \forall i$.

Therefore, for any $\boldsymbol{\theta}\in\mathbb{R}^{N_{\boldsymbol{\theta}}}$, (\ref{eq:Z_theta_contain_R}) is satisfied. This implies $\mathbb{R}^{N_{\boldsymbol{\theta}}}\subseteq\mathcal{Z}_{\boldsymbol{\theta}}$.

\end{proof}

\subsection{Proof of Lemma \ref{lemma:decomp_p_opt} and Theorem \ref{theorem:approximation_opt_control}}\label{app:validity_LFNN}
\textbf{Lemma \ref{lemma:decomp_p_opt}.} (Structure of feasible control) 
Any feasible control function $p:\mathcal{X}\mapsto {\mathcal{Z}}$, where $\mathcal{Z}$ is defined in \eqref{admissible_strat_lev}, has the function decomposition

\begin{equation}
    p(x)=\varphi(\omega(x),x),
\end{equation}
where
$\varphi:\Tilde{\mathcal{Z}}\times\mathcal{X}\mapsto{\mathcal{Z}}$ is defined in (\ref{eq:psi_decomp}), i.e.
\begin{equation}
\varphi(z)=\Big(z_{N_a+1}\cdot (z_1,\cdots,z_{N_l}),(1-z_{N_a+1})\cdot(z_{N_l+1},\cdots,z_{N_a})\Big)^\top\cdot\textbf{1}_{\boldsymbol{x}\in\mathcal{X}_1}+\boldsymbol{e}_{N_l+1}\cdot\textbf{1}_{\boldsymbol{x}\in\mathcal{X}_2},
\end{equation}
and $\omega:\mathcal{X}\mapsto\Tilde{\mathcal{Z}}$. Here
\begin{equation}
 \Tilde{\mathcal{Z}}=\Bigg\{z\in\mathbb{R}^{N_a+1},\sum_{i=1}^{N_l}z_i=1,\sum_{i=N_l+1}^{N_a}z_i=1,z_{N_a+1}\leq p_{max}, z_i\geq0, \forall i\Bigg\}.   
\end{equation}

\begin{proof}
We prove the lemma by existence.

Define $\omega$ as
\begin{equation}
    \omega(x)=
    \begin{cases}
    \phi\big(p(x)\big),\quad\qquad\qquad\qquad\qquad\qquad\qquad\text{if }x\in\mathcal{X}_1,\\
    \big(\frac{1}{N_l},\cdots,\frac{1}{N_l},\frac{1}{N_a-N_l},\cdots,\frac{1}{N_a-N_l},0\big)^\top,\;\;\text{if }x\in\mathcal{X}_2,
    \end{cases}
\end{equation}

where   
    for $\forall{z}=(z_1,\cdots,z_{N_a})^\top\in\mathcal{Z}_1$,  $y=\phi(z){\in\mathbb{R}^{N_a+1}}$ is as defined below
\begin{equation}
\phi(z) \equiv y =
\begin{cases}
\begin{cases}
    y_i = \frac{z_i}{\sum_{j=1}^{N_l}z_j},\; i\in\{1,\cdots,N_l\},\\
    y_i = \frac{z_i}{1-\sum_{j=1}^{N_l}z_j},\; i\in\{N_l,\cdots,N_a\},\\
    y_{N_a+1}=\sum_{j=1}^{N_l}z_j,
\end{cases} \quad\text{if $\sum_{i=1}^{N_l}z_i\in(0,1)\cup(1,p_{max}]$,}\\
\begin{cases}
    y_i = z_i,\; i\in\{1,\cdots,N_l\},\\
    y_i = 1/(N_a-N_l),\; i\in\{N_l,\cdots,N_a\},\\
    y_{N_a+1}=1,
\end{cases} \text{if $\sum_{i=1}^{N_l}z_i=1$,}\\
\begin{cases}
    y_i = 0,\; i\in\{1,\cdots,N_l\},\\
    y_i = z_i,\; i\in\{N_l,\cdots,N_a\},\\
    y_{N_a+1}=0,
\end{cases} \qquad\qquad\;\;\text{if $\sum_{i=1}^{N_l}z_i=0$,}\\
\end{cases}
\end{equation}
It can then be easily verified that $\omega:\mathcal{X}\mapsto\Tilde{\mathcal{Z}}$, and that $p(x)=\varphi(\omega(x),x)$.

\end{proof}


\begin{lemma}\label{lemma:general_approx} \textup{(Approximation of controls with a specific structure)} Assume a control function $p:\mathcal{X}\mapsto\mathcal{Z}$ has the structure 
\begin{equation}\label{eq:general_func}
    p(x)=\Phi(\Omega(x),x), x\in\mathcal{X},
\end{equation}
where $\mathcal{X}$ is compact, $\Omega\in C(\mathcal{X},\mathcal{Y})$, {i.e. $\Omega$ is a continuous mapping from $\mathcal{X}$ to $\mathcal{Y}$,} and $\Phi:\mathcal{Y}\times\mathcal{X}\mapsto\mathcal{Z}$ is Lipschitz continuous on $\mathcal{Y}\times\mathcal{X}_i$, $\forall i=1,\cdots,n$, where $\{\mathcal{X}_i,i=1,\cdots,n\}$ is a partition of $\mathcal{X}$, i.e. 
\begin{equation}
    \begin{cases}
    \bigcup_{i=1}^n \mathcal{X}_i=\mathcal{X},\\
    \mathcal{X}_i\bigcap \mathcal{X}_j=\varnothing,\forall 1\leq i,j\leq n.
    \end{cases}
\end{equation}
If $\exists m\in\mathbb{N}$ and $\Upsilon:\mathbb{R}^m\mapsto\mathcal{Y}$ such that
\begin{enumerate}[label=(\roman*)]
    \item $\Upsilon$ has a continuous right inverse on $Im(\Upsilon)$.
    \item $Im(\Upsilon)$ is dense in $\mathcal{Y}$, then $\forall\epsilon>0$.
    \item $\partial Im(\Upsilon)$ is collared.
\end{enumerate}
Then there exists a choice of $N_h$ and $\boldsymbol{\theta}$ such that the fully connected feedforward neural network function $\Tilde{f}(\cdot;\boldsymbol{\theta})$ defined in (\ref{def:FNN}) satisfies
\begin{equation}
    \sup_{x\in\mathcal{X}}\|\Phi\Big(\Upsilon\big(\Tilde{f}(x;\boldsymbol{\theta})\big),x\Big)-p(x)\|<\epsilon.
\end{equation}
\end{lemma}
\begin{proof}

Let 
\begin{equation}
L_\Phi=\max_{1\leq i\leq n}L_i,
\end{equation}
where $L_i$ is the Lipschitz constant for $\Phi$ on $\mathcal{Y}\times\mathcal{X}_i$.

Since  $\Omega\in C(\mathcal{X})$ in compact $\mathcal{X}$,  following \citet{kratsios2020non}, we know that $\forall {\epsilon},$, there exists $N_h\in\mathbb{N}$ and $\boldsymbol{\theta}\in\mathbb{R}^{N_{\boldsymbol{\theta}}}$ such that the corresponding FNN $\Tilde{f}(\cdot;\boldsymbol{\theta}):\mathcal{X}\mapsto\mathbb{R}^m$ defined in (\ref{def:FNN}) satisfies
\begin{equation}
    \sup_{x\in\mathcal{X}}\|\Upsilon\big(\Tilde{f}(x;\boldsymbol{\theta})\big)-\Omega(x)\| <\epsilon/{L_\Phi},
\end{equation}
Then
\begin{align}
\sup_{x\in\mathcal{X}}\|\Phi\Big(\Upsilon\big(\Tilde{f}(x;\boldsymbol{\theta})\big),x\Big)-p(x)\| &= \sup_{1\leq i\leq n}\sup_{x\in\mathcal{X}_i}\|\Phi\Big(\Upsilon\big(\Tilde{f}(x;\boldsymbol{\theta})\big),x\Big)-\Phi\Big(\Omega(x),x\Big)\|\\
&\leq\sup_{1\leq i\leq n}\sup_{x\in\mathcal{X}_i}L_i\cdot\Big(\|\Upsilon\big(\Tilde{f}(x;\boldsymbol{\theta})\big)-\Omega(x)\|\Big)\\
&<\sup_{1\leq i\leq n}\frac{L_i}{L_\Phi}\epsilon\\
&\leq\epsilon.
\end{align}

\end{proof}

\begin{remark}\textup{(Remark on Lemma \ref{lemma:general_approx}) Normally, the universal approximation theorem only applies to the approximation of continuous functions defined on a compact set \citep{hornik1991approximation}. Lemma \ref{lemma:general_approx} extends the universal approximation theorem to a broader class of functions that have the structure of (\ref{eq:general_func}). Furthermore, Lemma \ref{lemma:general_approx} provides guidance on constructing neural network functions that handle stochastic constraints on controls which are usually difficult to address in stochastic optimal control problems. Consider the following example: the control $p:\mathcal{X}\mapsto\mathbb{R}^{N_a}$ has stochastic constraints such that $p(\boldsymbol{x})\in[a(\boldsymbol{x}),b(\boldsymbol{x})]$ where $a,b:\mathcal{X}\mapsto\mathbb{R}^{N_a}$ are deterministic functions. This is a common setting in portfolio optimization problems in which allocation fractions to specific assets are subject to thresholds tied to the performance of the portfolio. With Lemma \ref{lemma:general_approx}, with a bit of engineering, one can easily construct a $\Phi$ so that the corresponding neural network satisfies the constraints naturally and be guaranteed that such a neural network can approximate the control well.}
\end{remark}

We then proceed to prove Theorem \ref{theorem:approximation_opt_control}.\\
\textbf{Theorem \ref{theorem:approximation_opt_control}.} (Approximation of optimal control) Following Assumption \ref{as:approx_control}, $\forall\epsilon>0$, there exists $N_h\in\mathbb{N}$, and $\boldsymbol{\theta}\in\mathbb{R}^{N_{\boldsymbol{\theta}}}$ such that the corresponding LFNN model $f(\cdot;\boldsymbol{\theta})$ described in (\ref{LFNN}) satisfies the following:
\begin{equation}
    \sup_{x\in\mathcal{X}}\|f(x;\boldsymbol{\theta})-p^*(x)\| <\epsilon.
\end{equation}
{
\begin{proof}
From (\ref{LFNN}) and Lemma \ref{lemma:property_psi}, we know that 
\begin{equation}
    f({x};\boldsymbol{\theta})=\psi\big(\Tilde{f}({x};\boldsymbol{\theta}),{x}\big)=\varphi\Big(\zeta\big(\Tilde{f}({x};\boldsymbol{\theta})\big),{x}\Big),
\end{equation}
where $\Tilde{f}$ is the FNN defined in (\ref{def:FNN}) and $\varphi:\Tilde{\mathcal{Z}}\times{\mathcal{X}}\mapsto\mathbb{R}^{N_a},\zeta:\mathbb{R}^{N_a+1}\mapsto \Tilde{\mathcal{Z}}$ are defined in (\ref{eq:psi_decomp}).

It can be easily verified that $\zeta$ satisfies the following:
\begin{enumerate}[label=(\roman*)]
\item $\zeta$ has a continuous right inverse, e.g.
\begin{equation}
\zeta^{-1}(z):Im(\zeta)\mapsto\mathbb{R}^{N_a+1}, \zeta^{-1}(z)=\Bigg(\log(z_1),\cdots,\log(z_{N_a}),\sigma^{-1}(z_{N_a+1}/p_{max})\Bigg)^\top, \end{equation}
where $\sigma^{-1}$ is the inverse function of the sigmoid function.

\item $Im(\zeta)$ is dense in $\Tilde{\mathcal{Z}}$. This is because $\overline{Im(\zeta)}$, the closure of $Im(\zeta)$, is $\Tilde{\mathcal{Z}}$.
\item $\partial Im(\zeta)$ is collared \citep{brown1962locally,connelly1971new,baillif2022collared}.
\end{enumerate}

Furthermore, consider the partition of $\mathcal{X}$, $\big\{\mathcal{X}_1,\mathcal{X}_2\big\}$, which is defined in Definition \ref{def:partition}. It is easily verifiable that $\varphi$ is Lipschitz continuous on $\Tilde{\mathcal{Z}}\times\mathcal{X}_1$ and $\Tilde{\mathcal{Z}}\times\mathcal{X}_2$ respectively.

Finally, according to Assumption \ref{as:approx_control}, $p^*(x)=\varphi\big(\omega^*(x),x\big)$, where $\omega^*\in C(\mathcal{X},\Tilde{\mathcal{Z}})$.

Applying Lemma \ref{lemma:general_approx} with $\mathcal{Y}=\Tilde{\mathcal{Z}}, \Omega(\cdot) =\omega^*(\cdot)$, $\Upsilon(\cdot)=\zeta(\cdot)$, and $\Phi (\cdot,\cdot)=\varphi(\cdot,\cdot)$, we know that there exists $N_h\in\mathbb{N}$, and $\boldsymbol{\theta}\in\mathbb{R}^{N_{\boldsymbol{\theta}}}$ such that the corresponding LFNN model $f({x};\boldsymbol{\theta})=\varphi\Big(\zeta\big(\Tilde{f}({x};\boldsymbol{\theta})\big),{x}\Big)$ satisfies the following:
\begin{equation}
    \sup_{x\in\mathcal{X}}\|f(x;\boldsymbol{\theta})-p^*(x)\| <\epsilon.
\end{equation}

\end{proof}
}

\section{Comparing LFNN with closed-form solution}\label{sec:validate_NN}
In this section, we {compare the performance of the strategy following the learned shallow LFNN model (which we refer to as the ``neural network strategy'' from now on)} with the closed-form solution (\ref{def:optimal_control}), and provide empirical validation of the LFNN approach. 
\subsection{Approximate form under realistic assumptions}\label{sec:approx_form}
We first note that the closed-form solution $p^*$ defined in (\ref{def:optimal_control}) is obtained under several unrealistic {assumptions}, namely continuous rebalancing, unlimited leverage, 
and continuing  trading in insolvency.\footnote{Note that we consider a two-asset scenario here, thus the scalar $p^*\in\mathbb{R}$ (allocation fraction for the stock index) fully describes the allocation strategy $\boldsymbol{p}^*$, since $\boldsymbol{p}^*=(p^*,1-p^*)^\top$.} 
In practice, 
investors have constraints such as discrete rebalancing, limited leverage, and no trading when insolvent. {For a meaningful comparison,} instead of comparing the neural network strategy with the closed-form solution $p^*$ directly, we compare the 
neural network strategy with {an easily obtainable approximation to} 
the closed-form solution which satisfies realistic constraints.

In particular, we consider an equally-spaced discrete rebalancing schedule $\mathcal{T}_{\Delta t}$ defined as
\begin{equation}
    \mathcal{T}_{\Delta t}=\Big\{t_i:\;i=0,\cdots,N\Big\},\label{def:equal_disc_schedule}
\end{equation}
where $t_i=i\Delta t$, and $\Delta t=T/N$. Then, the {\it clipped form} $\bar{p}_{\Delta t}:\mathcal{T}_{\Delta t}\times\mathbb{R}^3\mapsto\mathbb{R}$ is defined as
\begin{equation}(\text{Clipped form}):\quad
    \bar{p}_{\Delta t}(t_i,\bar{W}_{\Delta t}(t_i),\hat{W}_{\Delta t}(t_i),\hat{\varrho}) = \min\Bigg(\max\Big(p^*(t_i,\bar{W}_{\Delta t}(t_i),\hat{W}_{\Delta t}(t_i),\hat{\varrho}),p_{min}\Big), p_{max}\Bigg).\label{def:clipped_form}
\end{equation}
Here $[p_{min},p_{max}]$,  where $p_{min}=0$  and $p_{max} \geq 1$,  is the allowed range, $\bar{W}_{\Delta t}(t_i)$ is the wealth of the active portfolio at $t_i$ following $\bar{p}_{\Delta t}$ from $t_0$ to $t_i$, $\hat{W}_{\Delta t}(t_i)$ is the wealth of the benchmark portfolio at $t_i$ following the fixed-mix strategy described by constant allocation fraction $\hat{\varrho}$, but only rebalanced discretely according to $\mathcal{T}_{\Delta t}$. Clearly, the allocation strategy from $\bar{p}_{\Delta t}$ follows the discrete schedule of $\mathcal{T}_{\Delta t}$, and satisfies the leverage constraint that $\bar{p}_{\Delta t}\in[p_{min},p_{max}]$. $p_{\Delta t}$ {approaches the closed-form solution $p^*$ as $\Delta t\downarrow0,p_{min}\downarrow-\infty$ and $p_{max}\uparrow\infty$.} 
We note that a similar clipping idea is explored in \citet{vigna2014efficiency} in the context of closed-form solutions for multi-period mean-variance asset allocation.  However, it should be emphasized that the clipped form $\bar{p}_{\Delta t}$ with finite $(p_{min}, p_{max})$ is a feasible, but in general sub-optimal, control of the leverage-constrained CD problem (\ref{prob:CD_disc}).

We then address the assumption that trading continues when insolvent, i.e., when the wealth of the portfolio reaches zero. While necessary for the mathematical derivation of the closed-form solution, we acknowledge that this is by no means reasonable for practitioners. Under the continuous rebalancing case (no jumps), if the control (allocation) is bounded, it is shown that the wealth of the portfolio can never be negative \citep{wang2012comparison}. 
However, with discrete rebalancing, even with a bounded control, as long as the upper bound $p_{max}>1$, it is theoretically possible that the portfolio value becomes negative. We address this assumption by applying an overlay on strategies so that in the case of insolvency, 
we assume the manager liquidates the {long-only positions and allocates the debt (negative wealth) to a shortable (bond) 
asset (consistent with Assumption \ref{as:no_trading_insolvency})} to allow outstanding debt to 
accumulate until the end of the investment horizon. Going forward, when we refer to any strategy 
(e.g. neural network strategy, clipped form), 
we mean the strategy with this overlay applied.
We remark that in practice, this overlay has little effect. In numerical experiments with 
10,000 samples of observed wealth trajectories (based on calibrated jump-diffusion model or bootstrap resampled data paths), we do not observe any single wealth trajectory that ever hits negative wealth for any strategy (e.g., neural network strategy, clipped form, etc).

In summary, the clipped form satisfies the realistic constraints and is a comparable benchmark for the neural network strategy. In the following section, we will numerically compare the performance of the clipped form, the neural network strategy, and the closed-form solution.

\subsection{Comparison: LFNN  strategy vs clipped-form solution}
{We assess and compare} the performance of the neural network strategy and the clipped form, we assume the following investment scenario described in Table \ref{tb:validation_case}. 

\begin{table}[htb]
\begin{center}
\begin{tabular}{lc} \toprule
Investment horizon $T$ (years) & 10  \\
Assets & CRSP cap-weighted index (real)/30-day T-bill (U.S.) (real) \\
Index Samples& Concatenated 1940:8-1951:7, 1968:9-1985:10\\
Initial portfolio wealth/annual cash injection  & 100/10 \\
Rebalancing frequency & Monthly, quarterly, semi-annually, annually\\
Maximum leverage & 1.3\\
Benchmark equity percentage & 0.7\\
Outperformance target rate $\beta$ & 1\% (100 bps)\\
\bottomrule
\end{tabular}
\caption{Investment scenario. 
\label{tb:validation_case}}
\end{center}
\end{table}

We assume the stock index and the bond index prices follow a double exponential jump model  (\ref{model:jump_diffusion}), see e.g., \citep{kou2002jump,kou2004option}, i.e., for the jump variable $\xi_i$, $y_i=\log(\xi_i)$ follows the double exponential distribution with density functions $g_i(y_i)$ defined as follows
\begin{equation}
    g_i(y_i)=\nu_i\iota_ie^{-\iota_iy_i}\textbf{1}_{y_i\geq0}+(1-\nu_i)\varsigma_ie^{-\varsigma_iy_i}\textbf{1}_{y_i<0},\;i=1,2.\label{def:double_exponential}
\end{equation}
where $\nu_i$ is the probability for an upward jump, and $\iota_i$ and $\varsigma_i$ are parameters that describe the upward jump and downward jump respectively. The double exponential jump-diffusion model allows the flexibility of modeling asymmetric upward and downward jumps in asset prices, which seems an appropriate assumption for inflation regimes.\footnote{We remind the reader that the closed-form solution is derived under the jump-diffusion model.}

Using the threshold technique \citep{mancini2009non,cont2011nonparametric,dang2016better}, we calibrate the double exponential jump-diffusion models to the historical high-inflation periods described in Section \ref{sec:filtering_regimes}. The calibrated parameters can be found in Appendix \ref{app:calibrated_parans}. Then, we construct  a training data set $\boldsymbol{Y}$ and a testing data set $\boldsymbol{Y}^{test}$ by sampling the 
calibrated model, each with 10,000 samples.

The neural network strategy follows the LFNN model learned from $\boldsymbol{Y}$. We then evaluate the performance
of the neural network strategy and the approximate form (\ref{def:clipped_form}) on the testing data set $\boldsymbol{Y}^{test}$. Specifically, we compare the value of the CD objective function (\ref{prob:CD_disc}) for the neural network strategy and the clipped form on $\boldsymbol{Y}^{test}$. In particular, this training/testing process is repeated for various rebalancing frequencies from monthly to annually, as described in Table \ref{tb:validation_case}. 
\begin{table}[htb]
\centering
\begin{tabular}{l c c c c c}
\hline
\multicolumn{5}{c}{Closed-form solution objective function value: 418 }   \\ \hline
   Strategy & $\Delta t=1$ & $\Delta t=1/2$  & $\Delta t=1/4$ & $\Delta t=1/12$  & $\Delta t=0$  \\ \hline
Clipped form  & 545 & 504 & 479 & 467 & 461 (extrapolated) \\ \hline
Neural network & 537 & 498 &  476 & 464 & 458 (extrapolated) \\ \hline
\end{tabular}
\caption{CD objective function values. Results shown are evaluated on $\boldsymbol{Y}^{test}$, the lower the better.}\label{tb:benchmark}
\end{table}

In Table \ref{tb:benchmark}, we can see that the neural network strategy consistently outperforms the clipped form in terms of the objective function value for all rebalancing frequencies. From Table 4.2 we can see that the objective function values of both the neural network strategy and the clipped form converge at roughly a first-order rate as $\Delta t\downarrow0$. Assuming this to be true, we extrapolate the solution to $\Delta t=0$ using Richardson extrapolation.
These extrapolated values are estimates of the exact value of the continuous-time CD objective function (\ref{prob:CD_cont}) for the clipped form and the neural network strategy. We can see that the neural network strategy still outperforms the clipped form in terms of the extrapolated objective function value.  
We can also see that the extrapolated
neural network objective function value is lower than the (suboptimal) clipped form extrapolated value, but,
of course, larger than the unconstrained closed-form solution. 

\begin{figure}[htb]
\centerline{%
\begin{subfigure}[t]{.45\linewidth}
\centering
\includegraphics[width=\linewidth]{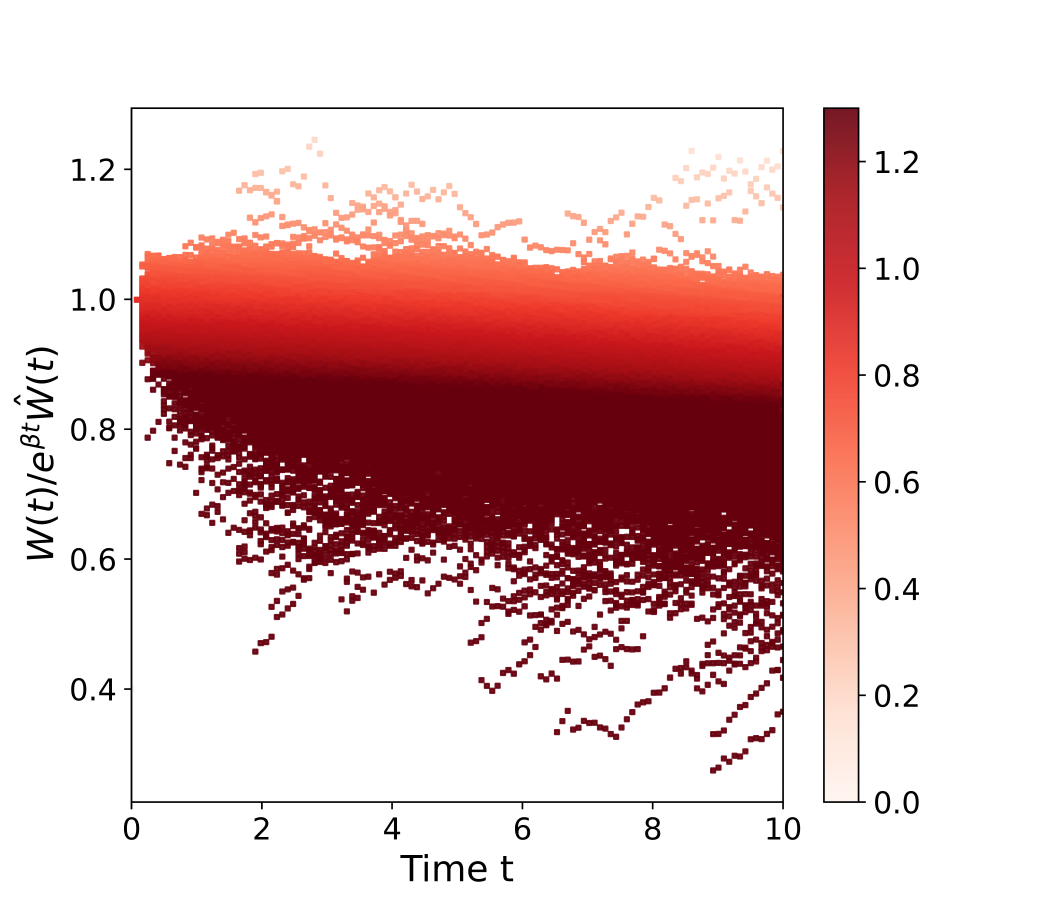}
\caption{Clipped form}
\label{fig:heatmap_clipped_validate}
\end{subfigure}
\begin{subfigure}[t]{.45\linewidth}
\centering
\includegraphics[width=\linewidth]{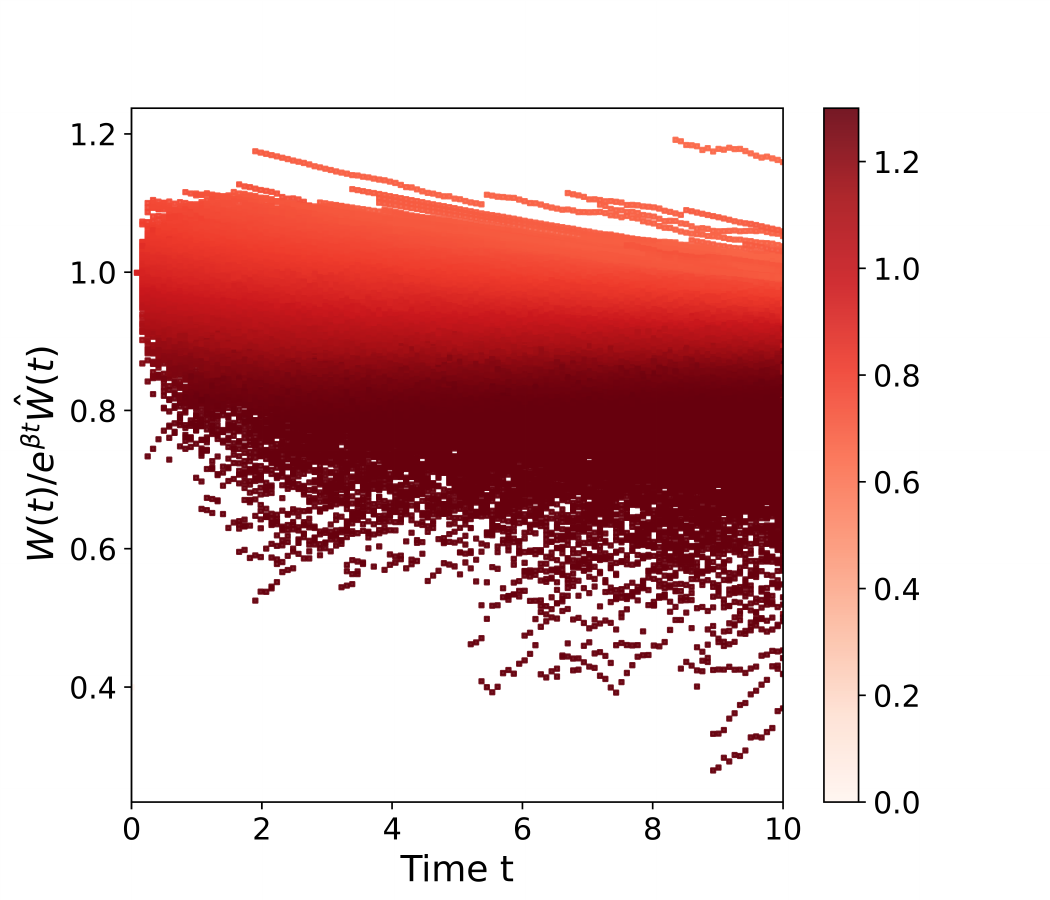}
\caption{Neural network strategy}
\label{fig:heatmap_nn_validate}
\end{subfigure}
}
\caption{
Stock allocation fraction w.r.t. tracking ratio $W(t)/\Big(e^{\beta t}\hat{W}(t)\Big)$ and time $t$. Results are based on the evaluation on the testing data set $\boldsymbol{Y}^{test}$, and monthly rebalancing (i.e. $\Delta t=1/12$).
}
\label{fig:heatmap_nn_vs_clipped}
\end{figure}

Finally, we compare the neural network allocation strategy with the clipped form strategy. Specifically, in Figure \ref{fig:heatmap_nn_vs_clipped}, we consider the case of monthly rebalancing and present the scatter plots of the allocation fraction 
in the stock index with respect to time $t$ and the ratio between the wealth of the active portfolio $W(t)$ and the elevated target $e^{\beta t}\hat{W}(t)$. For simplicity, we call this ratio the ``tracking ratio''. We plot the 3-tuple $\Big(\frac{W(t)}{e^{\beta t}\hat{W}(t)}, t, p_1(W(t),\hat{W}(t),t)\Big)$ (obtained from the evaluation of the strategies on samples from $\boldsymbol{Y}^{test}$) by using time $t$ as the {x-axis, the tracking ratio $\frac{W(t)}{e^{\beta t}\hat{W}(t)}$ as the y-axis,} and the values of the corresponding allocation fraction to the cap-weighted index $p_1(W(t),\hat{W}(t),t)$ to color the scattered dots on the plot. A darker shade of the color indicates a higher allocation fraction.

As we can see from Figure \ref{fig:heatmap_nn_vs_clipped}, the stock allocation fraction of the neural network strategy behaves similarly to the stock allocation fraction from the clipped form. Both strategies invest more wealth in the stock when the tracking ratio is lower, which is consistent with the insights we obtained in Section \ref{sec:insights}. In addition, the transition patterns of the allocation fraction of the two strategies are also highly similar. One can almost draw an imaginary horizontal dividing line around $\frac{W(t)}{e^{\beta t}\hat{W}(t)}=0.9$ that separates high stock allocation and low stock allocation for both strategies. 

We remark that a common criticism towards the use of neural networks is about the lack of interpretability compared to more interpretable counterparts such as the regression models \citep{rudin2019stop}. In this section, we see that the neural network strategy closely resembles the closed-form solution for the CD objective. The closed-form solution, in turn, complements the neural network model and offers an alternative way of interpreting results obtained from the neural network.

\section{Moving-window inflation filter}
\subsection{Filtering algorithm}\label{sec:window_algo}
Algorithm \ref{Algo:filter} presents the pseudocode for the moving-window filtering algorithm.
\begin{algorithm}[htp]
\SetAlgoLined
\KwData{}
$~~~$CPI[i]; $i=1,\ldots,N$ \tcc{CPI Index}
$~~~$Cutoff \tcc{High inflation cutoff: annualized}
$~~~$$\Delta t$ \tcc{CPI index time interval}
$~~~$$K$ \tcc{smoothing window size}
\KwResult{Flag[i]; $ i=1,\ldots,N$ 
\tcc{= 1 high-inflation month; = 0 otherwise}
 }
\tcc{initialization}
Flag[i] =0; $i=1,\ldots,N$\;
\For{$i = 1,\ldots, N-K$}{
    \If{$\log(CPI[i+K]/CPI[i])/(K * \Delta t) > $Cutoff}{
     \For{$j = 0,\ldots, K$}{
        Flag[i+j] = 1 \;
       }
 }
}
\caption{Pseudocode window inflation filter}\label{bootstrap_appendix}
\label{Algo:filter}
\end{algorithm}
\subsection{Effect of moving window size}\label{app:filter_window_size}
Figure \ref{moving_window_fig} shows the filtering results for windows of size 12, 60, and 120 months. We can see that the five-year window produces two obvious inflation regimes: 1940:8-1951:7 and 1968:9-1985:10, which correspond to well-known market shocks (i.e. the second world war, and price controls; the oil price shocks and stagflation of the seventies). Increasing the window 
size to 10 years results in similar-looking plots as the five-year window size, but the number of months in each window increases, and the average inflation rate is lower. Since our objective is to determine the effect of high-inflation periods on allocation strategies, we {choose the five-year window size.}

\begin{figure}[htb!]
\centerline{%
\begin{subfigure}[t]{.33\linewidth}
\centering
\includegraphics[width=\linewidth]{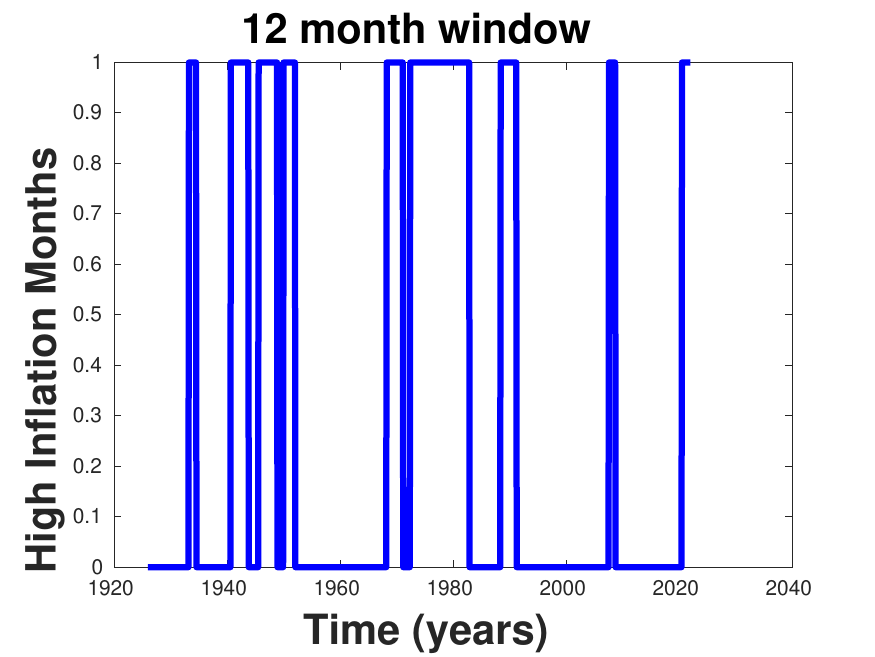}
\label{twelve_month_fig}
\end{subfigure}
\begin{subfigure}[t]{.33\linewidth}
\centering
\includegraphics[width=\linewidth]{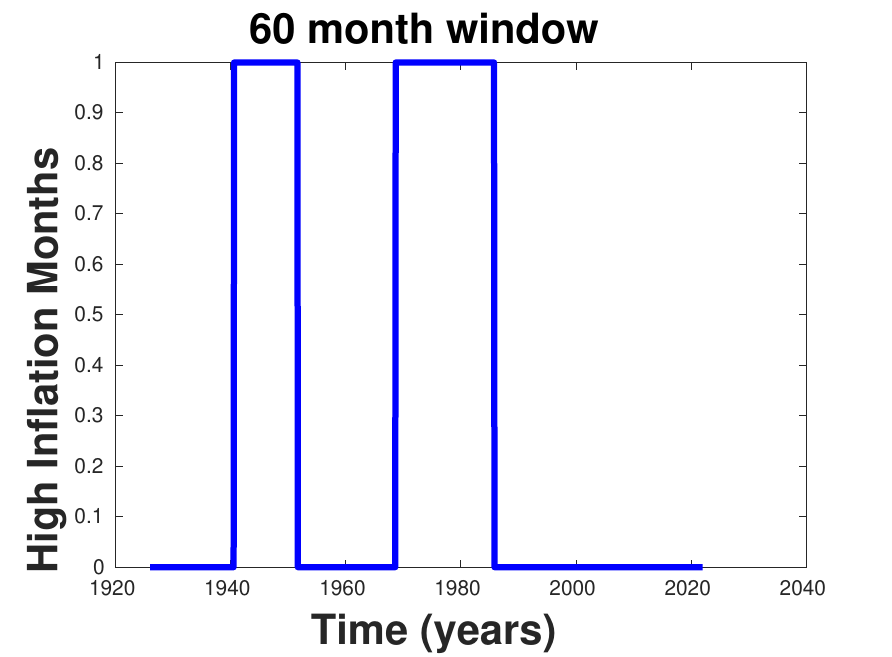}
\label{60_month_fig}
\end{subfigure}
\begin{subfigure}[t]{.33\linewidth}
\centering
\includegraphics[width=\linewidth]{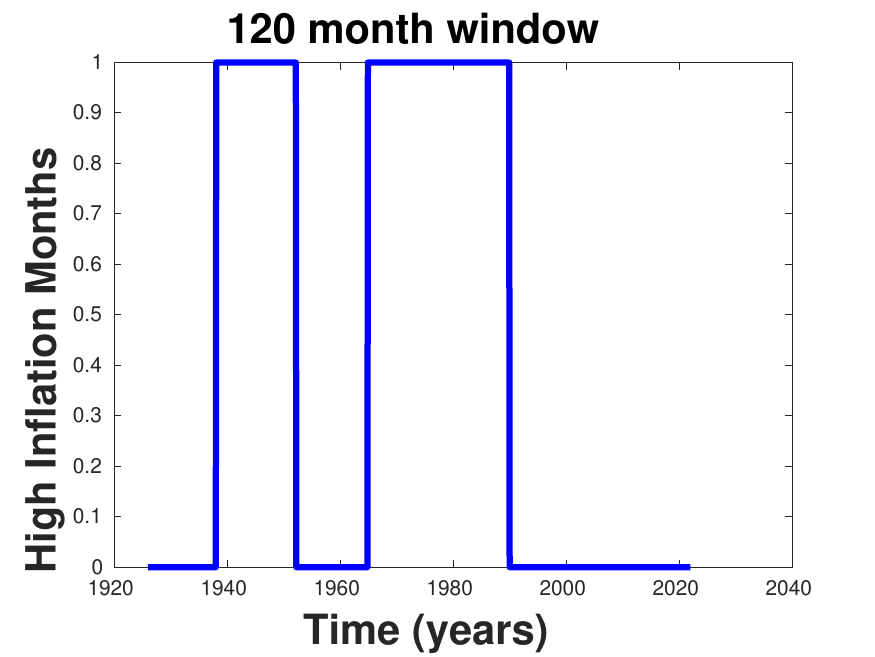}
\label{120_month_fig}
\end{subfigure}
}
\caption{high-inflation regimes, using the moving-window method, with the window
size shown.  The cutoff for {\em high-inflation} regimes was $0.05$.
High-inflation months have a label value of one, and low-inflation months
have a label value of zero.  CPI data identified from the historical period 1926:1-2022:1.
}
\label{moving_window_fig}
\end{figure}

\subsection{Asset performance during high inflation}\label{app:inflation_asset_return}

To gain some intuition on the behavior of asset returns during the inflation periods, we assume that each real (adjusted by CPI index) index follows geometric
Brownian motion (GBM).  For example, given an index with value $S$, then
\begin{eqnarray}
   dS & = & \mu S ~dt + \sigma S ~dZ ~
\end{eqnarray}
where $dZ$ is the increment of a Wiener process. We use
maximum likelihood estimation to fit the drift rate $\mu$  (expected arithmetic return) and volatility $\sigma$  in each regime, for each index, as shown in Table \ref{GBM_table}.  We also show a series constructed by: converting the indexes in each regime to returns, concatenating the two return series, and converting the concatenated return series back to an index. This concatenated
index does not, of course, correspond to an actual historical index, but is a pseudo-index constructed from high-inflation regimes. This amounts to a worst-case sequence of
returns in terms of the duration of historical inflation periods, that could plausibly be expected during a long period of high inflation. 

\begin{table}[hbt!]
\begin{center}
\begin{tabular}{lccc} \toprule
Index   & $\mu$  & $\sigma$ & $\mu - \sigma^2/2$\\
\hline
   & \multicolumn{3}{c}{1940:8-1951:7}\\
\hline
CapWt & 0.079 & 0.140 & .069\\
EqWt  & 0.145 &  0.190& .127 \\
10 Year Treasury & -0.035 &  0.036 & -.036\\
30-day T-bill    & -0.050 & 0.029 & -.050\\
\hline
    &   \multicolumn{3}{c}{1968:9-1985:10}\\
\hline
CapWt & 0.026     & 0.164& .013\\
EqWt  & 0.065     & 0.220 &  .041\\
10 Year Treasury & 0.011  &     0.093 & .007 \\
30-day T-bill    & 0.009   &  0.012 & .009\\
\hline
   & \multicolumn{3}{c}{Concatenated: 1940:8-1951:7 and 1968:9 - 1985:10}\\
\hline
CapWt &  0.049     & 0.156 & .038\\ 
EqWt  &  0.098    & 0.209 & .076\\
10 Year Treasury & -0.008 &     0.076& -.011\\
30-day T-bill    & -0.014  &   0.022& -.014\\
\bottomrule
\end{tabular}
\caption{GBM parameters for the indexes shown.  All indexes are real (deflated).
$\mu$ is the expected annualized arithmetic return.  $\sigma$ is the annualized volatility.
($\mu - \sigma^2/2$) is the annualized mean geometric return.
}
\label{GBM_table}
\end{center}
\end{table}

It is striking that in each historical inflation regime (i.e., 1940:8-1951:7 and 1968:9-1985:10) in Table \ref{GBM_table}, the drift rate $\mu$ 
for the equal-weighted index is much larger
than the drift rate for the cap-weighted index.  
We can observe that the {mean} geometric return for the cap-weighted index,
in the period 1968:9-1985:10, was only about one percent per year.

It is also noticeable that bonds
performed very poorly in the period 1940:8-1951:7. As well, during the period 1968:9-1985:10, there was
essentially no term premium for 10-year treasuries, compared with 30-day T-bills. In addition,
the 10-year treasury index had much higher volatility compared to the 30-day T-bill index. 
Looking at the concatenated series, it appears that 30-day T-bills are arguably the 
better defensive asset here since the volatility of this index is quite low (but with a
negative (real) drift rate).

\section{Bootstrap resampling}
\subsection{Stationary block bootstrap algorithm}\label{app:bootstrap}
Algorithm \ref{Algo:bootstrap} presents the pseudocode for the stationary block bootstrap. See \citet{ni2022optimal} for more discussion.

\begin{algorithm}[htb]
\SetAlgoLined
\tcc{initialization}
bootstrap\_samples = [ ]\;
\tcc{loop until the total number of required samples are reached}
\While{True }{
\tcc{choose random starting index in [1,\ldots,N], N is the index of the last historical sample}
index = UniformRandom( 1, N )\;
\tcc{actual blocksize follows a shifted geometric distribution with the expected value of exp\_block\_size}
blocksize = GeometricRandom( $\frac{1}{exp\_block\_size}$ )\;
\For{$i = 0;\ i < blocksize;\ i = i + 1$}{
\tcc{if the chosen block exceeds the range of the historical data array, do a circular bootstrap}
\eIf{index + i $>$ N}{
bootstrap\_samples.append( historical\_data[ index + i - N ] )\;
}{bootstrap\_samples.append( historical\_data[ index + i ] )\;}

\If{bootstrap\_samples.len() == number\_required}{\Return bootstrap\_samples\;} }
}
\caption{Pseudocode for stationary block bootstrap}
\label{Algo:bootstrap}
\end{algorithm}

\subsection{Effect of blocksize}\label{app:blk_effect}
As discussed, we 
will use bootstrap resampling 
\citep{politis1994stationary,
politis2004automatic,patton2009correction,dichtl2016,anarkulova2022stocks},
to analyze the performance of using
the equal-weighted index compared to the cap-weighted index, during
periods of high inflation (our concatenated series: 1940:8-1951:7, 1968:9-1985:10).

First, we examine the effect of the expected blocksize parameter in
the bootstrap resampling algorithm.  We will use a paired sampling approach,
where we simultaneously draw returns from the bond and stock indexes.\footnote{This preserves
correlation effects.}  The algorithm
in \cite{politis2004automatic} was developed for single asset time series. It is therefore important to assess the effect of the blocksize on numerical results. In Table  \ref{block_table_a}, we examine the effect of different blocksizes
on the statistics of stationary block bootstrap resampling.

\begin{table}[hbt!]
\begin{center}
\begin{tabular}{lcccc} \hline
Expected blocksize (months) &  Median[$W_T$] & E[$W_T$] &   std[$W_T$] &5th Percentile     \\ \hline
 1                          & 170.9        & 191.6    &  97.6    & 78.6  \\
 3                          & 174.6        & 202.9    & 120.4    & 69.3   \\
 6                          & 174.2        & 204.2    & 125.9    & 66.8  \\
 12                         & 175.5        & 204.4    & 124.2    & 67.9 \\
 24                         & 179.2        & 205.1    & 118.4    & 68.7 \\ \hline
\end{tabular}
\caption{Effect of expected blocksize, on the statistics of
the final wealth $W(T)$ at $T=10$ years. Constant weight, scenario in Table \ref{base_case_1}.
Equity weight: 0.7, rebalanced monthly.  Bond index: 30-day T-bill.
Equity index: equal-weighted.
Concatenated series: 1940:8-1951:7, 1968:9-1985:10 (high-inflation regimes).
All quantities are real (inflation-adjusted).  Initial wealth 100. 
Bootstrap resampling, $10,000$ resamples).
\label{block_table_a}
}
\end{center}
\end{table}

Perhaps a more visual way of analyzing the effect of the expected blocksize is
shown in Figure \ref{blk_effect},  where we show the cumulative distribution function (CDF) of the
final wealth after 10 years,
for different blocksizes.   We show the CDF since this gives us a visualization of the
entire final wealth distribution, not just a few summary statistics.

Since the data frequency is at one-month intervals,
specifying a geometric mean expected blocksize of one month means that the blocksize
is always a constant one month.  This effectively means that we are assuming that
the data is i.i.d.  However, the one-month results are an outlier, compared to the
other choices of expected blocksize.  There is hardly any difference between the CDFs
for any choice of expected blocksize in the range of 3-24 months.  In this article, we use an expected blocksize of 6 months.

\begin{figure}[htb!]
\centering
\includegraphics[width=3.0in]{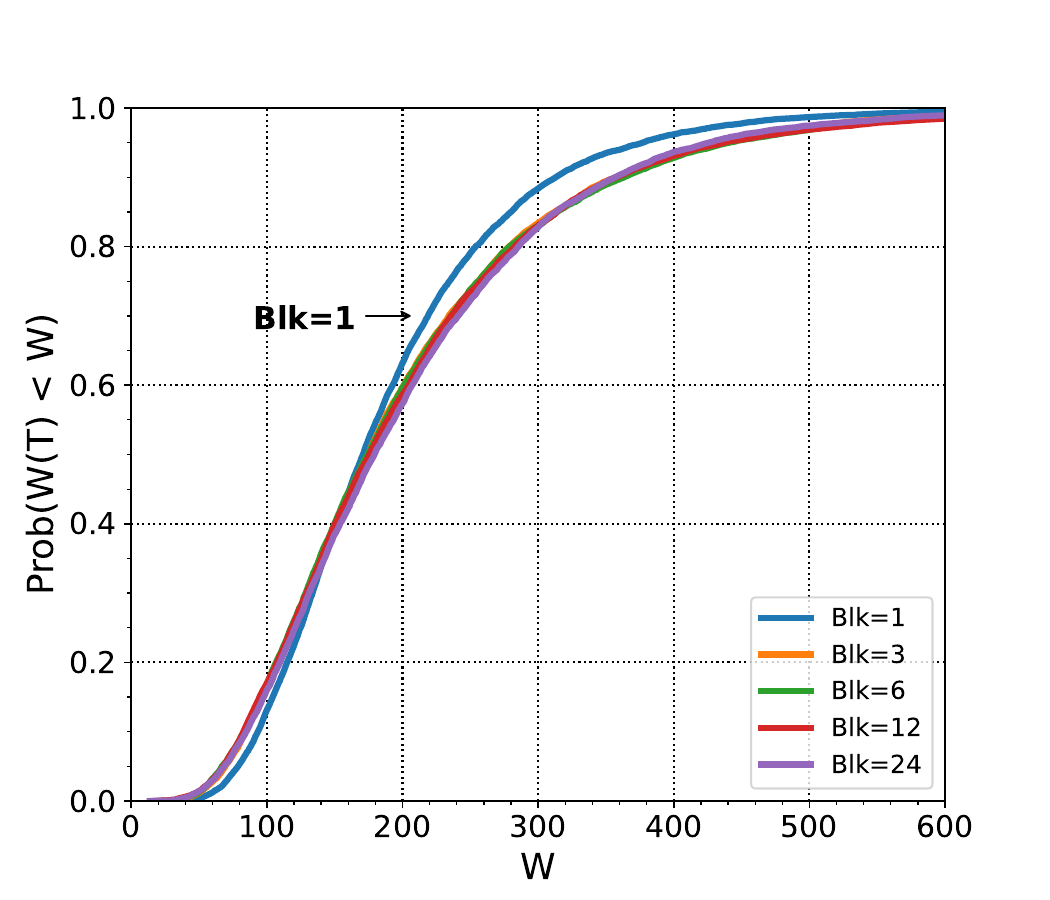}
\caption{
Cumulative distribution function (CDF), final
wealth $W(T)$ at $T=10$ years. The effect of expected blocksize.  Constant weight, scenario in Table \ref{base_case_1}.
Equity weight: 0.7, rebalanced monthly.  Bond index: 30-day T-bill.
Equity index: equal-weighted.
Concatenated series: 1940:8-1951:7, 1968:9-1985:10 (high-inflation regimes).
All quantities are real (inflation-adjusted).  Initial wealth 100.
Bootstrap resampling, expected blocksize one year, $10,000$ resamples.
}
\label{blk_effect}
\end{figure}

\subsection{Bootstrapping from non-contiguous data segments}\label{app:bootstrap_noncont}
As discussed in Section \ref{sec:high_inflation_bootstrap}, we have identified two historical inflation regimes: 1940:8-1951:7 and 1968:9-1985:10. As traditional bootstrap methods assume one segment of the underlying data segment, it naturally becomes a question of how to bootstrap from two non-contiguous segments of data appropriately. In the main sections of the article, we first concatenate the two data segments, then treat the concatenated data samples as a complete segment and apply bootstrap methods to them. This method is in line with the work of \citet{anarkulova2022stocks}, in which the authors concatenate stock returns from different countries and bootstrap from the concatenated series.

A second intuitive bootstrap method would be to bootstrap randomly from each of the two segments. Briefly, each bootstrap resample consists of (i) selecting a random segment (probability proportional to the length of the segment), (ii) selecting a random starting date in
the selected segment, (iii) then selecting a block (of random size) of
consecutive returns from this start date, (iv) in the event that the end of the data set in a segment is reached, we use circular block bootstrap resampling within that segment, and (v) repeating this process until a sample of the total desired length is obtained. 

We compare the bootstrapped data from concatenated segments and separate segments, by evaluating the performance of the 70\%/30\% equal-weighted index/T-bill fixed-mix portfolio, using the investment scenario described in Table \ref{base_case_1}. 
\begin{table}[htb]
\begin{center}
\begin{tabular}{lcccc} \hline
                                          &  Median[$W_T$] & E[$W_T$] &   std[$W_T$] &5th Percentile     \\ \hline
 Bootstrap from concatenated segments     & 174.2        & 204.2    & 125.9    & 66.8  \\
 Bootstrap from separate segments      & 176.9        & 208.0    & 132.4    & 65.4 \\ \hline
\end{tabular}
\caption{Effect of bootstrap method - bootstrap from concatenated segments vs bootstrap from separate segments, on the statistics of
the final wealth $W(T)$ at $T=10$ years. Constant weight, scenario in Table \ref{base_case_1}.
Equity weight: 0.7, rebalanced monthly.  Bond index: 30-day T-bill.
Equity index: equal-weighted.
Concatenated series: 1940:8-1951:7, 1968:9-1985:10 (high-inflation regimes).
All quantities are real (inflation-adjusted).  Initial wealth 100. 
Bootstrap resampling, $10,000$ resamples).
\label{tb:noncont}
}
\end{center}
\end{table}

We can observe from Table \ref{tb:noncont} that the strategy performance on bootstrap resampled data using two methods only varies slightly. This indicates that the two methods do not yield much difference for practical purposes. This is indeed expected - after all, the difference between the two methods only occurs
when a random block crosses the edge of each of the segments. 
However, such a situation only occurs with a very low probability. 
Except for this low-probability situation, the two bootstrap methods are identical.

\section{Comparing passive strategies in high inflation regimes}\label{sec:sto_dom}
In this section, we compare the performances of two fixed-mix strategies. The first strategy, the ``EqWt'' strategy, maintains a 70\% allocation to the equal-weighted index, and 30\% allocation to the 30-day T-bill index. The second strategy, the ``CapWt'' strategy, maintains a 70\% allocation to the cap-weighted index, and 30\% allocation to the 30-day T-bill index. 
\begin{table}[htb]
\begin{center}
\begin{tabular}{lc} \toprule
Investment horizon $T$ (years) & 10  \\
Equity market indexes & CRSP cap-weighted/equal-weighted index (real) \\
Bond index & 30-day T-bill (U.S.) (real) \\
Index Samples& Concatenated 1940:8-1951:7, 1968:9-1985:10\\
Initial portfolio wealth  & 100 \\
Rebalancing frequency & Monthly\\
\bottomrule
\end{tabular}
\caption{Investment scenario. 
\label{base_case_1}}
\end{center}
\end{table}


Figure \ref{plot_cum_high_inflation_eq_cap_fig} compares the CDF (cumulative distribution functions) of the terminal wealth of the EqWt strategy and the CapWt strategy
based on 
$10,000$ block bootstrap resampled data samples \citep{politis1991circular,dichtl2016,anarkulova2022stocks}
from the concatenated CRSP combined time series from 1940:8-1951:7 and 1968:9-1985:10, with an expected blocksize of six months. 
Both strategies
 assume  an initial wealth of 100 with no further cash injections and withdrawals,
  the investment horizon is 10 years, with monthly rebalancing to maintain the constant weights in the portfolio, see also Table \ref{base_case_1}.

We  first  recall the concept of 
{\it partial stochastic dominance}. Suppose two investment strategies $A$ and $B$ are evaluated 
on a set of data samples under the same investment scenario. We consider the
CDFs of terminal wealth $W$ associated with both strategies. Specifically, we denote the CDF of 
strategy A by $\text{CDF}_A(W)$ and that of strategy B by $\text{CDF}_B(W)$. 
Let $W_T$ be the random wealth at time $T$ and $W$ be a possible wealth realization, then we can interpret $\text{CDF}_A(W)$ as
\begin{eqnarray}
   \text{CDF}_A(W) &= &Prob(W_T \leq W) ~. \label{cdf_def}
\end{eqnarray}

Following \citet{ATK_1987,van2021distribution}, we define partial first-order stochastic dominance.

\begin{definition}[Partial first order stochastic dominance]\label{partial_def}
Given an investment strategy A which generates a CDF of terminal wealth $W$ given by $\text{CDF}_A(W)$, 
and a strategy B with
 $\text{CDF}_B(W)$,
then strategy $A$ partially stochastically dominates strategy B (to first order) {in the interval $(W_{lo},W_{hi})$} if
\begin{eqnarray}
    \text{CDF}_A(W) & \leq & \text{CDF}_B(W),\; \forall W \in (W_{lo}, W_{hi}) 
\end{eqnarray}
with strict inequality for at least one point in $(W_{lo}, W_{hi})$.
\end{definition}
The arguments for relaxing the usual definition of stochastic dominance are 
given in \citep{ATK_1987,van2021distribution}. Given some initial wealth $W_0$, if $W_{hi} \gg W_0$, then an investor may not be concerned that strategy $A$ underperforms strategy $B$ at these very high wealth values. In this case, the investor is fabulously wealthy. 
 Suppose that $W_{lo} \ll W_0$,
Assume $\text{CDF}_A(W_{lo}) = \text{CDF}_{B}( W_{lo})$.  As an extreme example, suppose $W_{lo} = $2 cents.  The fact
that strategy $B$ has a higher probability of ending up with one cent, compared with strategy $A$ is cold
comfort, and not particularly interesting.  On the other hand, suppose $\text{CDF}(W_{lo}) \ll 1$.  Again, an investor
may not be interested in events with exceptionally low probabilities.

\begin{figure}[htb!]
\centering
\includegraphics[width=3.0in]{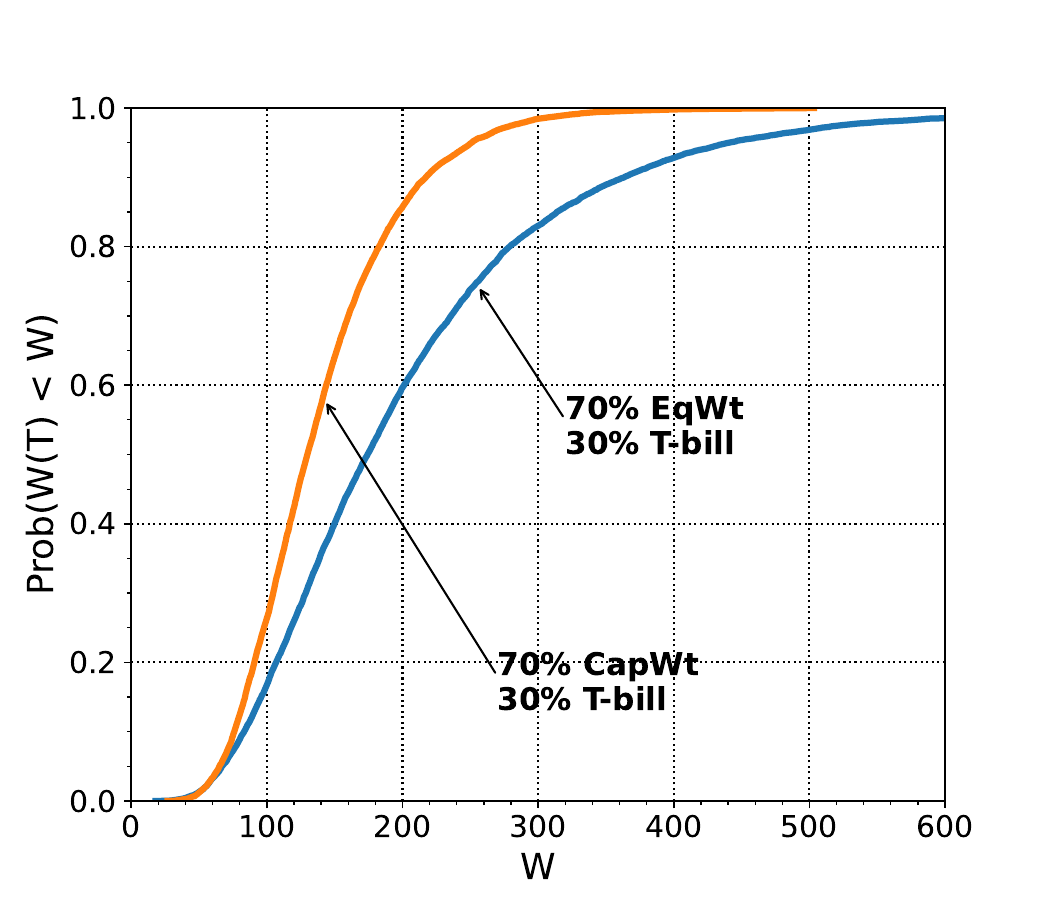}
\caption{Cumulative distribution function of final real wealth $W$ at $T=10$ years, bootstrap resampling
expected blocksize six months, $10,000$ resamples (Appendix \ref{bootstrap_appendix}).
$T = 10 $ years.
Data: concatenated returns, 1940:8-1951:7, 1968:9-1985:10.
Scenario described in Table \ref{base_case_1}.
}
\label{plot_cum_high_inflation_eq_cap_fig}
\end{figure}

Remarkably, the EqWt strategy appears to partially stochastically
dominate the CapWt strategy, since the CDF curve of the EqWt strategy 
almost appears to be entirely on the right side of the CDF curve of the CapWt strategy, except at very low probability
values, see  \citet{ATK_1987,van2021distribution} for a definition of stochastic partial dominance.
Close examination shows that the curves cross at the point $F_{EqWt}(W_{lo}) = F_{CapWt}(W_{lo}) \simeq .02$,
with a slight underperformance of the EqWt strategy compared to the CapWt strategy in this extreme left tail.

The fact that the EqWt strategy partially stochastically dominates the CapWt strategy seems to suggest 
that the equal-weighted stock index is the better choice for the stock index 
than the cap-weight stock index during high inflation times. 
We note that, using recent data\footnote{Since about 2010.  Of course, this is outside 
a period of sustained high inflation.},
the situation is not as clear \citep{taljaard2021has},
since the equal-weighted index appears to underperform.
However, \citet{taljaard2021has} suggests that this is due to the recent
market concentration in tech stocks.\footnote{As of February 2023, Apple, Microsoft, Amazon and
Alphabet (A and C) in total comprised 17\% of the market capitalization of the S\&P 500.} In fact, a plausible explanation for the outperformance (historically) of an equal-weighted index is that this is simply due to the small-cap effect, which was not widely known until about 1981 \citep{Banz_1980}.  
\citet{Plyakha2021} acknowledge that the equal-weighted index has significant exposure to the size factor.  However, \citet{Plyakha2021} argue that the equal-weighted index also has a larger exposure to the value factor.  In addition, there is a significant {\em alpha} effect due to the contrarian strategy of frequent rebalancing to equal weights. It would appear to be simplistic to dismiss an equal weight strategy on the grounds that this is simply a small cap effect that has become less effective.

\section{Calibrated synthetic model parameters}\label{app:calibrated_parans}
\begin{table}[H]
\centering
\begin{tabular}{ccccccccccccc}
\hline
$\mu_1$ & $\sigma_1$ & $\lambda_1$ & $\nu_1$ & $\iota_1$ & $\varsigma_1$ & 
    $\mu_2$ & $\sigma_2$ & $\lambda_2$ & $\nu_2$ & $\iota_2$ & $\varsigma_2$ & $\rho$ \\ \hline
0.051   & 0.146      & 0.178       & 0.2         & 7.13       & 7.33       & -0.014  & 0.017      & 0.321       & 0           & N/A       & 44.48      & 0.14   \\ \hline
\end{tabular}
\caption{Estimated annualized parameters for double exponential 
jump-diffusion model (\ref{def:double_exponential}) from CRSP cap-weighted stock index, 30-day U.S. T-bill index deflated by the CPI. Sample period: concatenated
1940:8-1951:7 and 1968:9-1985:10.}
\end{table}

\section{Comparison of CD and CS objectives}\label{sec:shorfall_obj}
In this section, we numerically compare the CS objective function (\ref{prob:CS_cont}) with the CD objective function (\ref{prob:CD_cont}).

As we briefly discussed in Section \ref{sec:obj_choice}, one caveat of the CD objective function is that it not only penalizes the underperformance relative to the elevated target but also penalizes the outperformance over the elevated target. In practice, the outperformance of the elevated target is favorable, and managers may not want to penalize the strategy when it happens. Therefore, in such cases, the cumulative quadratic shortfall (CS) objective (\ref{prob:CS_cont}) and (\ref{prob:CS_disc}) may be more appropriate. For the remainder of the paper, we focus on the discrete-time CS problem with the LFNN parameterization and the equally-spaced rebalancing schedule $T_{\Delta t}$ defined in Appendix \ref{sec:approx_form}, i.e.,
\begin{equation}{\text{(Parameterized $CS(\beta)$)}:\quad}
   \inf_{\boldsymbol{\theta}\in\mathbb{R}^{N_{\boldsymbol{\theta}}}}\mathbb{E}_{f(\cdot;\boldsymbol{\theta})}^{(t_0,w_0)}\Bigg[\sum_{t\in\mathcal{T}_{\Delta t}}\Big(\min\big(W_{\boldsymbol{\theta}}(t)-e^{\beta t}\hat{W}(t),0\big)\Big)^2+\epsilon W_{\boldsymbol{\theta}}(T)\Bigg],\label{prob:CS_disc_opt}
\end{equation}

The CS objective function in (\ref{prob:CS_disc_opt}) only penalizes the underperformance against the elevated target. Here $\epsilon W(T)$ is a regularization term. We remark that problem (\ref{prob:CS_disc_opt}) without the regularization term can be ill-posed. To see this, consider a case where $W_{\boldsymbol{\theta}}(t) \gg e^{\beta t} \hat{W}(t)$, for some $t\in[t_0,T]$. In this case, the future cumulative quadratic shortfall (on $[t,T]$) will almost surely be zero without the regularization term, 
{  so the control from thereon has no effect on the objective function under that scenario.}
We choose $\epsilon$ to be a small positive scalar.  As William Bernstein once said, ``if you have won the game stop playing.'' If one has accumulated as much wealth as Warren Buffet, then it does not matter what assets she invests in. The positive regularization factor of $\epsilon$ forces the strategy to put all wealth into less risky assets when the portfolio has already performed extremely well.

We design a numerical experiment to compare the CS objective with the symmetric CD objective in the following problem (\ref{prob:CD_disc_opt}) with the same LFNN parameterization and equally-spaced rebalancing schedule $\mathcal{T}_{\Delta t}$
\begin{equation}{\text{(Parameterized $CD(\beta)$)}:\quad}
   \inf_{\boldsymbol{\theta}\in\mathbb{R}^{N_{\boldsymbol{\theta}}}}\mathbb{E}_{f(\cdot;\boldsymbol{\theta})}^{(t_0,w_0)}\Bigg[\sum_{t\in\mathcal{T}_{\Delta t}}\big(W_{\boldsymbol{\theta}}(t)-e^{\beta t}\hat{W}(t)\big)^2\Bigg].\label{prob:CD_disc_opt}
\end{equation}

Specifically, we adopt the investment scenario in Table \ref{tb:validation_case} with $\beta=0.02$ and reuse the training and testing data sets simulated from the calibrated double exponential jump-diffusion model. The neural network strategies follow the trained LFNN models on the training data set $\boldsymbol{Y}$ for both the CS and CD objective. We then evaluated both strategies on the same testing data set $\boldsymbol{Y}^{test}$. 

We compare the {\it wealth ratio}, i.e., the wealth of the managed portfolio divided by the wealth of the benchmark portfolio, over time, for both strategies. The wealth ratio metric reflects how well the active strategy performs against the benchmark strategy along the investment horizon; a higher wealth ratio metric is better. Below we show the percentiles of the wealth ratio for both strategies evaluated on $\boldsymbol{Y}^{test}$.

\begin{figure}[H]
\begin{minipage}{1\textwidth}
\centering
    \includegraphics[width=3.2in]{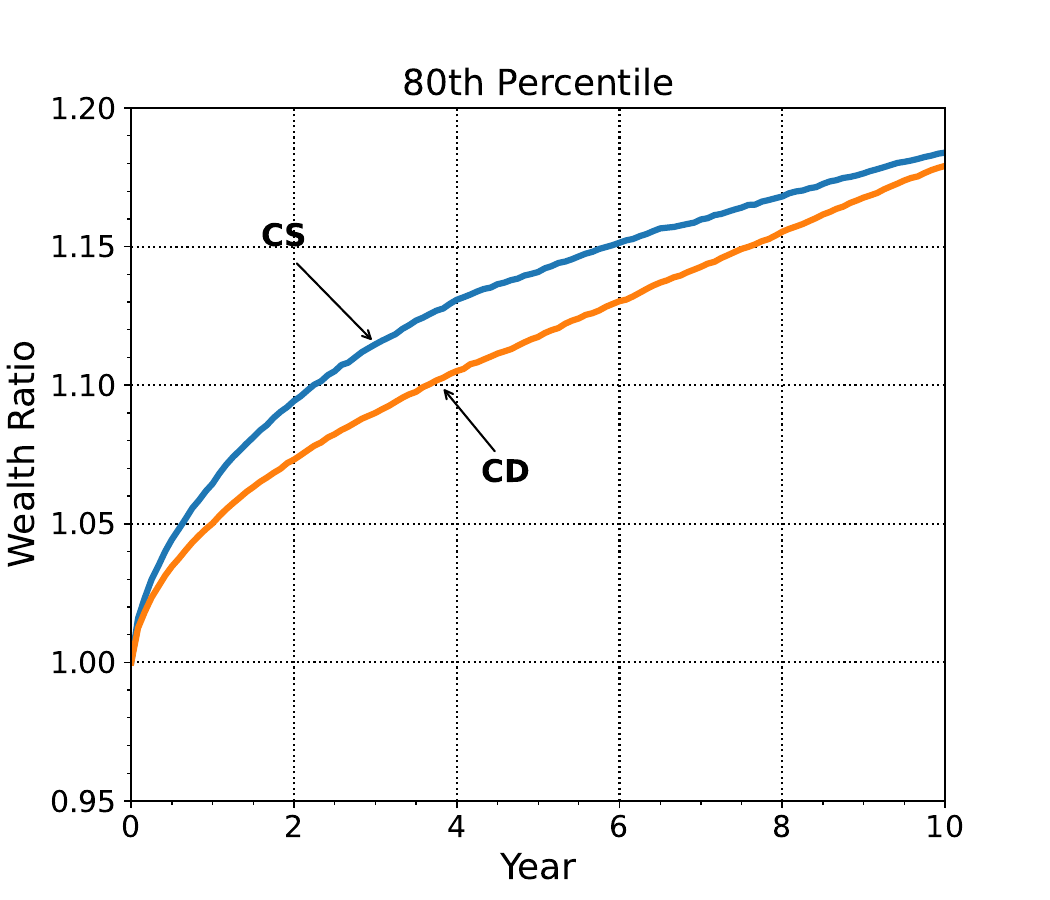}
    \includegraphics[width=3.2in]{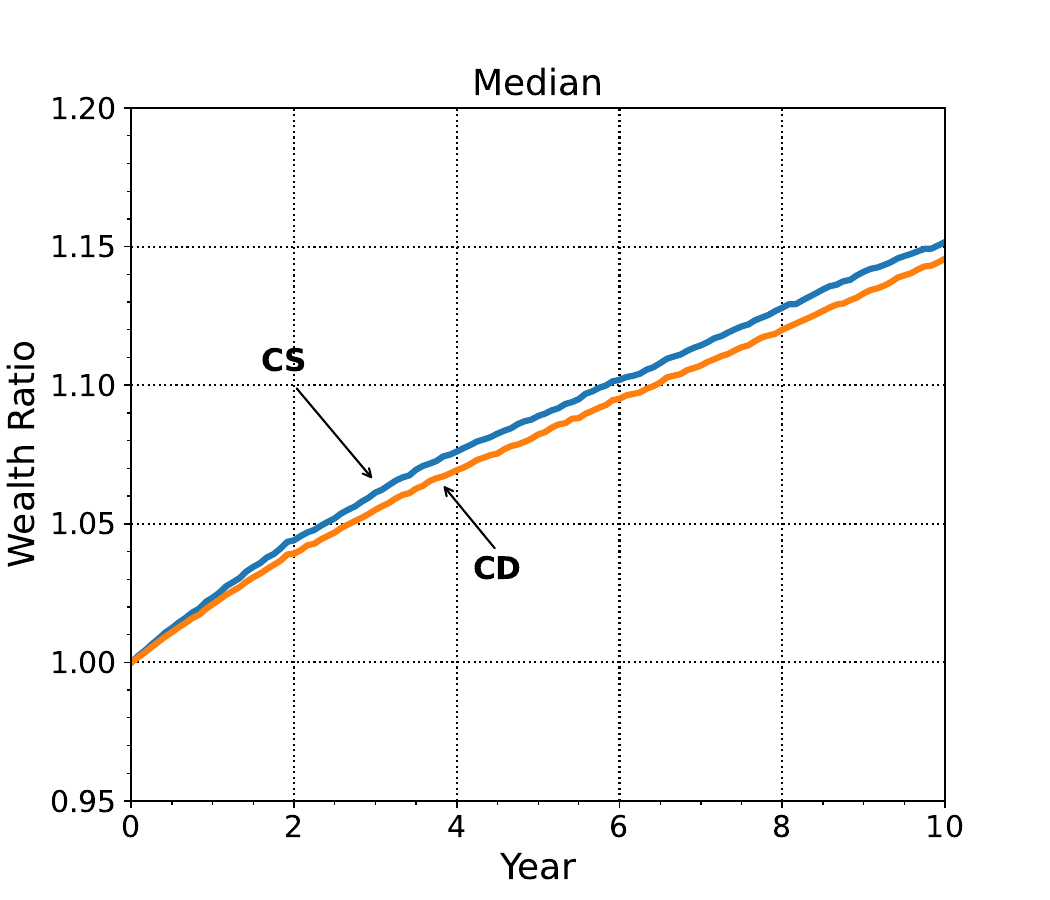}
    \vspace{0.2cm}
    \includegraphics[width=3.2in]{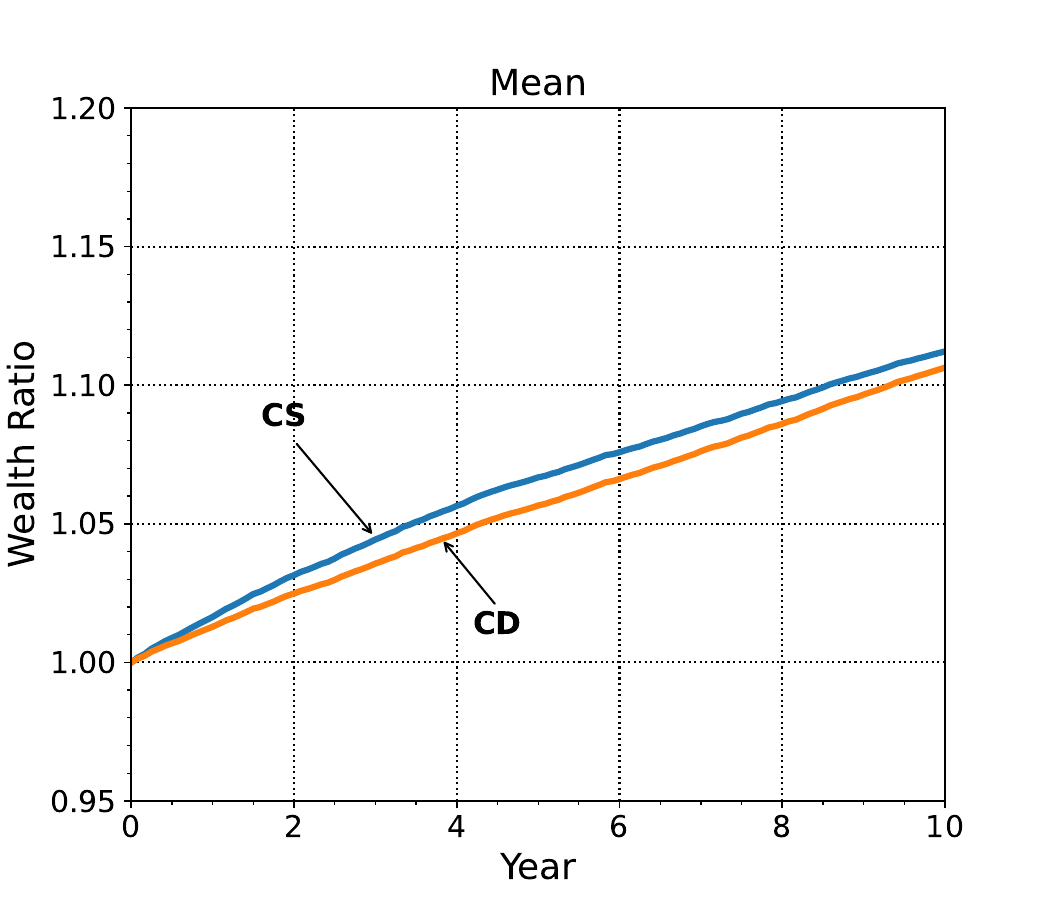}
    \includegraphics[width=3.2in]{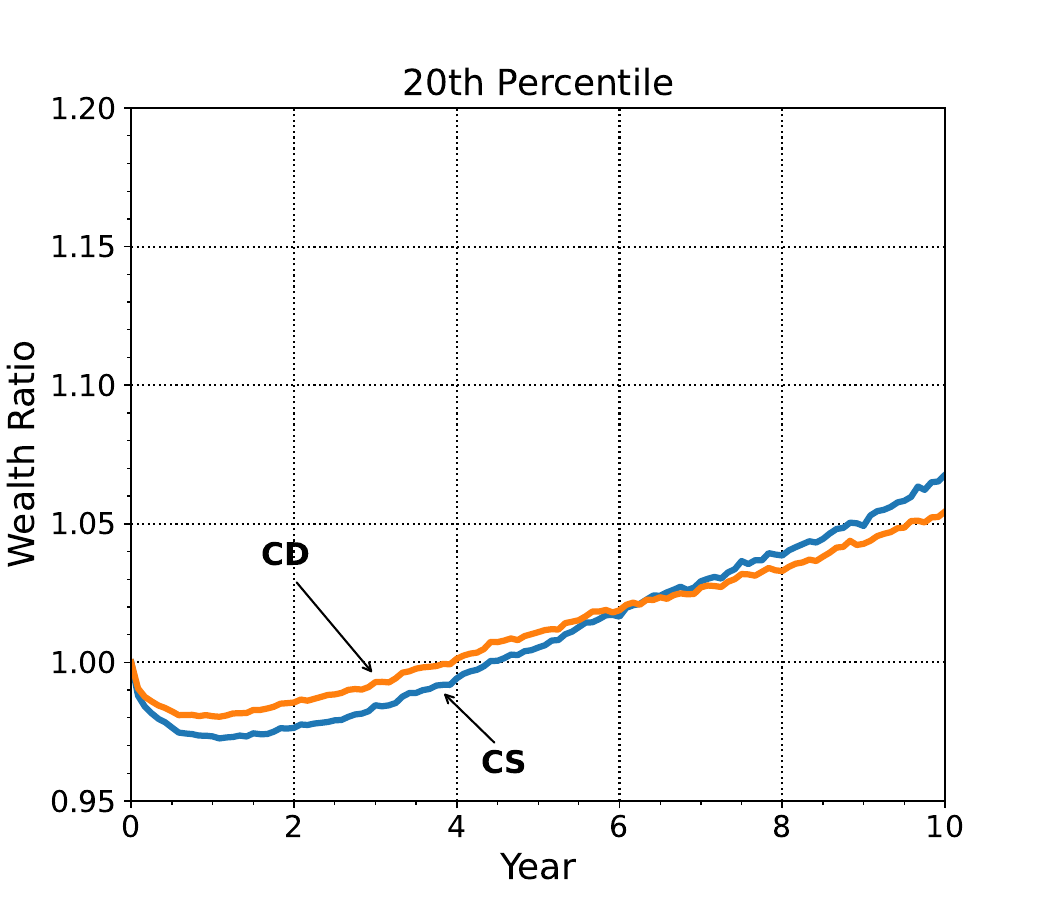}
\end{minipage}
    \caption{Percentiles of wealth ratio of the neural network strategy (i.e., the neural network model) learned under the cumulative quadratic tracking difference (CD) objective and the neural network strategy learned under the cumulative quadratic shortfall (CS) objective. The results shown are based on evaluations on the testing data set $\boldsymbol{Y}^{test}$.}\label{fig:obj_quad_vs_downside}
\end{figure}

We can see from Figure \ref{fig:obj_quad_vs_downside} that the CS strategy (neural network strategy trained under the CS objective) yields a more favorable wealth ratio than the CD strategy (neural network strategy trained under the CD objective). On average, the CS strategy achieves a consistently higher terminal wealth ratio than the CD strategy. Even in the 20th percentile case, the CS strategy lags initially but recovers over time.\footnote{The CS strategy starts with a higher allocation to the stock, and thus encounters more volatility early on.} The result indicates that the CS objective might be a wiser choice for managers in practice. In the following numerical experiments with bootstrap resampled data, we use the CS objective (\ref{prob:CS_disc_opt}) instead of the CD objective (\ref{prob:CD_disc_opt}).\footnote{Unfortunately, we cannot derive a closed-form solution under the CS objective (\ref{prob:CS_cont}) under continuous rebalancing using the PIDE approach, due to the non-smoothness introduced by the $\min(\cdot)$ function.}

\section{Experiments with non-zero borrowing premium}\label{sec:borrow_premium}
In Section \ref{sec:numerical_experiments}, we conducted the numerical experiments, assuming the borrowing premium is zero. This assumption is based on the fact that large sovereign wealth funds are often considered to have almost risk-free credit ratings, due to their state-backed nature. In other words, we assume that sovereign wealth funds can borrow funding at the same rate as risk-free treasury bills. 

This assumption may be too benign for general public funds. In general, it is unlikely that a non-sovereign wealth fund can borrow at a risk-free rate. However, the actual borrowing cost within large public funds is often 
unavailable. For this reason, we use the corporate bond yields issued by corporations with similar credit ratings as these large public funds as an approximation to the borrowing cost. 
Currently, large public funds such as the Blackstone Group or Apollo Global Management are rated between Aaa and Baa rating by Moody's. We obtain the nominal corporate bond yields with Moody's Aaa \citep{moodyaaa2023} and Baa \citep{moodybaa2023} ratings and adjust them with CPI returns. 
During the two high-inflation regimes we have identified, 
Aaa-rated corporate bonds have an average real yield of 0.7\%, 
while Baa-rated bonds have 1.8\%. Taking an average of the two, we use 1.25\% as an estimate for the real yield of corporate bonds as well as the borrowing cost of large public funds.\footnote{Note that the corporate bonds from Moody's yield data have maturities of more than 20 years. Usually, long-term bonds have higher yields than short-term bonds. 
Thus, using corporate yields likely overestimates the borrowing cost, since we assume the manager is only borrowing short-term funding.} As discussed in Section \ref{app:inflation_asset_return}, the average real return for T-bill index is -1.4\%. This gives us an average borrowing premium rate of $2.65\%$. 
In this section, as a stress test, we conduct the same experiment as in Section \ref{sec:numerical_experiment}, except that we use a fixed borrowing premium of $3\%$ instead of 0. We note that the historical corporate bond yields are based on bonds with a long maturity. Typically, long-term yields are higher than short-term yields, which accounts for the term risk. 
Therefore, the assumption of a $3\%$ borrowing premium 
should be a fairly aggressive stress test for the use of leverage.

\begin{figure}[H]
\centerline{%
\begin{subfigure}[t]{.45\linewidth}
\centering
\includegraphics[width=\linewidth]{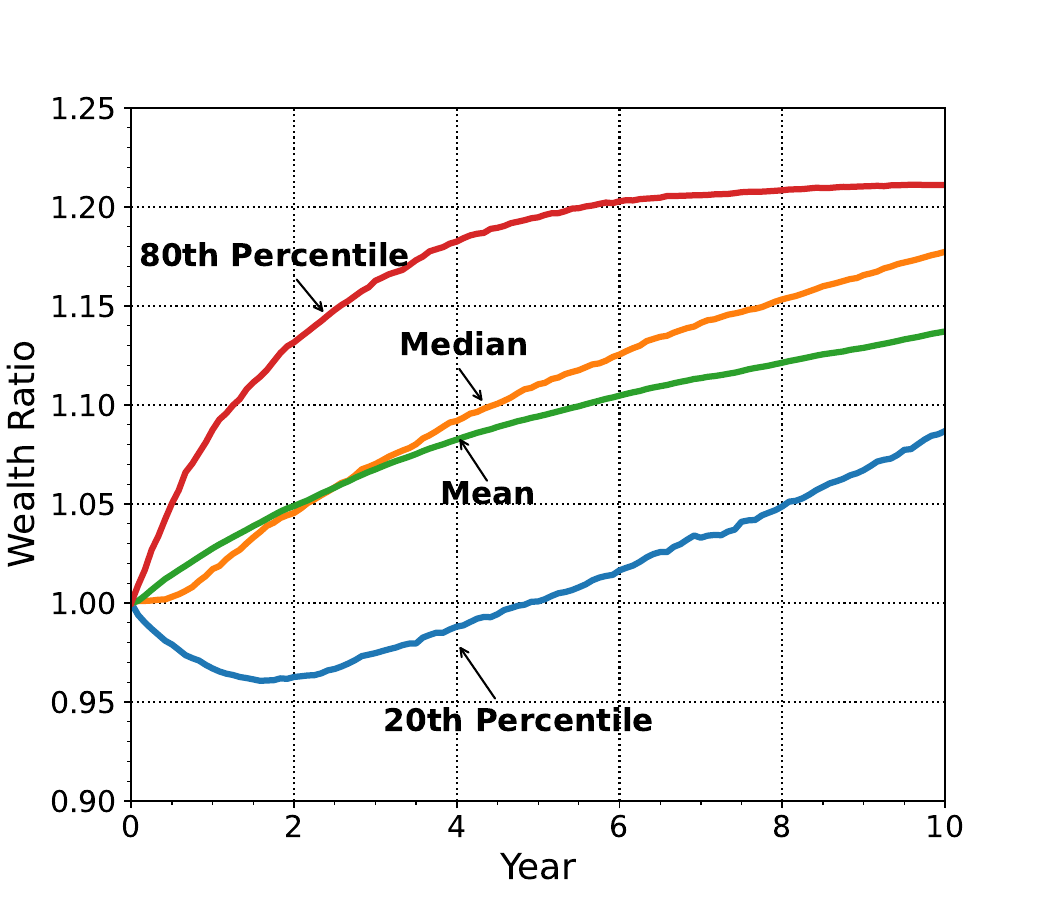}
\caption{Wealth ratio $\frac{W(t)}{\hat{W}(t)}$}
\label{fig:path_ratio_porem=3}
\end{subfigure}
\begin{subfigure}[t]{.45\linewidth}
\centering
\includegraphics[width=\linewidth]{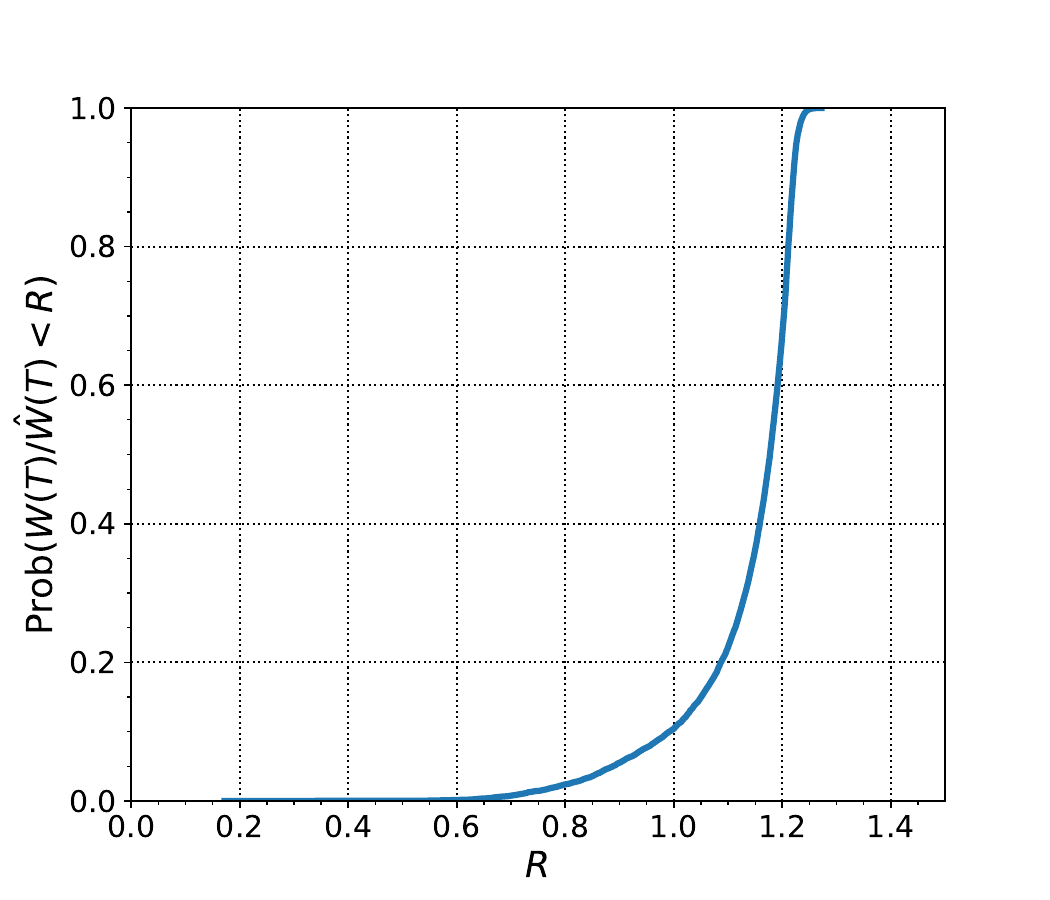}
\caption{Terminal wealth ratio $\frac{W(T)}{\hat{W}(T)}$}
\label{fig:cdf_irr_prem=3}
\end{subfigure}
}
\caption{
Percentiles of wealth ratio over the investment horizon, and CDF of terminal wealth ratio. Annualized borrowing premium is 3\%. Results are based on the evaluation of the learned neural network model (from high-inflation data) on the testing data set (low-inflation data).
}
\label{fig:prem=3}
\end{figure}

\begin{table}[H]
\begin{center}
\begin{tabular}{lccccc} \hline
Strategy &  Median[$W_T$] & E[$W_T$] &   std[$W_T$] & 5th Percentile  & Median IRR (annual)    \\ \hline
 Neural network & 362.7        & 401.3    &  212.6    & 133.9  & 0.078\\
 Benchmark     & 308.5        & 342.9     & 165.0    & 149.0 & 0.056 \\ \hline
 Neural network (zero premium) & 364.2     & 403.4    &  211.8    & 136.3  & 0.078\\\hline
\end{tabular}
\caption{Statistic of strategies. Annualized borrowing premium is 3\%. Results are based on the evaluation of the learned neural network model (from high-inflation data) on the testing data set (low-inflation data).}
\label{tb:prem=3}
\end{center}
\end{table}

We plot the percentiles of the wealth ratio and the CDF of the terminal wealth ratio. 
As we can see from Figure \ref{fig:prem=3} and Table \ref{tb:prem=3}, the neural network strategy is only 
marginally affected by the increased borrowing premium rate. 
Specifically, the terminal wealth statistics  for the case
with the borrowing premium all slightly worse. 
However, the impact is so marginal that the median 
IRR does not change, and the neural network strategy 
still maintains more than a 200 bps advantage in terms of median IRR
compared to the benchmark.

The most noticeable difference is in the allocation 
fraction, as shown in Figure \ref{fig:alloc_prem=3}. With a 
significantly higher borrowing cost, the neural network 
strategy does not leverage as much in the first two years, 
resulting in a less negative allocation to the T-bill and 
a lower allocation to the equal-weighted stock index. However, 
as we have seen in Figure \ref{fig:prem=3} and 
Table \ref{tb:prem=3}, this only results in minimal impact on 
the performance of the strategy. 

\begin{figure}[H]
\centering
\includegraphics[width=3.0in]{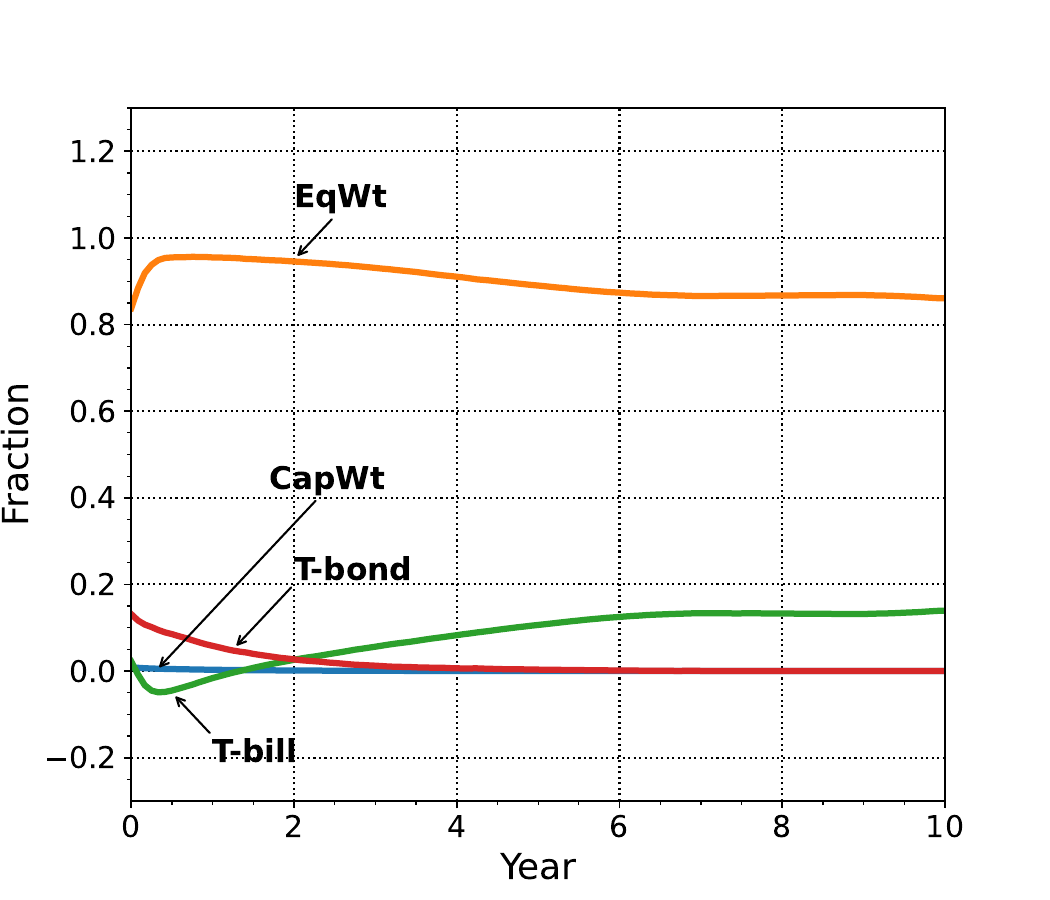}
\caption{Mean of allocation fraction of the learned neural network strategy over time, evaluated on the testing data set. Annualized borrowing premium is 3\%. The neural network strategy is learned from bootstrap resampled data based on concatenated series: 1940:8-1951:7, 1968:9-1985:10 (high-inflation regimes). The investment scenario is described in Table \ref{nn_main_case}.}\label{fig:alloc_prem=3}
\end{figure}

\section{Strategy performance in low inflation regimes}\label{sec:low_inflation}
Until now, the entire discussion has been centered around the imaginary scenario of a long, persistent inflation environment. We have shown that the neural network strategy consistently outperforms the benchmark strategy under a high-inflation regime. However, what if we are wrong? More dramatically, if the high-inflation situation ends immediately and we enter a low-inflation environment, how will the neural network strategy learned under high-inflation regimes perform? 

To answer these questions, we evaluate the neural network strategy learned under high-inflation regimes on a testing data set bootstrapped from low-inflation historical regimes. Specifically, we exclude the two inflation regimes (1940:8-1951:7 and 1968:9-1985:10) from the full historical data of 1926:1-2022:1 and obtain several low-inflation data segments. 
We concatenate the low-inflation data segments and use the stationary bootstrap (Appendix \ref{app:bootstrap}) to generate a testing data set. We adopt the investment scenario described in Table \ref{nn_main_case} and evaluate the performance of the neural network strategy obtained in Section \ref{sec:numerical_experiment} on this low-inflation data set. 

\begin{figure}[htb!]
\centering
\includegraphics[width=3.0in]{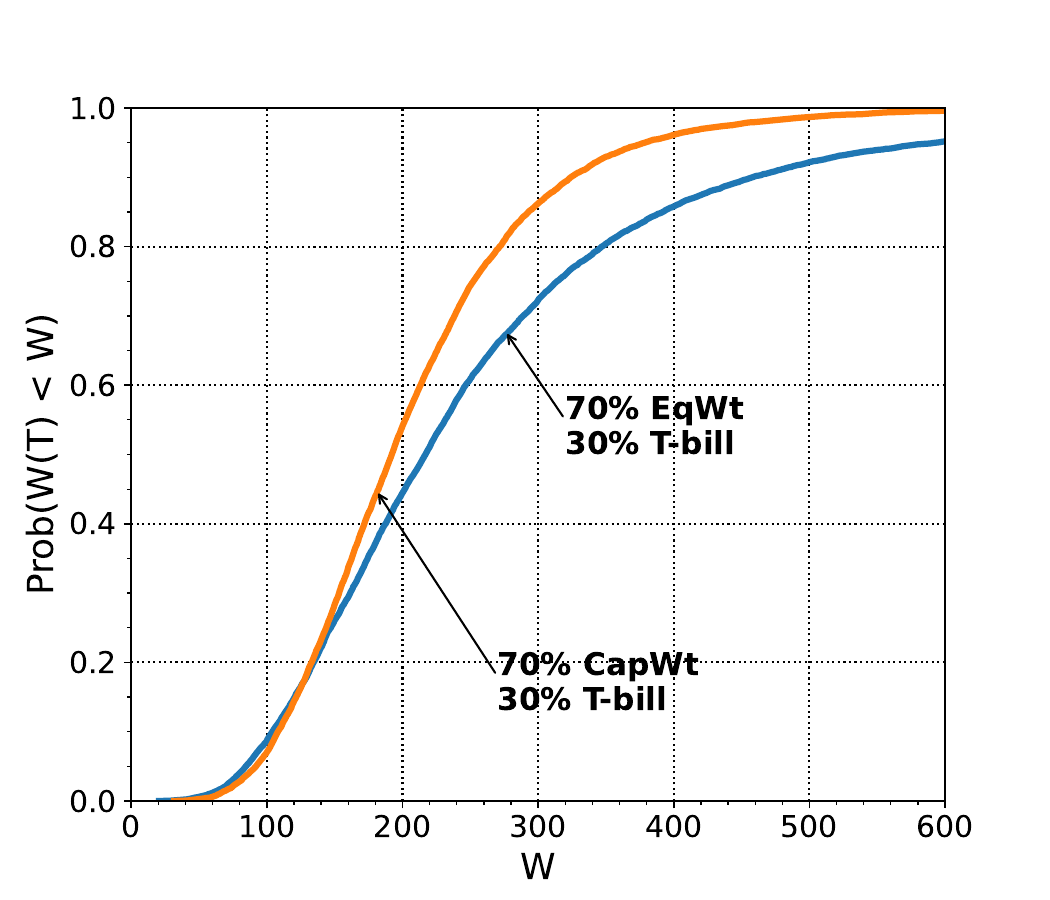}
\caption{
Cumulative distribution functions (CDFs) for
cap-weighted and equal-weighted indexes, as a function of final real wealth $W$ at $T=10$ years.
Initial stake $W_0 = 100$, no cash injections or withdrawals.  Block
bootstrap resampling, expected blocksize 6 months.  70\% stocks, 30\% bonds, rebalanced
monthly.  Bond index: 30-day U.S. T-bills.  Stock index: CRSP capitalization-weighted or CRSP equal-weighted index. Underlying data excludes high-inflation regimes.  All indexes are deflated
by the CPI. $10,000$ resamples. Data set 1926:1-2022:1, excluding high inflation regimes (1940:8-1951:7 and 1968:9-1985:10).
}
\label{fig:low_inf_fixedmix_CDF}
\end{figure}

Note that we continue to use the equal-weighted stock index/30-day T-bill 
fixed-mix portfolio as the benchmark. This is validated by Figure \ref{fig:low_inf_fixedmix_CDF}, which plots the CDF of the terminal wealth of the fixed-mix portfolios using 70\% equal-weighted stock index vs 70\% cap-weighted stock index (both with 30\% 30-day U.S. T-bill as the bond component). As we can see from Figure \ref{fig:low_inf_fixedmix_CDF}, the fixed-mix portfolio with an equal-weighted stock index clearly has a more right-skewed distribution than the portfolio with a cap-weighted stock index. This seems to suggest that the equal-weighted index is the superior choice to use in the benchmark portfolio, even in low inflation regimes.

\begin{figure}[htb]
\centerline{%
\begin{subfigure}[t]{.45\linewidth}
\centering
\includegraphics[width=\linewidth]{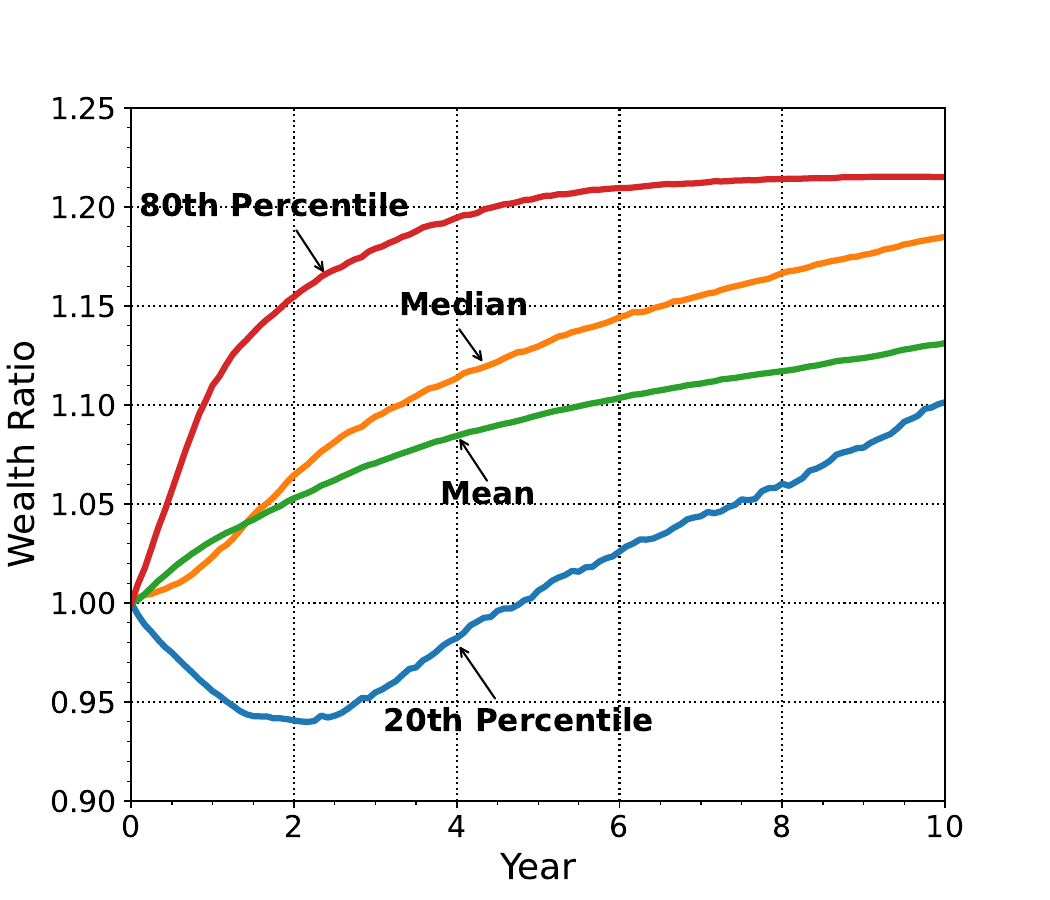}
\caption{Wealth ratio $\frac{W(t)}{\hat{W}(t)}$}
\label{fig:path_ratio_lowinf}
\end{subfigure}
\begin{subfigure}[t]{.45\linewidth}
\centering
\includegraphics[width=\linewidth]{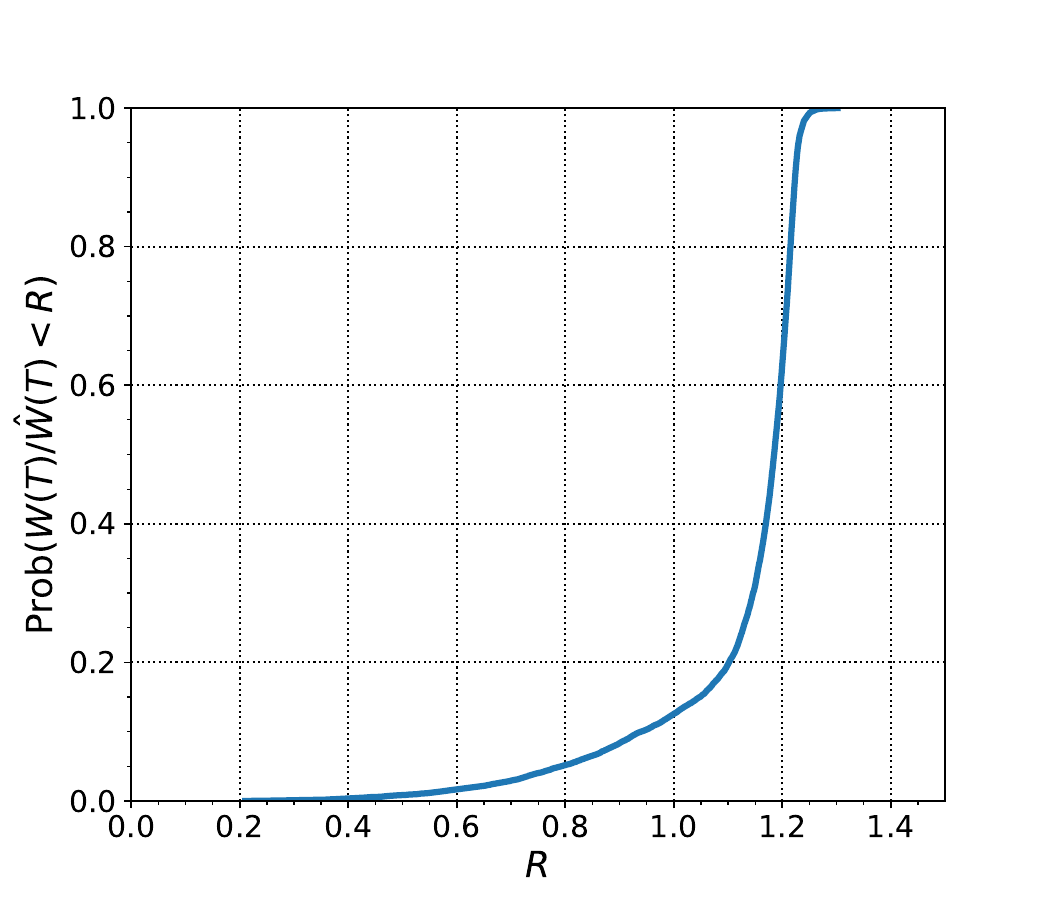}
\caption{Terminal wealth ratio $\frac{W(T)}{\hat{W}(T)}$}
\label{fig:cdf_irr_lowinf}
\end{subfigure}
}
\caption{
Percentiles of wealth ratio over the investment horizon, and CDF of terminal wealth ratio. Results are based on the evaluation of the learned neural network model (from high-inflation data) on the low-inflation testing data set).
}
\label{fig:lowinf}
\end{figure}

\begin{table}[hbt!]
\begin{center}
\begin{tabular}{lccccc} \hline
Strategy &  Median[$W_T$] & E[$W_T$] &   std[$W_T$] & 5th Percentile  & Median IRR (annual)    \\ \hline
 Neural network & 429.7        & 489.6    &  301.9    & 151.6  & 0.100\\
 Benchmark     & 368.3        & 420.8     & 238.2    & 175.7 & 0.079 \\ \hline
\end{tabular}
\caption{Statistic of strategies. Results are based on the evaluation of the learned neural network model (from high-inflation data) on the low-inflation testing data set.}
\label{tb:lowinf}
\end{center}
\end{table}

We then present the performance of the neural network strategy learned on high-inflation data on 
the testing data set bootstrapped from low-inflation historical returns. 
Surprisingly, as we can see from Figure \ref{fig:path_ratio_lowinf}, 
the performance of the neural network strategy learned under high-inflation regimes 
performs quite well in low-inflation environments. Compared to the testing results on the 
high-inflation data set, there is a noticeable performance degradation; 
for example, the probability of outperforming the benchmark strategy in terminal 
wealth is now slightly less than 90\%. However, the degradation is quite minimal. 
The neural network strategy still has more than an 85\% chance of outperforming 
the benchmark strategy at the end of the investment horizon. 
As shown in Table \ref{tb:lowinf}, the median IRR of the neural network 
strategy is still 2\% higher than the median IRR of the benchmark strategy, meeting the investment target.

The above results indicate that the neural network strategy is surprisingly robust. 
Despite being specifically trained under a high-inflation scenario, 
the strategy performs admirably well in a low-inflation environment.

\clearpage
\bibliographystyle{chicago}
\bibliography{references}

\end{document}